\newcommand{\fullornot}[2]{#1}
\newcommand{\citegs}{\cite{gs}}

\documentclass{article}

\usepackage[utf8]{inputenc}
\usepackage[inner=1cm,outer=2cm,top=.75cm,bottom=1.5cm]{geometry}
\usepackage{setspace}
\usepackage[rgb,table,usenames,dvipsnames]{xcolor}
\usepackage{verbatim}
\usepackage{subcaption}
\usepackage{boxedminipage}

\usepackage{fancyhdr}
\usepackage{graphicx}
\usepackage{fullpage}

\usepackage{hdb_macros}
\usepackage{csquotes}

\usepackage{algpseudocode}
\usepackage{thm-restate}
\usepackage{braket}

\usepackage{stmaryrd}

\newcommand{\basis}{\vec{B}}

    \usepackage{xcolor}
    \hypersetup{
    	colorlinks,
    	linkcolor={red!75!black},
    	citecolor={blue!75!black},
    	urlcolor={red!75!black},
    }
     \usepackage[hyperpageref]{backref}

\usepackage{esvect}

\renewcommand{\C}{\mathcal{C}}
\newcommand{\supp}{\mathsf{Supp}}
\newcommand{\brak}[1]{\llbracket #1 \rrbracket}
\newcommand{\suppvec}{\vv{\mathbf{Supp}}}

\newcommand{\mindist}{d_{\min}}

\newcommand{\gs}{\widetilde}

\title{\bf More basis reduction for linear codes: backward reduction, BKZ, slide reduction, and more}

\author{Surendra Ghentiyala\thanks{Cornell University.  \email{sg974@cornell.edu}. This work is supported in part by the NSF under Grants Nos.~CCF-2122230 and CCF-2312296, a Packard Foundation Fellowship, and a generous gift from Google.}
\and Noah Stephens-Davidowitz\thanks{Cornell University. \email{noahsd@gmail.com}. This work is supported in part by the NSF under Grants Nos.~CCF-2122230 and CCF-2312296, a Packard Foundation Fellowship, and a generous gift from Google. Some of this work was completed while the author was visiting the National University of Singapore and the Centre for Quantum Technologies}}

\date{}

\usepackage{tikz}
\usepackage{pgfplots}
\pgfplotsset{width=10cm,compat=1.10}
\usepgfplotslibrary{fillbetween}
\usetikzlibrary{patterns}
\usetikzlibrary{patterns.meta}
\allowdisplaybreaks

\begin{document}
\pagenumbering{roman}
\maketitle

\begin{abstract}
We expand on recent exciting work of Debris-Alazard, Ducas, and van Woerden [Transactions on Information Theory, 2022], which introduced the notion of \emph{basis reduction for codes}, in analogy with the extremely successful paradigm of basis reduction for lattices. 
We generalize DDvW's LLL algorithm and size-reduction algorithm from codes over $\F_2$ to codes over $\F_q$, and we further develop the theory of proper bases. We then show how to instantiate for codes the BKZ and slide-reduction algorithms, which are the two most important generalizations of the LLL algorithm for lattices. 

Perhaps most importantly, we show a new and very efficient basis-reduction algorithm for codes, called \emph{full backward reduction}. This algorithm is quite specific to codes and seems to have no analogue in the lattice setting. We prove that this algorithm finds vectors as short as LLL does in the worst case (i.e., within the Griesmer bound) and does so in less time. We also provide both heuristic and empirical evidence that it outperforms LLL in practice, and we give a variant of the algorithm that provably outperforms LLL (in some sense) for random codes.

Finally, we explore the promise and limitations of basis reduction for codes. In particular, we show upper and lower bounds on how ``good'' of a basis a code can have, and we show two additional illustrative algorithms that demonstrate some of the promise and the limitations of basis reduction for codes. \end{abstract}

\thispagestyle{empty}
\newpage

\tableofcontents

\newpage
\pagenumbering{arabic}

\section{Introduction}
\label{sec:intro}
\subsection{Codes and lattices}

A \emph{linear code} $\C \subseteq \F_q^n$ is a subspace of the vector space $\F_q^n$ over the finite field $\F_q$, i.e., it is the set of all $\F_q$-linear combinations of a linearly independent list of vectors $\basis := (\vec{b}_1;\ldots ; \vec{b}_k)$, 
\[
    \C := \C(\basis) := \{z_1 \vec{b}_1 + \cdots + z_k \vec{b}_k \ : \ z_i \in \F_q \}
    \; .
\]
We call $\vec{b}_1,\ldots, \vec{b}_k$ a \emph{basis} for the code and $k$ the \emph{dimension} of the code. We are interested in the geometry of the code induced by the Hamming weight $|\vec{c}|$ for $\vec{c} \in \F_q^n$, which is simply the number of coordinates of $\vec{c}$ that are non-zero. For example,  it is natural to ask about a code's \emph{minimum distance}, which is the minimal Hamming weight of a non-zero codeword, i.e.,
\[
    \mindist(\C) := \min_{\vec{c} \in \C_{\neq \vec0}} |\vec{c}|
    \; .
\]

At a high level, there are two fundamental computational problems associated with linear codes. In the first, the goal is to find a short non-zero codeword---i.e., given a basis for a code $\C$, the goal is to find a non-zero codeword $\vec{c} \in \C_{\neq \vec0}$ with relatively small Hamming weight $|\vec{c}|$. In the second, the goal is to find a codeword that is close to some target vector $\vec{t} \in \F_q^n$---i.e., given a basis for a code $\C$ and a target vector $\vec{t} \in \F_q^n$, the goal is to find a codeword $\vec{c} \in \C$ such that $|\vec{c} - \vec{t}|$ is relatively small. (Here, we are being deliberately vague about what we mean by ``relatively small.'')

We now describe another class of mathematical objects, which are also ubiquitous in computer science. Notice the striking similarity between the two descriptions.

A \emph{lattice} $\lat \subset \Q^n$ is the set of all \emph{integer} linear combinations of linearly independent basis vectors $\vec{A} := (\vec{a}_1;\ldots; \vec{a}_k) \in \Q^n$, i.e.,
\[
    \lat := \lat(\vec{A}) := \{z_1 \vec{a}_1 + \cdots + z_k \vec{a}_k \ : \ z_i \in \Z\}
    \; .
\]
We are interested in the geometry of the lattice induced by the Euclidean norm $\|\vec{a}\| := (a_1^2 + \cdots + a_n^2)^{1/2}$. In particular, it is natural to ask about a lattice's \emph{minimum distance}, which is simply the minimal Euclidean norm of a non-zero lattice vector, i.e.,
\[
    \lambda_1(\lat) := \min_{\vec{y} \in \lat_{\neq \vec0}} \|\vec{y}\|
    \; .
\]

At a high level, there are two fundamental computational problems associated with lattices. In the first, the goal is to find a short non-zero lattice vector---i.e., given a basis for a lattice $\lat$, the goal is to find a non-zero lattice vector $\vec{y} \in \lat_{\neq \vec0}$ with relatively small Euclidean norm $\|\vec{y}\|$. In the second, the goal is to find a lattice vector that is close to some target vector $\vec{t} \in \Q^n$---i.e., given a basis for a lattice $\lat$ and a target vector $\vec{t} \in \Z^n$, the goal is to find a lattice vector $\vec{y} \in \lat$ such that $\|\vec{y} - \vec{t}\|$ is relatively small. (Again, we are being deliberately vague here.)

Clearly, there is a strong analogy between linear codes and lattices, where to move from codes to lattices, one more-or-less just replaces a finite field $\F_q$ with the integers $\Z$ and the Hamming weight $|\cdot|$ with the Euclidean norm $\|\cdot\|$. It is therefore no surprise that this analogy extends to many applications. For example, lattices and codes are both used for decoding noisy channels. They are both used for cryptography (see, e.g.,~\cite{mceliecePublickeyCryptosystemBased1978,ajtaiGeneratingHardInstances1996,ajtaiPublicKeyCryptosystemWorstCase1997,hoffsteinNTRURingbasedPublic1998,alekhnovichMoreAverageCase2003,regevLatticesLearningErrors2009}; in fact, both are used specifically for \emph{post-quantum} cryptography). And, many complexity-theoretic hardness results were proven simultaneously or nearly simultaneously for coding problems and for lattice problems, often with similar or even identical techniques.\footnote{For example, similar results that came in separate works in the two settings include \cite{boasAnotherNPcompleteProblem1981} and \cite{berlekampInherentIntractabilityCertain1978}, \cite{ajtaiShortestVectorProblem1998} and \cite{VarAlgorithmicComplexityCoding1997}, \cite{Mic01svp} and \cite{dumerHardnessApproximatingMinimum2003}, \cite{bennettQuantitativeHardnessCVP2017,ASGapETHHardness2018} and \cite{stephens-davidowitzSETHhardnessCodingProblems2019}, etc. Works that simultaneously published the same results in the two settings include~\cite{aroraHardnessApproximateOptima1997} and~\cite{DKRSApproximatingCVPAlmostpolynomial2003}.}

\subsection{Basis reduction for lattices}

However, the analogy between lattices and codes has been much less fruitful for algorithms. Of course, there are many algorithmic techniques for finding short or close codewords and many algorithmic techniques for finding short or close lattice vectors. But, in many parameter regimes of interest, the best algorithms for lattices are quite different from the best algorithms for codes. 

In the present work, we are interested in \emph{basis reduction}, a ubiquitous algorithmic framework in the lattice literature. 
At a very high level, given a basis $\vec{A} := (\vec{a}_1;\ldots; \vec{a}_k)$ for a lattice $\lat$, basis reduction algorithms work by attempting to find a ``good'' basis of $\lat$ (and in particular, a basis whose first vector $\vec{a}_1$ is short) by repeatedly making ``local changes'' to the basis. Specifically, such algorithms manipulate the Gram-Schmidt vectors $\gs{\vec{a}}_1 := \vec{a}_1,\gs{\vec{a}}_2 := \pi_{\{\vec{a}_1\}}^\perp(\vec{a}_2), \ldots, \gs{\vec{a}}_k := \pi_{\{\vec{a}_1,\ldots, \vec{a}_{k-1}\}}^\perp(\vec{a}_k)$. Here, $\pi^\perp_{\{\vec{a}_1,\ldots, \vec{a}_{i-1}\}}$ represents orthogonal projection onto the subspace orthogonal to $\vec{a}_1,\ldots, \vec{a}_{i-1}$. Notice that we can view the Gram-Schmidt vector $\gs{\vec{a}}_i$ as a lattice vector in the lower-dimensional lattice generated by the \emph{projected block} $\vec{A}_{[i,j]} := \pi^\perp_{\{\vec{a}_1,\ldots, \vec{a}_{i-1}\}}(\vec{a}_i;\ldots; \vec{a}_j)$. Basis reduction algorithms work by making many ``local changes'' to $\vec{A}$, i.e., changes to the block $\vec{A}_{[i,j]}$ that leave the lattice $\lat(\vec{A}_{[i,j]})$ unchanged. The goal is to use such local changes to make earlier Gram-Schmidt vectors shorter. (In particular, $\gs{\vec{a}}_1 = \vec{a}_1$ is a non-zero lattice vector. So, if we can make the first Gram-Schmidt vector short, then we will have found a short non-zero lattice vector.) One accomplishes this, e.g., by finding a short non-zero vector $\vec{y}$ in $\lat(\vec{A}_{[i,j]})$ and essentially replacing the first vector in the block with this vector $\vec{y}$. (Here, we are ignoring how exactly one does this replacement.)

This paradigm was introduced in the celebrated work of Lenstra, Lenstra, and Lov{\'a}sz~\cite{lll82}, which described the polynomial-time LLL algorithm. Specifically, (ignoring important details that are not relevant to the present work) the LLL algorithm works by repeatedly replacing Gram-Schmidt vectors $\gs{\vec{a}}_i$ with a shortest non-zero vector in the lattice generated by the dimension-two block $\vec{A}_{[i,i+1]}$. The LLL algorithm itself has innumerable applications. (See, e.g.,~\cite{nguyenLLLAlgorithmSurvey2010}.) Furthermore, generalizations of LLL yield the most efficient algorithms for finding short non-zero lattice vectors in a wide range of parameter regimes, including those relevant to cryptography.

Specifically, the Block-Korkine-Zolotarev basis-reduction algorithm (BKZ), originally due to Schnorr~\cite{schnorrHierarchyPolynomialTime1987}, is a generalization of the LLL algorithm that works with larger blocks. It works by repeatedly modifying blocks $\vec{A}_{[i,i+\beta-1]}$ of a lattice basis $\vec{A} := (\vec{a}_1;\ldots;\vec{a}_k)$ in order to ensure that the Gram-Schmidt vector $\gs{\vec{a}}_i$ is a shortest non-zero vector in the lattice generated by the block. Here, the parameter $\beta \geq 2$ is called the \emph{block size}, and the case $\beta = 2$ corresponds to the LLL algorithm (ignoring some technical details). Larger block size $\beta$ yields better bases consisting of shorter lattice vectors. But, to run the algorithm with block size $\beta$, we must find shortest non-zero vectors in a $\beta$-dimensional lattice, which requires running time $2^{O(\beta)}$ with the best-known algorithms~\cite{ajtaiSieveAlgorithmShortest2001,conf/stoc/AggarwalDRS15,BDGLNewDirectionsNearest2016}. So, BKZ yields a tradeoff between the quality of the output basis and the running time of the algorithm. (Alternatively, one can view BKZ as a reduction from an approximate lattice problem in high dimensions to an exact lattice problem in lower dimensions, with the approximation factor depending on how much lower the resulting dimension is.)

BKZ is the fastest known algorithm in practice for the problems relevant to cryptography. However, BKZ is notoriously difficult to understand. Indeed, we still do not have a proof that the BKZ algorithm makes at most polynomially many calls to its $\beta$-dimensional oracle, nor do we have a tight bound on the quality of the bases output by BKZ, despite much effort. (See, e.g., \cite{walterLatticeBlogReduction}. However, for both the running time and the output quality of the basis, we now have a very good \emph{heuristic} understanding~\cite{CNBKZBetterLattice2011,BSWMeasuringSimulatingExploiting2018}.)

Gama and Nguyen's slide-reduction algorithm is an elegant alternative to BKZ that is far easier to analyze~\cite{gamaFindingShortLattice2008}. In particular, it outputs a basis whose quality (e.g., the length of the first vector) essentially matches our heuristic understanding of the behavior of BKZ, and it provably does so with polynomially many calls to a $\beta$-dimensional oracle for finding a shortest non-zero lattice vector. Indeed, for a wide range of parameters (including those relevant to cryptography), \cite{gamaFindingShortLattice2008} yields the fastest algorithm with proven correctness for finding short non-zero lattice vectors.

\paragraph{Dual reduction, and some foreshadowing.} One of the key ideas used in Gama and Nguyen's slide-reduction algorithm (as well as in other work, such as~\cite{micciancioPracticalPredictableLattice2016}) is the notion of a \emph{dual-reduced block} $\vec{A}_{[i,j]}$. The motivation behind dual-reduced blocks starts with the observation that the product $\|\gs{\vec{a}}_{i}\| \cdots \|\gs{\vec{a}}_{j}\|$ does not change if the lattice $\lat(\vec{A}_{[i,j]})$ is not changed. Formally, this quantity is the determinant of the lattice $\lat(\vec{A}_{[i,j]})$, which is a lattice invariant. So, while it is perhaps more natural to think of basis reduction in terms of making earlier Gram-Schmidt vectors in a block shorter, with the ultimate goal of making $\vec{a}_1$ short, one can more-or-less equivalently think of basis reduction in terms of making \emph{later} Gram-Schmidt vectors \emph{longer}.

One therefore defines a \emph{dual-reduced} block as a block $\vec{A}_{[i,j]}$ such that the last Gram-Schmidt vector $\gs{\vec{a}}_{j}$ is as \emph{long} as it can be without changing the associated lattice $\lat(\vec{A}_{[i,j]})$. When $\beta := j-i+1> 2$, a dual-reduced block is not the same as a block whose first Gram-Schmidt vector is as short as possible. However, there is still some pleasing symmetry here. In particular, it is not hard to show that the last Gram-Schmidt vector $\gs{\vec{a}}_j$ corresponds precisely to a non-zero (primitive) vector in the \emph{dual} lattice of $\lat(\vec{A}_{[i,j]})$ with length $1/\|\gs{\vec{a}}_{j}\|$. This of course explains the terminology. It also means that making the last Gram-Schmidt vector $\gs{\vec{a}}_{j}$ as long as possible corresponds to finding a shortest non-zero vector in the dual of $\lat(\vec{A}_{[i,j]})$, while making the first Gram-Schmidt vector $\gs{\vec{a}}_i$ as short as possible of course corresponds to finding a shortest non-zero vector in $\lat(\vec{A}_{[i,j]})$ itself. Either way, this amounts to finding a shortest non-zero vector in a $\beta$-dimensional lattice, which takes time that is exponential in the block size $\beta$.

\subsection{Basis reduction for codes!}

As far as the authors are aware, until very recently there was no work attempting to use the ideas from basis reduction in the setting of linear codes. This changed with the recent exciting work of Debris-Alazard, Ducas, and van Woerden, who in particular showed a simple and elegant analogue of the LLL algorithm for codes~\cite{DDvAlgorithmicReductionTheory2022}.

Debris-Alazard, Ducas, and van Woerden provide a ``dictionary'' (\cite[Table 1]{DDvAlgorithmicReductionTheory2022}) for translating important concepts in basis reduction from the setting of lattices to the setting of codes, and it is the starting point of our work. Below, we describe some of the dictionary from~\cite{DDvAlgorithmicReductionTheory2022}, as well as some of the barriers that one encounters when attempting to make basis reduction work for codes.

\subsubsection{Projection, epipodal vectors, and proper bases}

Recall that when one performs basis reduction on lattices, one works with the Gram-Schmidt vectors $\gs{\vec{a}}_i := \pi_{\{\vec{a}_1,\ldots, \vec{a}_{i-1}\}}^\perp(\vec{a}_i)$ and the projected blocks $\vec{A}_{[i,j]} := \pi^\perp_{\{\vec{a}_1,\ldots, \vec{a}_{i-1}\}}(\vec{a}_i;\ldots; \vec{a}_j)$, i.e., the \emph{orthogonal projection} of $\vec{a}_i,\ldots,\vec{a}_j$ onto the subspace $\{\vec{a}_1,\ldots, \vec{a}_{i-1}\}^\perp$ orthogonal to $\vec{a}_1,\ldots, \vec{a}_{i-1}$.

So, if we wish to adapt basis reduction to the setting of linear codes (and we do!), it is natural to first ask what the analogue of \emph{projection} is in the setting of codes. \cite{DDvAlgorithmicReductionTheory2022} gave a very nice answer to this question.\footnote{\cite{DDvAlgorithmicReductionTheory2022} formally worked with the case $q = 2$ everywhere. Rather than specialize our discussion here to $\F_2$, we will largely ignore this distinction in this part of the introduction. While the more general definitions that we provide here for arbitrary $\F_q$ are new to the present work, when generalizing to $\F_q$ is straightforward, we will not emphasize this in the introduction.} In particular, for a vector $\vec{x} = (x_1,\ldots, x_n) \in \F_q^n$ we call the set of indices $i$ such that $x_i$ is non-zero the \emph{support} of $\vec{x}$, i.e.,
\[
    \supp(\vec{x}) := \{i \in [n] \ : \ x_i \neq 0\}
    \; .
\]
Then,~\cite{DDvAlgorithmicReductionTheory2022} define $\vec{z} := \pi_{\{\vec{x}_1,\ldots, \vec{x}_\ell\}}^\perp(\vec{y})$ as follows. If $i \in \bigcup_j \supp(\vec{x}_j)$, then $z_i = 0$. Otherwise, $z_i = y_i$. In other words, the projection simply ``zeros out the coordinates in the supports of the $\vec{x}_j$.'' This notion of projection shares many (but certainly not all) of the features of orthogonal projection in $\R^n$, e.g., it is a linear contraction mapping (though not a strict contraction) that is idempotent.

Armed with this notion of projection,  \cite{DDvAlgorithmicReductionTheory2022} then defined the \emph{epipodal vectors} associated with a basis $\vec{b}_1,\ldots, \vec{b}_n$ as $\vec{b}_1^+ := \vec{b}_1, \vec{b}_2^+ := \pi_{\{\vec{b}_1\}}^\perp(\vec{b}_2), \ldots, \vec{b}_n^+ := \pi_{\{\vec{b}_1,\ldots, \vec{b}_{n-1}\}}^\perp(\vec{b}_n)$, in analogy with the Gram-Schmidt vectors. In this work, we go a bit further and define 
\[
    \vec{B}_{[i,j]} := \pi_{\{\vec{b}_1,\ldots, \vec{b}_{i-1}\}}^\perp(\vec{b}_{i};\ldots; \vec{b}_{j})
    \; ,
\]
in analogy with the notation in the literature on \emph{lattice} basis reduction.

Here, \cite{DDvAlgorithmicReductionTheory2022} already encounters a bit of a roadblock. Namely, the epipodal vectors $\vec{b}_i^+$ can be zero! E.g., if $\vec{b}_1 = (1,1,\ldots,1)$ is the all-ones vector, then $\vec{b}_i^+$ will be zero for all $i > 1$!\footnote{Of course, similar issues do not occur over $\R^n$, because if $\vec{a}_1,\ldots, \vec{a}_k \in \R^n$ are linearly independent, then $\pi_{\vec{a}_1,\ldots, \vec{a}_{k-1}}^\perp(\vec{a}_k)$ cannot be zero.} This is rather troublesome and could lead to many issues down the road. For example, we might even encounter entire blocks $\vec{B}_{[i,j]}$ that are zero! Fortunately, \cite{DDvAlgorithmicReductionTheory2022} shows how to get around this issue by defining \emph{proper} bases, which are simply bases for which all the epipodal vectors are non-zero. They then observe that proper bases exist and are easy to compute. (In \cref{sec:proper}, we further develop the theory of proper bases.) So, this particular roadblock is manageable, but it already illustrates that the analogy between projection over $\F_q^n$ and projection over $\R^n$ is rather brittle.

The LLL algorithm for codes then follows elegantly from these definitions. In particular, a basis $\basis = (\vec{b}_1;\ldots;\vec{b}_k)$ is \emph{LLL-reduced} if it is proper and if $\vec{b}_i^+$ is a shortest non-zero codeword in the dimension-two code generated by $\basis_{[i,i+1]}$ for all $i = 1,\ldots, k-1$. \cite{DDvAlgorithmicReductionTheory2022} then show a simple algorithm that computes an LLL-reduced basis in polynomial time. Specifically, the algorithm repeatedly makes simple local changes to any block $\basis_{[i,i+1]}$ for which this condition is not satisfied until the basis is reduced. 

In some ways, this new coding-theoretic algorithm is even more natural and elegant than the original LLL algorithm for lattices. For example, the original LLL algorithm had to worry about numerical blowup of the basis entries. And, the original LLL algorithm seems to require an additional slack factor $\delta$ in order to avoid the situation in which the algorithm makes a large number of minuscule changes to the basis. Both of these issues do not arise over finite fields, where all vectors considered by the algorithm have entries in $\F_q$ and integer lengths between $1$ and $n$.

\subsubsection{What's a good basis and what is it good for?}

Given the incredible importance of the LLL algorithm for lattices, it is a major achievement just to show that one can make sense of the notion of ``LLL for codes.'' But, once \cite{DDvAlgorithmicReductionTheory2022} have defined an LLL-reduced basis for codes and shown how to compute one efficiently, an obvious next question emerges: what can one do with such a basis? 

In the case of lattices, the LLL algorithm is useful for many things, but primarily for the two most important computational lattice problems: finding short non-zero lattice vectors and finding close lattice vectors to a target. In particular, the first vector $\vec{a}_1$ of an LLL-reduced basis is guaranteed to $\|\vec{a}_1\| \leq 2^{k-1} \lambda_1(\lat)$. This has proven to be incredibly useful, despite the apparently large approximation factor.

For codes over $\F_2$, \cite{DDvAlgorithmicReductionTheory2022} show that the same is true, namely that if $\basis = (\vec{b}_1;\ldots; \vec{b}_k)$ is an LLL-reduced basis for $\C \subseteq \F_2^n$, then $|\vec{b}_1| \leq 2^{k-1} \mindist(\C)$. 
They prove this by showing that if $\vec{b}_i^+$ has minimal length among the non-zero codewords in $\C(\basis_{[i,i+1]})$, then $|\vec{b}_i^+| \leq 2 |\vec{b}_{i+1}^+|$. It follows in particular that $|\vec{b}_1| = |\vec{b}_1^+| \leq 2^{i-1} |\vec{b}_i^+|$ for all $i$. One can easily see that $\mindist(\C) \geq \min_i |\vec{b}_i^+|$, from which one immediately concludes that $|\vec{b}_1^+| \leq 2^{k-1} \mindist(\C)$. A simple generalization of this argument shows that over $\F_q$, we have $|\vec{b}_1^+| \leq q^{k-1} \mindist(\C)$. (We prove something more general and slightly stronger in \cref{subsec: bkz_approx_factor}.)

However, notice that all codewords have length at most $n$ and $\mindist(\C)$ is always at least $1$. Therefore, an approximation factor of $q^{k-1}$ is non-trivial only if $n > q^{k-1}$. Otherwise, literally any non-zero codeword has length less than $q^{k-1} \mindist(\C)$! On the other hand, if $n > q^{k-1}$, then we can anyway find an exact shortest vector in time roughly $q^k n \lesssim n^2$ by simply enumerating all codewords. (The typical parameter regime of interest is when $n = O(k)$.)

In some sense, the issue here is that the space $\F_q^n$ that codes live in is too ``cramped.'' While lattices are infinite sets that live in a space $\Q^n$ with arbitrarily long and arbitrarily short non-zero vectors, codes are finite sets that live in a space $\F_q^n$ in which all non-zero vectors have integer lengths between $1$ and $n$. So, while for lattices, any approximation factor between one and, say, $2^{k}$ is very interesting, for codes the region of interest is simply more narrow.

\cite{DDvAlgorithmicReductionTheory2022} go on to observe that because $|\vec{b}_{i+1}^+|$ is an integer, for codes over $\F_2$ an LLL-reduced basis must actually satisfy
\[
    |\vec{b}_{i+1}^+| \geq \left\lceil \frac{|\vec{b}_i^+|}{2} \right\rceil
    \; .
\]
With this slightly stronger inequality together with the fact that $\sum |\vec{b}_i^+| \leq n$, they are able to show that $\vec{b}_1$ of an LLL-reduced basis will meet the Griesmer bound~\cite{GriBoundErrorcorrectingCodes1960},
\begin{equation}
    \label{eq:griesmer}
    \sum_{i=1}^k \left\lceil \frac{|\vec{b}_1|}{2^{i-1}} \right\rceil \leq n
    \; ,
\end{equation}
which is non-trivial. E.g., as long as $k \geq \log n$, it follows that 
\begin{equation}
    \label{eq:griesmer_special_case}
    |\vec{b}_1| - \frac{\ceil{\log |\vec{b}_1|}}{2} \leq \frac{n-k}{2} + 1
\end{equation}
(We generalize this in \cref{sec:LLL} to show that the appropriate generalization of LLL-reduced bases to codes over $\F_q$ also meet the $q$-ary Griesmer bound.)

\paragraph{Finding close codewords and size reduction.} For lattices, Babai also showed how to use an LLL-reduced basis to efficiently find \emph{close} lattice vectors to a given target vector~\cite{babaiLovaszLatticeReduction1986}, and like the LLL algorithm itself, Babai's algorithm too has innumerable applications. More generally, Babai's algorithm tends to obtain closer lattice vectors if given a ``better'' basis, in a certain precise sense. \cite{DDvAlgorithmicReductionTheory2022} showed an analogous ``size-reduction'' algorithm that finds close codewords to a given target vector, with better performance given a ``better'' basis. Here, the notion of ``better'' is a bit subtle, but essentially a basis is ``better'' if the epipodal vectors tend to have similar lengths. (Notice that $\sum_i |\vec{b}_i^+| = |\supp(\C)|$, so we cannot hope for all of the epipodal vectors to be short.) 

The resulting size-reduction algorithm finds relatively close codewords remarkably quickly. (Indeed, in nearly linear time.) Furthermore, \cite{DDvAlgorithmicReductionTheory2022} showed how to use their size-reduction algorithm combined with techniques from information set decoding to speed up some information set decoding algorithms for finding short codewords or close codewords to a target. For this, their key observation was the fact that typically most epipodal vectors actually have length one (particularly the later epipodal vectors, as one would expect given that later epipodal vectors by definition have more coordinates ``zeroed out'' by projection orthogonal to the earlier basis vectors) and that their size-reduction algorithm derives most of its advantage from how it treats the epipodal vectors with length greater than one. They therefore essentially run information set decoding on the code projected onto the support of the epipodal vectors with length one and then ``lift'' the result to a close codeword using their size-reduction algorithm. 

They call the resulting algorithm Lee-Brickell-Babai because it is a hybrid of Babai-style size reduction and the Lee-Brickell algorithm~\cite{LBObservationSecurityMcEliece1988}. The running time of this hybrid algorithm is dominated by the cost of running information set decoding on a code with dimension $k-k_1$, where 
\[
    k_1 := |\{i \ : \ |\vec{b}_i^+| > 1\}|
\]
is the number of epipodal vectors that do not have length $1$. 
Indeed, the (heuristic) running time of this algorithm is better than Lee-Brickell by a factor that is exponential in $k_1$, so that even a small difference in $k_1$ can make a large difference in the running time.
They then show that LLL-reduced bases have $k_1 \gtrsim \log n$ (for random codes) and show that this reduction in dimension can offer significant savings in the running time of information set decoding.

Indeed, though the details are not yet public, the current record in the coding problem challenges~\cite{DecodingChallenge} was obtained by Ducas and Stevens, apparently using such techniques.

\subsection{Our contributions}

In this work, we continue the study of basis reduction for codes, expanding on and generalizing the results of \cite{DDvAlgorithmicReductionTheory2022} in many ways, and beginning to uncover a rich landscape of algorithms.

\subsubsection{Expanding on the work of \texorpdfstring{\cite{DDvAlgorithmicReductionTheory2022}}{DDvW22}}

\paragraph{Generalization to $\F_q$.}  Our first set of (perhaps relatively minor) contributions are generalizations of many of the ideas in \cite{DDvAlgorithmicReductionTheory2022} from $\F_2$ to $\F_q$, a direction proposed in that work. In fact, they quite accurately anticipated this direction. So, we quote directly from \cite[Section 1.3]{DDvAlgorithmicReductionTheory2022}:
\begin{displayquote}
    In principle, the definitions, theorems and algorithms of this article should be
generalizable to codes over $\F_q$ endowed with the Hamming metric\ldots\ 
 Some algorithms
may see their complexity grow by a factor $\Theta(q)$, meaning that the algorithms remains polynomial-time only for $q = n^{O(1)}$. It is natural to hope that such a generalised LLL would still match [the]
Griesmer bound for $q > 2$. However, we expect that the analysis of the fundamental
domain [which is necessary for understanding size reduction]\ldots\ would become significantly harder to carry out.
\end{displayquote}

In \cref{sec: fq}, we generalize from $\F_2$ to $\F_q$ the definitions of projection, epipodal vectors, and proper bases; the definition of an LLL-reduced bases and the LLL algorithm;\footnote{We actually describe the LLL algorithm as a special case of the more general algorithms that we describe below. See \fullornot{\cref{sec:LLL}}{the full version \citegs}.} and the size-reduction algorithm and its associated fundamental domain. Some of this is admittedly quite straightforward---e.g., given the definitions in \cite{DDvAlgorithmicReductionTheory2022} of projection, epipodal vectors, and proper bases for codes over $\F_2$, the corresponding definitions for codes over $\F_q$ are immediate (and we have already presented them in this introduction). And, the definition of an LLL-reduced basis and of size reduction follow more-or-less immediately from these definitions. In particular, we do in fact confirm that LLL over $\F_q$ achieves the Griesmer bound.

As \cite{DDvAlgorithmicReductionTheory2022} anticipated, the most difficult challenge that we encounter here is in the analysis of the fundamental domain that one obtains when one runs size reduction with a particular basis $\basis$. We refer the reader to \fullornot{\cref{sec:size_reduction,subsec: F_analysis}}{the full version \citegs} for the details.

(We do not encounter the running time issue described in the quote above---except for our algorithm computing the number of vectors of a given length in $\mathcal{F}(\basis^+)$. In particular, our versions of the LLL algorithm and the size-reduction algorithm---and even our generalizations like slide reduction---run in time that is proportional to a small polynomial in $\log q$.)

Along the way, we make some modest improvements to the work of \cite{DDvAlgorithmicReductionTheory2022}, even in the case of $\F_2$. In particular, using more careful analysis, we shave a factor of roughly $n$ from the proven running time of LLL.  (See \fullornot{\cref{sec:LLL}}{the full version \citegs}.)

\paragraph{More on the theory of proper bases. } In order to develop the basis-reduction algorithms that we will describe next, we found that it was first necessary to develop (in \cref{sec:proper}) some additional tools for understanding and working with \emph{proper bases}, which might be of independent interest. Specifically, we define the concept of a \emph{primitive codeword}, which is a non-zero codeword $\vec{c}$ such that $\supp(\vec{c})$ does not strictly contain the support of any other non-zero codeword. We then show that primitive codewords are closely related to proper bases. For example, we show that $\vec{c}$ is the first vector in some proper basis if and only if $\vec{c}$ is primitive, and that a basis is proper if and only if the epipodal vectors are primitive vectors in their respective projections. 

We find this perspective to be quite useful for thinking about proper bases and basis reduction in general. In particular, we use this perspective to develop algorithms that perform basic operations on proper bases, such as inserting a primitive codeword into the first entry of a basis without affecting properness. The resulting algorithmic tools seems to be necessary for the larger-block-size versions of basis reduction that we describe below, in which our algorithms must make more complicated changes to a basis.

\subsubsection{Backward reduction and redundant sets}

Our next contribution is the notion of \emph{backward reduction}, described in \cref{sec:redundant}. Recall that in the context of lattices, a key idea is the notion of a \emph{dual-reduced} block $\vec{A}_{[i,j]}$, in which the \emph{last} Gram-Schmidt vector $\gs{\vec{a}}_j$ is as \emph{long} possible, while keeping $\lat(\vec{A}_{[i,j]})$ fixed.

Backward-reduced blocks are what we call the analogous notion for codes. Specifically, we say that a block $\vec{B}_{[i,j]}$ is \emph{backward reduced} if it is a proper and the last epipodal vector $\vec{b}_j^+$ is as \emph{long} as possible, while keeping $\C(\vec{B}_{[i,j]})$ fixed. Just like in the case of lattices, this idea is motivated by an invariant. Here, the invariant is $|\vec{b}_i^+| + \cdots + |\vec{b}_j^+|$, which is precisely the support size of the code $\C(\vec{B}_{[i,j]})$. So, if one wishes to make earlier epipodal vectors shorter (and we do!), then one will necessarily make later epipodal vectors longer, and vice versa. In particular, in the case of LLL, when the block size $\beta := j-i+1$ is equal to $2$, there is no difference between minimizing the length of the first epipodal vector and maximizing the length of the second epipodal vector. Therefore, one could describe the LLL algorithm in \cite{DDvAlgorithmicReductionTheory2022} as working by repeatedly backward reducing blocks $\vec{B}_{[i,i+1]}$.

The above definition of course leads naturally to two questions. First of all, how do we produce a backward-reduced block (for block size larger than $2$)? And, second, what can we say about them? Specifically, what can we say about the length $|\vec{b}_j^+|$ of the last epipodal vector in a backward-reduced block $\vec{B}_{[i,j]}$?

One might get discouraged here, as one quickly discovers that long last epipodal vectors \emph{do not} correspond to short non-zero codewords in the dual code. So, the beautiful duality that arises in the setting of lattices simply fails in our new context. (The \emph{only} exception is that last epipodal vectors with length \emph{exactly} two correspond to dual vectors with length \emph{exactly} two.) This is why we use the terminology ``backward reduced'' rather than ``dual reduced.'' One might fear that the absence of this correspondence would make backward-reduced blocks very difficult to work with.

Instead, we show that long last epipodal vectors $\vec{b}_j^+$ in a block $\vec{B}_{[i,j]}$ have a simple interpretation. They correspond precisely to large \emph{redundant sets of coordinates} of the code $\C(\vec{B}_{[i,j]})$. In the special case when $q = 2$, a redundant set $S \subseteq [n]$ of coordinates is simply a set of coordinates in the support of the code such that for every $a,b \in S$ and every codeword $\vec{c}$, $c_a = c_b$. For larger $q$, we instead have the guarantee that $c_a = z c_b$ for fixed non-zero scalar $z \in \F_q^*$ depending only on $a$ and $b$. In particular, maximal redundant sets correspond precisely to the non-zero coordinates in a last epipodal vector. (See \cref{lem:last_epipodal_redundant}.)

This characterization immediately yields an algorithm for backward reducing a block. (See \cref{alg: backward_reduce}.) In fact, finding a backward-reduced block boils down to finding a set of most common elements in a list of at most $n$ non-zero columns, each consisting of $\beta := j-i+1$ elements from $\F_q$. One quite surprising consequence of this is that one can actually find backward-reduced blocks efficiently, even for large $\beta$! (Compare this to the case of lattices, where finding a dual-reduced block for large $\beta$ is equivalent to the NP-hard problem of finding a shortest non-zero vector in a lattice of dimension $\beta$.)

Furthermore, this simple combinatorial characterization of backward-reduced blocks makes it quite easy to prove a simple tight lower bound on the length of $\vec{b}_j^+$ in a backward-reduced block $\vec{B}_{[i,j]}$. (See \fullornot{\cref{lem:lower_bound_on_max_last_epipodal}}{the full version \citegs}.) Indeed, such a proof follows immediately from the pigeonhole principle. This makes backward-reduced blocks quite easy to analyze. In contrast, as we will see below, \emph{forward-reduced} blocks, in which the first epipodal vector $\vec{b}_i^+$ is as short as possible, are rather difficult to analyze for $\beta > 2$.

\subsubsection{Fully backward-reduced bases}

With this new characterization of backward-reduced blocks and the realization that we can backward reduce a block quite efficiently, we go on to define the notion of a \emph{fully backward-reduced basis}. We say that a basis is \emph{fully backward reduced} if \emph{all} of the prefixes $\basis_{[1,j]}$ are backward reduced for all $1 \leq j \leq k$.\footnote{Notice that this implies that $\basis_{[i,j]}$ is also backward reduced for any $1 \leq i < j$. So, such bases really are \emph{fully} backward reduced.} In fact, we are slightly more general than this, and consider bases that satisfy this requirement for all $j$ up to some threshold $\tau \leq k$.

We show that a fully backward-reduced basis achieves the Griesmer bound (\cref{eq:griesmer} for $q = 2$), just like an LLL-reduced basis. This is actually unsurprising, since it is not difficult to see that when the threshold $\tau = k$ is maximal, a fully backward-reduced basis is also LLL reduced. However, even when $\tau = \log_q n$, we still show that a backward-reduced basis achieves the Griesmer bound. (See \cref{thm:griesmer_backwards}.)

We then show a very simple and very efficient algorithm for computing fully backward-reduced bases. In particular, if the algorithm is given as input a \emph{proper} basis, then it will convert it into a fully backward-reduced basis up to threshold $\tau$ in time $\tau^2 n \cdot \poly(\log n, \log q)$. Notice that this is \emph{extremely} efficient when $\tau \leq \poly(\log n)$.\footnote{We argue (in \fullornot{\cref{sec: k1_heuristic}}{the full version \citegs}) that there is not much point in taking $\tau$ significantly greater than $\log_q^2(n)$.} Indeed, for most parameters of interest, this running time is in fact less than the time $O(n k \log q)$ needed simply to read the input basis $\vec{B} \in \F_q^{k \times n}$. (Of course, this is possible because the algorithm only looks at the first $\tau$ rows of the input basis.) So, if one already has a proper basis, one can convert it into a fully backward-reduced basis nearly for free.\footnote{Computing a proper basis seems to require time $O(nk^2 \log^2 q)$ (without using fast matrix multiplication algorithms), but in many contexts the input basis is in systematic form and is therefore proper.}

In contrast, the LLL algorithm runs in time $O(k n^2 \log^2 q)$. One \emph{can} perform a similar ``threshold'' trick and run the LLL algorithm only on the first $\tau$ basis vectors for $\tau = \ceil{\log_q n}$ (which would still imply that $|\vec{b}_1|$ must be bounded by the Griesmer bound). But, this would still yield a running time of $\Omega(\tau n^2 \log^2 q)$ in the worst case. The speedup that we achieve from fully backward reduction comes from the combination of this threshold trick together with the fact that full backward reduction runs in time proportional to $\tau^2 n$, rather than $\tau n^2$.

Furthermore, we show empirically that the resulting algorithm tends to produce better bases than LLL in practice. (See \fullornot{\cref{sec: experiments}}{the full version \citegs}.)

(It seems unlikely that any similar algorithm exists for lattices for two reasons. First, in the setting of lattices, computing a dual-reduced basis for large block sizes is computationally infeasible. Second, while for codes it is not unreasonable to look for a short non-zero codeword in the subcode generated by the first $\tau$ basis vectors, for lattices the lattice generated by the first $k-1$ basis vectors often contains no shorter non-zero vectors than the basis vectors themselves, even when the full lattice contains much shorter vectors.)

\paragraph{Heuristic analysis of full backward reduction. } We also provide heuristic analysis of full backward reduction, providing a compelling heuristic explanation for why its performance in practice seems to be much better than what worst-case analysis suggests. In particular, recall that we essentially characterize the length of the last epipodal vector of a backward-reduced block $\vec{B}_{[i,j]}$ in terms of the maximal number of times that a column in $\vec{B}_{[i,j]}$ repeats. We then naturally use the pigeonhole principle to argue that for suitable parameters there must be a column that repeats many times, even in the worst case. 

E.g., for $q = 2$, there must be at least one repeated non-zero column if the number of non-zero columns $s$ is larger than the number of possible non-zero columns $2^{\beta}-1$, where $\beta := j-i+1$ is the length of a column. This analysis is of course tight in the worst case. However, in the average case when the columns are independent and uniformly random, we know from the birthday paradox that we should expect to see a repeated column even if $s$ is roughly $2^{\beta/2}$, rather than $2^\beta$.

So, under the (mild but unproven) heuristic that the blocks $\basis_{[1,j]}$ in a fully backward-reduced basis behave like a random matrices for all $j$ (in terms of the number of redundant coordinates), it is easy to see that $k_1 \gtrsim 2\log_q n$, which is significantly better than what LLL achieves (both in the worst case and empirically).

This heuristic argument is backed up by experiments. (See \fullornot{\cref{sec: experiments}}{the full version \citegs}.) We also show (in \fullornot{\cref{sec:selective}}{the full the version \citegs}) a less natural variant of this algorithm that provably achieves $k_1 \gtrsim 2 \log_q n$ when its input is a random matrix. This variant works by carefully ``choosing which coordinates to look at'' for each block, in order to maintain independence. We view this as an additional heuristic explanation for full backward reduction's practical performance, since one expects an algorithm that ``looks at all coordinates'' to do better than one that does not.

This result about $k_1$ for backward-reduced bases also compares favorably with the study of the LLL algorithm in \cite{DDvAlgorithmicReductionTheory2022}. In particular, in \cite{DDvAlgorithmicReductionTheory2022}, they proved that LLL achieves $k_1 \gtrsim \log n$ for a random code for $q =2$, but in their experiments they observed that LLL combined with a preprocessing step called EpiSort actually seems to achieve $k_1 \approx c \log n$ for some constant $1 < c \leq 2$. However, the behavior of LLL and EpiSort seems to be much more subtle than the behavior of full backwards reduction. We therefore still have no decent explanation (even a heuristic one) for why LLL and EpiSort seem to achieve $k_1\approx  c\log n$ or for what the value of this constant $c$ actually is.

\subsubsection{BKZ and slide reduction for codes} 

Our next set of contributions are adaptations of the celebrated BKZ and slide-reduction algorithms to the setting of codes.

\paragraph{BKZ for codes. } Our analogue of the BKZ algorithm for codes is quite natural.\footnote{We note that the name ``BKZ algorithm for codes'' is perhaps a bit misleading. In the case of lattices, the BKZ algorithm is named after Korkine and Zolotarev due to their work on Korkine-Zolotarev-reduced bases (which can be thought of as BKZ-reduced bases with maximal block size $\beta = k$, and is sometimes also called a \emph{Hermite}-Korkine-Zolotarev-reduced basis). A \emph{Block}-Korkine-Zolotarev-reduced basis is (unsurprisingly) a basis in which each block $\basis_{[i,i+\beta-1]}$ is a Korkine-Zolotarev-reduced basis. For codes, the analogous notion of a Korkine-Zolotarev-reduced basis was called a Griesmer-reduced basis in \cite{DDvAlgorithmicReductionTheory2022}. So, we should perhaps call our notion ``Block-Griesmer-reduced bases'' and the associated algorithm ``the block-Griesmer algorithm.'' However, the authors decided to use the term ``BKZ'' here in an attempt to keep terminology more consistent between lattices and codes.} Specifically, our algorithm works by repeatedly checking whether the epipodal vector $\vec{b}_i^+$ is a shortest non-zero codeword in the code generated by the block $\basis_{[i,i+\beta-1]}$. If not, it updates the basis so that this is the case (using the tools that we have developed to maintain properness). The algorithm does this repeatedly until no further updates are possible. At least intuitively, a larger choice of $\beta$ here requires a slower algorithm because the resulting algorithm will have to find shortest non-zero codewords in $\beta$-dimensional codes. But, larger $\beta$ will result in a better basis.

As we mentioned above, in the setting of lattices, the BKZ algorithm is considered to be the best performing basis-reduction algorithm in most parameter regimes, but it is notoriously difficult to analyze. We encounter a roughly similar phenomenon in the setting of linear codes. In particular, we run experiments that show that the algorithm performs quite well in practice. (Though it requires significantly more running time than full backward reduction to achieve a similar profile. See \fullornot{\cref{sec: experiments}}{the full version \citegs}.) However, we are unable to prove that it terminates efficiently, except in the special case of $\beta = 2$, in which case we recover the LLL algorithm of \cite{DDvAlgorithmicReductionTheory2022}. For $\beta > 2$, we offer only an extremely weak bound on the running time. As in the case of lattices, the fundamental issue is that it is difficult to control the effect that changing $\vec{b}_i^+$ can have on the other epipodal vectors $\vec{b}_{i+1}^+,\ldots, \vec{b}_{i+\beta-1}^+$ in the block.

Here, we encounter an additional issue as well. In the case of lattices, there is a relatively simple tight bound on the minimum distance of the lattice generated by the block $\vec{A}_{[i,i+\beta-1]}$ to the lengths of the Gram-Schmidt vectors $\|\gs{\vec{a}}_i\|, \ldots, \|\gs{\vec{a}}_{i+\beta-1}\|$ in the block. In particular, Minkowski's celebrated theorem tells us that $\lambda_1(\lat(\vec{A}_{[i,i+\beta-1]})) \leq C\sqrt{\beta} (\|\gs{\vec{a}}_i\|\cdots \|\gs{\vec{a}}_{i+\beta-1}\|)^{1/\beta}$ for some constant $C > 0$, and one applies this inequality repeatedly with different $i$ to understand the behavior of basis reduction for lattices.

However, in the case of codes, there is no analogous simple tight bound on $d_{\min}(\C(\basis_{[i,i+\beta-1]}))$ in terms of the lengths of the epipodal vectors $|\vec{b}_i^+|,\ldots, |\vec{b}_{i+\beta-1}^+|$, \emph{except} in the special case when $\beta = 2$. Instead, there are many known incomparable upper bounds on $d_{\min}$ in terms of the dimension $\beta$ and the support size $s := |\vec{b}_i^+| + \cdots + |\vec{b}_{i+\beta-1}|$ (and, of course, the alphabet size $q$). Each of these bounds is tight or nearly tight for some support sizes $s$ (for fixed $\beta$) but rather loose in other regimes.
The nature of our basis-reduction algorithms is such that different blocks have very different support sizes $s$, so that we cannot use a single simple bound that will be useful in all regimes. And, due to the relatively ``cramped'' nature of $\F_q^n$, applying even slightly loose bounds on $d_{\min}$ can easily yield trivial results, or results that do offer no improvement over the $\beta = 2$ case. 
As a result, the bound that we obtain on the length of $\vec{b}_1$ for a BKZ-reduced basis does not have a simple closed form. (Since the special case of $\beta = 2$ yields a very simple tight bound $d_{\min} \leq (1-1/q) s$, this is not an issue in the analysis of the LLL algorithm in \cite{DDvAlgorithmicReductionTheory2022}.) 

In fact, we do not even know if the worst-case bound on $|\vec{b}_1|$ for a BKZ-reduced basis is efficiently computable, even if one knows the optimal minimum distance of $\beta$-dimensional codes for all support sizes. However, we do show an efficiently computable bound that is nearly as good. And, we show empirically that in practice it produces quite a good basis. (See \fullornot{\cref{sec: experiments}}{the full version \citegs}.)

\paragraph{Slide reduction for codes.} Given our difficulties analyzing the BKZ algorithm, it is natural to try to adapt Gama and Nguyen's slide-reduction algorithm~\cite{gamaFindingShortLattice2008} from lattices to codes. In particular, recall that in the case of lattices, the slide-reduction algorithm has the benefit that (unlike BKZ) it is relatively easy to prove that it terminates efficiently.

In fact, recall that the idea for backward-reduced bases was inspired by dual-reduced bases for lattices, which are a key component of slide reduction. We therefore define slide-reduced bases for codes by essentially just substituting backward-reduced blocks for dual-reduced blocks in Gama and Nguyen's definition for lattices.  Our slide-reduction algorithm (i.e., an algorithm that produces slide-reduced bases) follows similarly.

We then give a quite simple proof that this algorithm terminates efficiently. Indeed, our proof is a direct translation of Gama and Nguyen's elegant potential-based argument from the case of lattices to the case of codes. (Gama and Nguyen's proof is itself a clever variant of the beautiful original proof for the case when $\beta = 2$ in \cite{lll82}.)

Finally, we give an efficiently computable upper bound on $|\vec{b}_1|$ for a slide-reduced basis in a similar spirit to our upper bound on BKZ. Here, we again benefit from our analysis of backward-reduced blocks described above. Indeed, the behavior of the epipodal vectors in our backward-reduced blocks is quite easy to analyze. However, our bound does not have a simple closed form because the behavior of the forward-reduced blocks still depends on the subtle relationship between the minimal distance of a code and the parameters $n$ and $k$, as we described in the context of BKZ above. 

In our experiments (in \fullornot{\cref{sec: experiments}}{the full version \citegs}), slide reduction is far faster than BKZ but does not find bases that are as good.

\subsubsection{Two illustrative algorithms}

In \fullornot{\cref{sec:pathological}}{the full version \citegs}, we show yet two more basis-reduction algorithms for codes. We think of the importance of these algorithms as being less about their actual usefulness and more about what they show about the potential and limitations of basis reduction for codes. We explain below.

\paragraph{One-block reduction. } The one-block-reduction algorithm is quite simple. It finds a short non-zero codeword in a code $\C$ generated by some basis $\vec{B}$ by first ensuring that $\vec{B}$ is proper, and then by simply finding a shortest non-zero codeword in the subcode $\C(\basis_{[1,\beta]})$ generated by the prefix basis $\basis_{[1,\beta]}$. Notice that if $\beta \leq O(\log_q n)$, then this algorithm runs in polynomial time. In particular, enumerating all codewords in the subcode can be done in time roughly $O(n q^\beta \log q)$.

Furthermore, it is not hard to see that when $\beta = \ceil{\log_q n}$, this simple algorithm actually meets the Griesmer bound! (See \fullornot{\cref{clm:one_block_griesmer}}{the full version \citegs}.) At a high level, this is because (1) the worst case in the Griesmer bound has $|\vec{b}_i^+| = 1$ for all $i \geq \beta$; and  (2) the resulting bound is certainly not better than the minimum distance of a code with dimension $\beta$ and support size $n-(k-\beta)$. Here, the $k-\beta$ term comes from the fact that $\supp(\basis_{[1,\beta]}) = n - |\vec{b}_{\beta+1}^+| - \cdots - |\vec{b}_k^+|$. (Similar logic explains why full backward reduction achieves the Griesmer bound with $\tau \approx \log_q n$.)

More generally, it seems unlikely that a basis-reduction algorithm will be able to find $\vec{b}_1$ that is shorter than what is achieved by this simple approach if we take $\beta \geq \max\{k_1^*,\beta'\}$, where $\beta'$ is block size of the basis reduction algorithm and $k_1^*$ is the maximal index of an epipodal vector that has length larger than one. (In practice, $k_1^*$ is almost never much larger than $k_1$.) In particular, for a basis reduction algorithm to do better than this, it must manage to produce a block $\basis_{[1,\beta]}$ that has minimum distance less than what one would expect given its support size.

We therefore think of this algorithm as illustrating two points. 

First, the existence of this algorithm further emphasizes the importance of the parameter $k_1^*$ (and the closely related parameter $k_1$) as a sort of ``measure of non-triviality.'' If an algorithm achieves large $k_1^*$, then the above argument becomes weaker, since we must take $\beta \geq k_1^*$. Indeed, if $\beta$ is significantly larger than $2 \log_q k$, then the running time of one-block reduction (if implemented by simple enumeration) becomes significantly slower.

Second, the existence of the one-block-reduction algorithm illustrates that we should be careful not to judge basis-reduction algorithms \emph{entirely} based on $|\vec{b}_1|$. We certainly think that $|\vec{b}_1|$ is an important measure of study, and indeed it is the main way that we analyze the quality of our bases in this work. However, the fact that one-block-reduction exists shows that this should not be viewed as the only purpose of a basis-reduction algorithm.

Of course, the algorithms that we have discussed thus far are in fact non-trivial, because they (1) find short non-zero codewords \emph{faster} than one-block reduction; and (2) find whole reduced bases and not just a single short non-zero codeword. Such reduced bases have already found exciting applications in \cite{DDvAlgorithmicReductionTheory2022} and \cite{DecodingChallenge}, and we expect them to find more.

\paragraph{Approximate Griesmer reduction. }  Recall that \cite{DDvAlgorithmicReductionTheory2022} calls a basis $\vec{B} \in \F_q^{k \times n}$ \emph{Griesmer reduced} if $\vec{b}_i^+$ is a shortest non-zero codeword in $\C(\basis_{[i,k]})$ for all $i$. And, notice that, if one is willing to spend the time to find shortest non-zero codewords in codes with dimension at most $k$, then one can compute a Griesmer-reduced basis iteratively, by first setting $\vec{b}_1$ to be a shortest non-zero codeword in the whole code, then projecting orthogonal to $\vec{b}_1$ and building the rest of the basis recursively. (Griesmer-reduced bases are the analogue of Korkine-Zolotarev bases for lattices. We discuss Griesmer-reduced bases more below.)

Our approximate-Griesmer-reduction algorithm is a simple variant of this idea. In particular, it is really a family of algorithms parameterized by a subprocedure that finds short (but not necessarily shortest) non-zero codewords in a code. Given such a subprocedure, the algorithm first finds a short non-zero codeword $\vec{b}_1$ in the input code $\C$. It then projects the code orthogonally to $\vec{b}_1$ and builds the rest of the basis recursively. (To make sure that we end up with a proper basis, care must be taken to assure that $\vec{b}_1$ is primitive. We ignore this in the introduction. See \fullornot{\cref{sec: approx}}{the full version \citegs}.)

The running time of this algorithm and the quality of the basis produced of course depends on the choice of subprocedure. Given the large number of algorithms for finding short non-zero codewords with a large variety of performance guarantees for different parameters (some heuristic and some proven), we do not attempt here to study this algorithm in full generality. We instead simply instantiate it with the Lee-Brickell-Babai algorithm from \cite{DDvAlgorithmicReductionTheory2022} (an algorithm which itself uses \cite{DDvAlgorithmicReductionTheory2022}'s LLL algorithm as a subroutine). Perhaps unsurprisingly, we find that this produces significantly better basis profiles (e.g., smaller $|\vec{b}_1|$ and larger $k_1$ and $k_1^*$) than all of the algorithms that we designed here. The price for this is, of course, that the subprocedure itself must run in enough time to find non-zero short codewords in dimension $k$ codes.

We view this algorithm as a proof of concept, showing that at least in principle one can combine basis-reduction techniques with other algorithms for finding short codewords to obtain bases with very good parameters. This meshes naturally with the Lee-Brickell-Babai algorithm in \cite{DDvAlgorithmicReductionTheory2022}, which shows how good bases can be combined with other algorithmic techniques to find short non-zero codewords. Perhaps one can merge these techniques more in order to show a way to use a good basis to find a better basis, which itself can be used to find a better basis, etc?

\subsubsection{On ``the best possible bases''}

Finally, in \fullornot{\cref{sec: k1_heuristic}}{the full version \citegs}, we prove bounds on ``the best possible bases'' in terms of the parameters $k_1$ and $k_1^*$. Indeed, recall that the (heuristic) running time of \cite{DDvAlgorithmicReductionTheory2022}'s Lee-Brickell-Babai algorithm beats Lee-Brickell by a factor that is exponential in $k_1$. And, we argued above that $k_1^*$ can be viewed as a measure of the ``non-triviality'' of a basis reduction algorithm. So, it is natural to ask how large $k_1$ and $k_1^*$ can be in principle.

In \fullornot{\cref{sec:k1_lower}}{the full version \citegs}, we show that \emph{any} code over $\F_2$ has a basis with $k_1^* \geq \Omega(\log k^2)$, provided that the support size is at least $k + \sqrt{k}$ (which is quite a minor restriction). For this, we use Griesmer-reduced bases (not to be confused with the approximate-Griesmer-reduced bases described above; note in particular that it is NP-hard to compute a Griesmer-reduced basis).  Notice that this is a factor of $\Omega(\log k)$ better than the logarithmic $k_1^*$ achieved by all known efficient basis-reduction algorithms.

Here, we use the parameter $k_1^*$ and not $k_1$ because it is easy to see that in the worst case a code can have arbitrarily large support but still have no proper basis with $k_1 > 1$.\footnote{For example, take that code $\F_2^{k-1} \cup (\F_2^{k-1} + \vec{c})$ where $\vec{c} := (1,1,\ldots,1) \in \F_2^n$. Any proper basis of this code must have $k-1$ epipodal vectors with length one and therefore must have $k_1 = 1$.} Typically, of course, one expects $k_1^*$ and $k_1$ to be very closely related, so that one can view this as heuristic evidence that typical codes have bases with $k_1 \geq \Omega(\log^2 k)$.

In \fullornot{\cref{sec:k1_upper}}{the full version \citegs}, we argue under a mild heuristic assumption that any basis for a random code over $\F_2$ has $k_1 \leq k_1^* \leq O(\log^2 k)$, even if the support size $n$ is a large polynomial in the dimension $k$.

Taken together, these results suggest that the best possible bases for codes over $\F_2$ that we should expect to find in practice should have $k_1 \approx k_1^* = \Theta(\log^2 k)$ for typical settings of parameters. Such a basis would (heuristically) yield a savings of $k^{\Theta(\log k)}$ in \cite{DDvAlgorithmicReductionTheory2022}'s Lee-Brickell-Babai algorithm. So, it would be very exciting to find an efficient algorithm that found such a basis.

On the other hand, our (heuristic) upper bound on $k_1$ suggests a limitation of basis reduction for codes. In particular, we should not expect any improvement better than $k^{\Theta(\log k)}$ in Lee-Brickell-Babai. And, the upper bound also suggests that basis-reduction algorithms are unlikely to outperform the simple one-block-reduction algorithm for block sizes larger than $\Omega(\log^2 k)$. 
\section{Preliminaries}
\label{sec:prelims}

\subsection{Some notation}

Logarithms are base two unless otherwise specified, i.e., $\log (2^x) = x$. We write $\vec{I}_m$ for the $m \times m$ identity matrix.

If $\vec{b}_1, \dots, \vec{b}_k \in \F_q^n$, then $(\vec{b}_1, \dots, \vec{b}_k) \in \F_q^{n \times k}$ denotes the matrix where each $\vec{b}_i$ is a column and $(\vec{b}_1; \dots; \vec{b}_k) \in \F_q^{k \times n}$ denotes the matrix where each $\vec{b}_i$ is row $i$ of $\vec{B}$.

We use the notation $\vec{c}|_{S} \in \F_q^k$ to denote the projection of $\vec{c} = (c_1, \dots, c_n) \in \F_q^n$ onto a set of coordinates $S = \{ s_1, \dots, s_k \} \subseteq [1, n]$ where $s_1 \leq \dots \leq s_k$, i.e.,
$$\vec{c}|_{S} = (c_{s_1}, c_{s_2}, \dots, c_{s_k})
\; ,$$
and we similarly extend this to matrices.

We say that a matrix $\vec{B} \in \F_q^{k \times n}$ is \emph{in systematic form} if $\vec{A} = (\mathbf{I}_k, \vec{X}) \vec{P}$, where $\vec{P}$ is a permutation matrix (i.e., if $k$ contains the columns $\vec{e}_1^T,\ldots, \vec{e}_k^T$). 

For any basis $\vec{B} \in \F_q^{k \times n}$ and any subset $S \subseteq [n]$ with $|S| = k$ such that $\vec{B}|_{S}$ has full rank, we call the process of replacing $\vec{B}$ by $(\vec{B}|_S)^{-1} \vec{B}$ \emph{systematizing $\vec{B}$ with respect to $S$}. When the set $S$ is not important, we simply call this \emph{systematizing $\vec{B}$}. This procedure is useful at least in part because it results in a proper basis.

 We define two notions of the support of a vector. Specifically, we write
\[\supp(\vec{x}) := \{ i \in \brak{1, n} : x_i \neq 0 \}
\; ,\]
and similarly
\begin{align*}
\suppvec(\vec{x})_i := 
    \begin{cases} 
      0 & x_i = 0 \\
      1 & x_i \neq 0
      \; ,
   \end{cases}
\end{align*}
i.e., $\suppvec(\vec{x}) \in \F_q^n$ is the indicator vector of $\supp(\vec{x}) \subseteq [n]$.
We can also define the support of an $[n, k]_q$ code $\C$ by extending the definitions of $\supp$ and $\suppvec$,
\[
    \supp(\C) \triangleq \bigcup_{\vec{c} \in \C} \supp(\vec{c}) \qquad \qquad \suppvec(\C) = \bigvee_{\vec{c} \in \C} \suppvec(\vec{c})
    \;,
\]
and we define the support of a matrix $\vec{B} \in \F_q^{k \times n}$ as the support of the code generated by the matrix.

If $\vec{A} \in \F_q^{m \times n}$, $\vec{B} \in \F_q^{r \times s}$, then the direct sum of $\vec{A}$ and $\vec{B}$, denoted $\vec{A} \oplus \vec{B} \in \F_q^{(m+r) \times (n+s)}$, is
\[
\vec{A} \oplus \vec{B} = 
\begin{pmatrix}
\vec{A} & \vec{0}_{m \times s}\\
\vec{0}_{n \times r} & \vec{B}
\end{pmatrix}
\]
We will often use the following important property regarding matrix direct sums. If $\vec{A} \in \F_q^{m \times n}$, $\vec{B} \in \F_q^{r \times s}, \vec{x} \in \F_q^n, \vec{y} \in \F_q^s$, then
\[
(\vec{A} \oplus \vec{B}) 
\begin{pmatrix}
\vec{x}\\
\vec{y}
\end{pmatrix}
= 
\begin{pmatrix}
\vec{A} \vec{x}\\
\vec{B} \vec{y}
\end{pmatrix}
\; .
\]

We write $\odot$ to represent the Hadamard product. That is, for any $\vec{A}, \vec{B} \in \F_q^{m \times n}$, $\vec{A} \odot \vec{B} \in \F_q^{m \times n}$ is defined so that $$(\vec{A} \odot \vec{B})_{i, j} = \vec{A}_{i, j} \vec{B}_{i, j}
\; .$$

\subsection{Some basic probability facts}

We will need the following useful facts.

\begin{fact}
    \label{fact:random_full_rank}
    For any $\ell \geq k$, if $\vec{B} \sim \F_q^{k \times \ell}$, then $\vec{B}$ has $k$ linearly independent columns with probability at least
    \[
        1- q^{-\ell} \cdot \frac{q^{k}-1}{q-1}
        \; .
    \]
\end{fact}
\begin{proof}
    Since the row rank of a matrix is equal to its column rank, we can equivalently study the probability that the rows $\vec{b}_1,\ldots, \vec{b}_k$ of $\vec{B}$ are linearly independent. But, the probability that $\vec{b}_i$ is linearly \emph{dependent} on the first $i-1$ vectors is at most $q^{i-1-\ell}$ because there are at most $q^{i-1}$ vectors in the subspace spanned by $\vec{b}_1,\ldots, \vec{b}_{i-1}$. By the union bound, the probability that any of the $\vec{b}_i$ are linearly dependent is at most
    \[
        1/q^\ell + 1/q^{\ell-1} + \cdots + 1/q^{\ell-k+1} = q^{-\ell} \cdot \frac{q^{k}-1}{q-1}
        \; ,
    \]
    and the result follows.
\end{proof}

\begin{lemma}
	\label{lemma:chernoff}
	Suppose $X_1, \dots, X_n$ are independent random variables taking values in \{0, 1\}. Let $X$ denote their sum, and $\mu := \E[X]$. Then for any $0 \leq \eps \leq 1$
	\[\Pr[X \leq (1 - \eps) \mu] \leq \exp\Big(\frac{-\eps^2 \mu}{2}\Big).\]
\end{lemma}

\begin{lemma}[{\cite{hoeffdingProbabilityInequalitiesSums1963}}]
    \label{lem:repeated_column}
    If $\vec{B} \sim \F_q^{i \times \beta}$ is sampled uniformly at random with $\beta \geq 7$, then with probability at least $1-2^{-\beta/20} - 2^{-\beta^2/(20 q^i)}$, $\vec{B}$ will have at least two identical non-zero columns.
\end{lemma}
\begin{proof}
    Let $\ell$ be the number of non-zero columns in $\vec{B}$. Conditioned on some fixed value of $\ell$, the probability that all $\ell$ of these non-zero columns are distinct is 
    \begin{align*}
        (1-1/(q^i-1)) \cdot (1-2/(q^i-1)) \cdots (1-(\ell-1)/(q^i-1))
            &\leq e^{-1/q^i - 2/q^i - \cdots - (\ell-1)/q^i}
            \\
            &= e^{-\ell(\ell-1)/(2q^i)}
        \; .
    \end{align*}
    Furthermore, by the Chernoff-Hoeffding bound \cref{lemma:chernoff} we have that $\ell \geq \ceil{\beta/3}$ except with probability at most $2^{-\beta/20}$. The result follows.
\end{proof} 
\section{Generalizing epipodal vectors, size reduction, and the fundamental domain to \texorpdfstring{$\F_q$}{Fq}}
\label{sec: fq}
In this section, we generalize many of the fundamental concepts in \cite{DDvAlgorithmicReductionTheory2022} from codes over $\F_2$ to codes over $\F_q$. Specifically, we generalize the notions of projection, epipodal matrices, and the size-reduction algorithm. We then study the geometry of the fundamental domain that one obtains by running the size-reduction algorithm on a given input basis. 

Much of this generalization is straightforward (once one knows the theory developed for $\F_2$ in \cite{DDvAlgorithmicReductionTheory2022}). So, one might read much of this section as essentially an extension of the preliminaries. The most difficult part, in \fullornot{\cref{subsec: F_analysis}}{the full version \citegs}, is the analysis of the fundamental domain (which is not used in the rest of the paper).

\subsection{Projection and epipodal vectors}

The notions of projection and epipodal vectors extend naturally to $\F_q$ from the notions outlined in \cite{DDvAlgorithmicReductionTheory2022} for $\F_2$. However, to ensure that this work is as self-contained as possible, we will now explicitly outline how these notions extend to $\F_q$. Notice that these operations are roughly analogous to orthogonal projection maps over $\R^n$.

\begin{definition}
    If $\vec{x}_1 = (x_{1,1}, \dots, x_{1, n}), \dots, \vec{x}_k = (x_{k,1}, \dots, x_{k, n}) \in \F_{q}^{n}$, the function $\pi_{\{\vec{x}_1, \dots, \vec{x}_k\}}: \F_q^{n} \to \F_q^n$ is defined as follows:
    \[
    \pi_{\{\vec{x}_1, \dots, \vec{x}_k\}}(\vec{y})_i = 
    \begin{cases} 
      y_i & x_{1,i} \neq 0 \lor \dots \lor x_{k, i} \neq 0\\
      0 & \text{otherwise.}
   \end{cases}
    \]
        We call this ``projection onto the support of $\vec{x}_1,\ldots, \vec{x}_k$.''
\end{definition}
\begin{definition}
    If $\vec{x}_1 = (x_{1,1}, \dots, x_{1, n}), \dots, \vec{x}_k = (x_{k,1}, \dots, x_{k, n}) \in \F_{q}^{n}$, the function $\pi^\perp_{\{\vec{x}_1, \dots, \vec{x}_k\}}: \F_q^{n} \to \F_q^n$ is defined as follows: 
    \[
    \pi^\perp_{\{\vec{x}_1, \dots, \vec{x}_k\}}(\vec{y})_i = 
    \begin{cases} 
      y_i & x_{1, i} = 0 \land \dots \land x_{k, i} = 0\\
      0 & \text{otherwise.}
   \end{cases}
    \]
    We call this ``projection orthogonal to $\vec{x}_1,\ldots, \vec{x}_k$.''
\end{definition}
\noindent We will often simply write $\pi_{\vec{x}}$ to denote $\pi_{\{ \vec{x} \}}$ and $\pi^\perp_{\vec{x}}$ to denote $\pi^\perp_{\{ \vec{x} \}}$

\begin{fact}
    \label{lem: projection_linearity}
    For any $c \in \F_q$, $\vec{x} \in \F_{q'}^n$, $\vec{y}, \vec{z} \in \F_q^n$
    
    $$\pi^{\perp}_{\vec{x}}(c \vec{y} + \vec{z}) = c \pi^\perp_{\vec{x}}(\vec{y}) + \pi^{\perp}_{\vec{x}}(\vec{z})$$
\end{fact}

We now define the epipodal matrix of a basis for a code, which is the analogue of the Gram–Schmidt matrix.
 \begin{definition}
    Let $\vec{B} = (\vec{b}_1; \dots ; \vec{b}_k) \in \F_q^{k \times n}$ be a matrix with elements from $\F_q$. The $i$th projection associated to the matrix $\vec{B}$ is defined as $\pi_i := \pi^\perp_{\{\vec{b}_1, \dots, \vec{b}_{i-1}\}}$, where $\pi_1$ denotes the identity.

    The $i$th \emph{epipodal vector} is defined as $\vec{b}^+_i := \pi_i(\vec{b}_i)$.
    The matrix $\vec{B}^+ := (\vec{b}_1^+; \dots; \vec{b}_k^+) \in \F_q^{k \times n}$ is called the epipodal matrix of $\vec{B}$. 
\end{definition}

We will use the following fact quite a bit.

\begin{fact}
    \label{fact: lengt_invariance}
    For any basis $\vec{B} = (\vec{b}_1; \dots; \vec{b}_k) \in \F_q^{k \times n}$ of an $[n, k]_q$ code $\C$, if $|\supp(\C)|$ denotes the size of the support of the code, then
    $$|\supp(\C)| = \sum_{i=1}^k |\vec{b}_i^+|$$
\end{fact}

The following notation for a projected block will be helpful in defining our reduction algorithms. (The same notation is used in the lattice literature.) 

\begin{definition}
    \label{def: proj_subbasis}
    For a basis $\vec{B} = (\vec{b}_1; \dots; \vec{b}_k) \in \F_q^{k \times n}$ and $i, j \in [1, k]$ where $i \leq j$, we use the notation $\vec{B}_{[i, j]}$ as shorthand for $(\pi_i(\vec{b}_i); \dots; \pi_i(\vec{b}_j))$. Furthermore, for $i \in [1, k]$ and $j>k$, we define $\vec{B}_{[i, j]} = \vec{B}_{[i, k]}$ for all $j > k$.
\end{definition}

We will often write $\ell_i$ to denote $|\vec{b}_i^+|$ when the basis $\vec{B} = (\vec{b}_1; \dots; \vec{b}_k)$ is clear from context. 
To speed up the computation of the epipodal vectors, we introduce a matrix $\vec{S}(\vec{B})$ defined recursively by the recurrence
\begin{align*}
\vec{s}_0 = \vec{0}, &&
\vec{s}_i = \vec{s}_{i-1} \lor \suppvec(\vec{b}_i)
\; .
\end{align*}
In words, $\vec{s}_i$ is the support vector corresponding to the code generated by $\vec{b}_1,\ldots, \vec{b}_i$. 

Notice that $\vec{S}(\vec{B})$ is closely related to the epipodal matrix $\vec{B}^+$. In particular,
\begin{align*}
\pi_i(\vec{x}) = \pi^\perp_{\vec{s}_i}(\vec{x}), &&
\vec{b}_i^+ = \pi^\perp_{\vec{s}_i}(\vec{b}_i)
\; .
\end{align*}
In fact, as \cite{DDvAlgorithmicReductionTheory2022} observed, when working with the epipodal vectors $\vec{b}_1^+,\ldots, \vec{b}_k^+$, it is often convenient to store $\vec{S}(\vec{B})$ as a sort of auxiliary data structure. In particular, if we perform elementary row operations on the rows $\vec{b}_i, \vec{b}_{i+1}, \dots, \vec{b}_{i+\beta}$ of $\vec{B}$ (as we will often), then we can update $\vec{S}(\vec{B})$ in time $O(n \beta \log q)$---in particular, we only need to recompute $\vec{s}_i, \vec{s}_{i+1}, \dots, \vec{s}_{i+\beta}$. 

We can then use $\vec{S}(\basis)$ to compute $\vec{B}_{[i, i+\beta]}$ as $(\pi_{\vec{s}_i}^\perp(\vec{b}_i); \dots; \pi_{\vec{s}_i}^\perp(\vec{b}_{i + \beta}))$ in $O(n \beta \log(q))$ time. This will be particularly helpful for us since we focus on block reduction techniques, so we often need to compute $\vec{B}_{[i, i+\beta]}$ or update a contiguous set of vectors (a block) $\vec{b}_i, \dots, \vec{b}_{i+\beta}$. 

So, by maintaining the auxiliary data structure $\vec{S}(\vec{B})$, we can perform many elementary row operations on our basis $\basis$ and then compute a block $\basis_{[i,j]}$ in essentially the same amount of time that it takes to perform the elementary row operations (without the auxiliary data structure) and simply to read the block $\basis_{[i,j]}$---i.e., with essentially no overhead.
To keep our descriptions of our algorithms uncluttered, we will not explicitly mention the matrix $\vec{S}(\vec{B})$, but of course when our algorithms perform elementary row operations on our basis $\basis$ and then later compute a block $\vec{B}_{[i,j]}$, they always do so using this auxiliary data structure $\vec{S}(\vec{B})$.

\subsubsection{Basic operations on blocks}

It will often be beneficial for us to think about performing a linear transformation $\vec{A}$ to a block $\vec{B}_{[i, j]}$ ``without affecting the rest of the basis.'' In this section, we prove a number of basic lemmas about how to do this.

This first lemma shows how to modify $\vec{B}$ in order to apply a linear transformation $\vec{A} \in \F_q^{(j-i+1) \times (j-i+1)}$ to a block $\basis_{[i,j]}$.

\begin{lemma}
    \label{lem: lifting_reduction}
    For any basis $\vec{B} \in \F_q^{k \times n}$, integers $1 \leq i \leq j \leq k$, and  invertible matrix $\vec{A} \in \F_q^{(j-i+1) \times (j-i+1)}$,
    $$((\vec{I}_{i-1} \oplus \vec{A} \oplus \vec{I}_{k-j}) \vec{B})_{[i, j]} = \vec{A} (\vec{B}_{[i, j]})$$
\end{lemma}

\begin{proof}
    Let us first decompose $\vec{B} = (\vec{b}_1; \dots; \vec{b}_k)$ into  $\vec{B} = (\vec{B}_1; \vec{B}_2; \vec{B}_3)$ such that $\vec{B}_1 \in \F_q^{(i-1) \times n}$, $\vec{B}_2 \in \F_q^{(j-i+1) \times n}$, and $\vec{B}_3 \in \F_q^{(k-j) \times n}$. Let $\vec{v} := \suppvec(\vec{B}_1)$. Then, we have
    \begin{align*}
        ((\vec{I}_{i-1} \oplus \vec{A} \oplus \vec{I}_{k-j}) \vec{B})_{[i, j]} &= ((\vec{I}_{i-1} \oplus \vec{A} \oplus \vec{I}_{k-j}) (\vec{B}_1; \vec{B}_2; \vec{B}_3))_{[i, j]}\\
        &= (\vec{B}_1; \vec{A} \vec{B}_2 ; \vec{B}_3)_{[i, j]}
        \; .
    \end{align*}
    Let $\vec{A} \vec{B}_2 = (\vec{c}_1; \dots; \vec{c}_{j-i+1})$. Notice that for every $t \in [1, j-i+1]$, $\vec{c}_t = A_{t, 1} \vec{b}_i + \dots + A_{t, i-j+1} \vec{b}_j$ where $A_{g, h}$ denotes the entry in row $g$ and column $h$ of $\vec{A}$. So row $t$ of $(\vec{B}_1; \vec{A} \vec{B}_2 ; \vec{B}_3)_{[i, j]}$ is equal to
    $$\pi_{\vec{v}}^\perp(\vec{c}_t)=
    \pi_{\vec{v}}^\perp(A_{t, 1} \vec{b}_i + \dots + A_{t, i-j+1} \vec{b}_j)
    = A_{t, 1} \pi_{\vec{v}}^\perp(\vec{b}_i) + \dots + A_{t, i-j+1} \pi_{\vec{v}}^\perp(\vec{b}_j)$$
    where the last equality holds due to the linearity of $\pi^\perp_{\vec{v}}$ (\cref{lem: projection_linearity}). Notice that this is exactly row $t$ of $\vec{A} (\vec{B}_{[i, j]})$, as needed.
\end{proof}

The next few lemmas tell us how applying a linear transformation to a block $\vec{B}_{[i,j]}$ as above affect other blocks.

\begin{lemma}
    \label{lem: epipodal_locality}
    For any basis $\vec{B} = (\vec{b}_1; \dots; \vec{b}_k) \in \F_q^{k \times n}$, integers $1 \leq i \leq j \leq k$, and invertible matrix $\vec{A} \in \F_q^{(j-i+1) \times (j-i+1)}$, if we let $\vec{B}' := (\vec{b}'_1; \dots; \vec{b}'_k) = (\vec{I}_{i-1} \oplus \vec{A} \oplus \vec{I}_{k-j}) \vec{B}$, then for any integer $1 \leq m_1 \leq m_2 \leq k$ such that $[m_1, m_2] \cap [i, j] = \emptyset$, we have $\vec{B}_{[m_1, m_2]} = \vec{B}'_{[m_1, m_2]}$. In particular, $|\vec{b}_m^+| = |(\vec{b}'_m)^+|$ for all $m \notin [i, j]$.
\end{lemma}

\begin{proof}
    Let $\vec{B}^* := (\vec{b}_i; \dots; \vec{b}_j)$. Then $\vec{B}'= (\vec{b}_1; \dots; \vec{b}_{i-1}; \vec{A} \vec{B}^*; \vec{b}_{j+1}; \dots; \vec{b}_k)$. If $m_1 \leq m_2 < i$, then since $(\vec{b}_1; \dots, \vec{b}_{m_2}) = (\vec{b}'_1; \dots, \vec{b}'_{m_2})$, we have $\vec{B}_{[m_1, m_2]} = \vec{B}'_{[m_1, m_2]}$. 

    In the case that $m_1 > j$, $\C(\vec{b}_1, \dots, \vec{b}_{m_1-1}) = \C(\vec{b}'_1, \dots, \vec{b'}_{m_1-1})$ since $\vec{I}_{i-1} \oplus \vec{A} \oplus \vec{I}_{m_1-j-1}$ is invertible.
    This in particular implies that $\pi_{\vec{b}_1',\ldots, \vec{b}_{m_1-1}'}^\perp = \pi_{\vec{b}_1,\ldots, \vec{b}_{m_1-1}}^\perp$. And, since $\vec{b}_{\ell}' = \vec{b}_\ell$ for all $\ell \geq m_1 > j$, by definition, it follows that
    \[
        \vec{B}_{[m_1,m_2]}' =( \pi_{\vec{b}_1',\ldots, \vec{b}_{m_1-1}'}^\perp(\vec{b}_{m_1}');\ldots; \pi_{\vec{b}_1',\ldots, \vec{b}_{m_1-1}'}^\perp(\vec{b}_{m_2}'))  = (\pi_{\vec{b}_1,\ldots, \vec{b}_{m_1-1}}^\perp(\vec{b}_{m_1});\ldots;\pi_{\vec{b}_1,\ldots, \vec{b}_{m_1-1}}^\perp(\vec{b}_{m_2})) = \vec{B}_{[m_1,m_2]}
        \; ,
    \]
    as claimed. (The fact that $\vec{b}_m^+ = (\vec{b}_m')^+$ follows by setting $m_1 = m_2 = m$.)
\end{proof}

Finally, we observe that a certain simple modification to a basis does not change the epipodal vectors.

\begin{lemma}
    \label{lem: constant_epipodal_sum}
    For any basis $\vec{B} = (\vec{b}_1; \dots; \vec{b}_k) \in \F_q^{k \times n}$, integers $1 \leq i \leq j \leq k$, and invertible $\vec{A} \in \F_q^{(j-i+1) \times (j-i+1)}$, if we let $\vec{B}' := (\vec{b}'_1; \dots; \vec{b}'_k) := (\vec{I}_{i-1} \oplus \vec{A} \oplus \vec{I}_{k-j}) \vec{B}$, then $\sum_{m=i}^j |(\vec{b}'_m)^+| = \sum_{m=i}^j |\vec{b}_m^+|$.
\end{lemma}

\begin{proof}
     By \cref{lem: lifting_reduction}, $\vec{B}'_{[i,j]} = \vec{A} \vec{B}_{[i,j]}$. Since $\vec{A}$ is invertible, this means that $\C(\vec{B}'_{[i,j]}) = \C(\vec{B}_{[i,j]})$. In particular, these two codes have the same support. The result then follows by \cref{fact: lengt_invariance}, which shows that the two sums must both be equal to this support.
\end{proof}

\begin{lemma}
    \label{lem: late_epipodal_invariant}
    For any basis $\vec{B} = (\vec{b}_1; \dots; \vec{b}_k) \in \F_q^{k \times n}$, $c \in \F_q$, and integers $1 \leq i < j \leq k$, let $\vec{B}' := (\vec{b}_1;\ldots, \vec{b}_{j-1};\vec{b}_j + c \vec{b}_i; \vec{b}_{j+1};\ldots;\vec{b}_k)$. Then, $\vec{b}_m^+ = (\vec{b}_m')^+$ for all $m$.
\end{lemma}

\begin{proof}
    It follows from \cref{lem: epipodal_locality} that $\vec{b}_m^+ = (\vec{b}_m')^+$ for all $m \notin [i,j]$. For $m \in [i,j]$ with $m \neq j$, it follows from \cref{lem: lifting_reduction} that $\vec{b}_m^+ = (\vec{b}_m')^+$. Finally, notice that $\pi_{\vec{b}_1,\ldots, \vec{b}_{j-1}}^\perp = \pi_{\vec{b}_1',\ldots, \vec{b}_{j-1}'}^\perp$, since the first $j-1$ basis vectors are unchanged. It follows that
    \[
        (\vec{b}_j')^+ = \pi_{\vec{b}_1',\ldots, \vec{b}_{j-1}'}^\perp(\vec{b}_j + c \vec{b}_i) = \pi_{\vec{b}_1,\ldots, \vec{b}_{j-1}}^\perp(\vec{b}_j) + c \pi_{\vec{b}_1,\ldots, \vec{b}_{j-1}}^\perp(\vec{b}_i) =\pi_{\vec{b}_1,\ldots, \vec{b}_{j-1}}^\perp(\vec{b}_j) = \vec{b}_j^+
        \; ,
    \]
    where we have used the linearity of projection plus the fact that $\pi_{\vec{b}_1,\ldots, \vec{b}_{j-1}}^\perp(\vec{b}_i) = \vec0$.
\end{proof}

\subsection{Size reduction and its fundamental domain}
\label{sec:size_reduction}

(This section and the next can be safely skipped by a reader who is only interested in our basis reduction algorithms.)

We will now define the size-reduction algorithm. Intuitively, the size reduction algorithm takes as input $\vec{e}$ and then works to iteratively minimize $|\pi_{\vec{b}_i^+}(\vec{e})|$  by adding scalar multiples of $\vec{b}_i$ to $\vec{e}$ for $i = k,\ldots, 1$. To ensure that the behavior of the algorithm is unambiguous (which is important for its analysis), we must be careful to appropriately handle the case when there are multiple choices of $a$ such that $|\pi_{\vec{b}_i^+}(\vec{e} + \vec{b}_i)|$ is minimal. 

To that end (following \cite{DDvAlgorithmicReductionTheory2022}) we define the a tie-breaking function $\mathrm{TB}$. In fact, we could use any function $\mathrm{TB}_{\vec{p}}(\vec{y})$ for $\vec{p} \neq 0$ with the properties that (1) $0 \leq \mathrm{TB}_{\vec{p}}(\vec{y}) < 1$, and (2) if $\vec{p} \neq \vec0$, then $\mathrm{TB}_{\vec{p}}(\vec{y}) \neq \mathrm{TB}_{\vec{p}}(\vec{y} + a\vec{p})$ for any $\vec{y} \in \F_q^n$ and $a \in \F_q^*$. It will be convenient to take 
\[
\mathrm{TB}_{\vec{p}}(\vec{y}) := y_j p_j^{-1}/q
\; ,
\]
where $j := \min \supp(\vec{p})$ is the minimal index in the support of $\vec{p}$ and where we interpret this as a value in $[0,1)$ by identifying each $y_j \in \F_q$ with a unique integer in $\{0,\ldots, q-1\}$ (in any way). We note that our choice of the tie-breaking function has the advantage that for any $\vec{c} \in (\F_q^*)^n$, $\mathrm{TB}_{\vec{p}}(\vec{y}) = \mathrm{TB}_{\vec{c} \odot \vec{p}}(\vec{c} \odot \vec{y})$ (where recall that $\odot$ represents the coordinate-wise product of vectors).

\begin{definition}
    \label{def: fundamental_domain}
    Let $\vec{B} = (\vec{b}_1, \dots, \vec{b}_k)$ be a proper basis. The fundamental domain relative to $\vec{B}$ is defined as
    $$
    \mathcal{F}(\vec{B}^+) := 
    \{ \vec{y} \in \mathbb{F}_q^n :
    \forall i\in [1, k], \forall a \in \F_q^*,
    |\pi_{\vec{b}_i^+}(\vec{y})| + \mathrm{TB}_{\vec{b}_i^+}(\vec{y})
    < |\pi_{\vec{b}_i^+}(a \vec{b}_i+\vec{y})| + \mathrm{TB}_{\vec{b}_i^+}(a \vec{b}_i+\vec{y})\}
    \; .
    $$
    Vectors in the fundamental domain are said to be \emph{size reduced with respect to $\vec{B}$}.
\end{definition}

\cref{def: fundamental_domain} captures the idea that a word $\vec{y} \in \F_q^n$ is in the fundamental domain if, for all $i$, $|\pi_{\vec{b}_i^+}(\vec{y})| + \mathrm{TB}_{\vec{b}_i^+}(\vec{y})$ cannot be made smaller by adding any multiple of the basis vector $\vec{b}_i$ to it. We now show that what we have defined as the fundamental domain in \cref{def: fundamental_domain} is indeed a fundamental domain of the code $\C$ by showing that it is both covering and packing. In fact, we only prove that it is packing. The fact that it is covering will follow immediately from the correctness of our size-reduction algorithm.

\begin{claim}
    Let $\vec{B} := (\vec{b}_1; \dots, \vec{b}_k) \in \F_q^{k \times n}$ be a proper basis for a code $\C$. Then, $\mathcal{F}(\vec{B}^+)$ is $\mathcal{C}$-packing, i.e., for every non-zero codeword $\vec{c} \in \C_{\neq \vec0}$, 
    \[(\vec{c} + \mathcal{F}(\vec{B}^+)) \cap \mathcal{F}(\vec{B}^+) = \emptyset
    \]
\end{claim}

\begin{proof}
    Suppose for contradiction that $\vec{x} \in (\vec{c} + \mathcal{F}(\vec{B}^+) \cap \mathcal{F}(\vec{B}^+)$. Then, notice that both $\vec{x}$ and $\vec{x} - \vec{c}$ are in $\mathcal{F}(\vec{B}^+)$. Since $\vec{c} \in \C_{\neq \vec0}$, there are some coefficients $a_i \in \mathbb{F}_q$ not all zero such that $\vec{c} = \sum_{i=1}^k a_i \vec{b}_i$. 
    
    Let $i$ be maximal such that $a_i \neq 0$. Notice that since $\vec{x} \in \mathcal{F}(\vec{B}^+)$ and $a_i \neq 0$, we must have that 
    \[|\pi_{\vec{b}_i^+}(\vec{x})| + \mathrm{TB}_{\vec{b}_i^+}(\vec{x}) < |\pi_{\vec{b}_i^+}(\vec{x} - a_i \vec{b}_i)| + \mathrm{TB}_{\vec{b}_i^+}(\vec{x} - a_i \vec{b}_i)
    \; .
    \]
    But, notice that by the definition of $i$, we have that $\pi_{\vec{b}_i^+}(\vec{c}) = a_i \vec{b}_i$ and notice that $\mathrm{TB}_{\vec{b}_i^+}(\vec{y})$ depends only on $\pi_{\vec{b}_i^+}(\vec{y})$. It follows that
    \begin{align*}
        |\pi_{\vec{b}_i^+}(\vec{x} - \vec{c} + a_i\vec{b}_i)| + \mathrm{TB}_{\vec{b}_i^+}(\vec{x} - \vec{c} + a_i \vec{b}_i) 
            &= |\pi_{\vec{b}_i^+}(\vec{x})| + \mathrm{TB}_{\vec{b}_i^+}(\vec{x}) \\
            &< |\pi_{\vec{b}_i^+}(\vec{x} - a_i \vec{b}_i)| + \mathrm{TB}_{\vec{b}_i^+}(\vec{x} - a_i \vec{b}_i) \\
            &= |\pi_{\vec{b}_i^+}(\vec{x} - \vec{c})| + \mathrm{TB}_{\vec{b}_i^+}(\vec{x} - \vec{c})
        \; .
    \end{align*}
    But, this contradicts the assumption that $\vec{x} - \vec{c} \in \mathcal{F}(\basis^+)$, as needed.
\end{proof}

We present in \cref{alg:size_red} our generalization to $\F_q$ of the size-reduction algorithm of \cite{DDvAlgorithmicReductionTheory2022}. The algorithm uses as a subprocedure a process for finding $a \in \F_q$ minimizing $|\pi_{\vec{b}_i^+}(\vec{e}+a\vec{b}_i^+)| + \mathrm{TB}_{\vec{b}_i^+}(\vec{e}+a\vec{b}_i^+)$. We show how to do this in $O(n \log^2 q)$ time in \cref{subsec: log_q}.

\RestyleAlgo{ruled}
\begin{algorithm2e}
\KwIn{A proper basis $\vec{B} = (\vec{b}_1; \dots; \vec{b}_k) \in \F_q^{k \times n}$ and a target $\vec{y} \in \F_q^n$}
\KwOut{$\vec{e} \in \mathcal{F}(\vec{B}^+)$ such that $\vec{e} - \vec{y} \in \C(\vec{B})$}
\caption{Size-reduce $\vec{y}$ with respect to $\vec{B}$}
\label{alg:size_red}

$\vec{e} \gets \vec{x}$\\
\For{$i = k$ down to 1}{

find $a \in \F_q$ minimizing $|\pi_{\vec{b}_i^+}(\vec{e}+a\vec{b}_i^+)| + \mathrm{TB}_{\vec{b}_i^+}(\vec{e}+a\vec{b}_i^+)$\

$\vec{e} \gets \vec{e} + a \vec{b}_i$
}

\Return $\vec{e}$

\end{algorithm2e}

\begin{claim}
    On input a proper basis $\basis \in \F_q^{k \times n}$ and target $\vec{y} \in \F_q^n$ \cref{alg:size_red} is correctly outputs a point $\vec{e} \in \mathcal{C}(\basis)$ with $\vec{e} \in \mathcal{F}(\basis^+)$. Furthermore, the algorithm runs in time $O( kn\log^2 q)$.
\end{claim}

\begin{proof}
    To see that the algorithm runs in time $O(nk \log^2 q)$, it suffices to notice the finding $a$ takes time $O(n\log^2 q)$ (see \cref{subsec: log_q}) and of course that computing $\vec{e} + a \vec{b}_i$ similarly takes time $O(n \log^2 q)$. The loop runs $k$ times, so the total running time is $O(kn \log^2  q)$, as claimed.
    
    To see that the algorithm is correct, first notice that throughout the algorithm we have $\vec{e} - \vec{y} \in \C(\vec{B})$, since we only add multiples of basis vectors $\vec{b}_i$ to $\vec{e}$. Therefore, all that remains to show correctness is proving that $\vec{e} \in \mathcal{F}(\vec{B}^+)$. Notice that step $i$ of the for loop enforces that for any $a \in \F_q^*$,
    $|\pi_{\vec{b}_i^+}(\vec{e})| + \mathrm{TB}_{\vec{b}_i^+}(\vec{e})
    < |\pi_{\vec{b}_i^+}(a \vec{b}_i+\vec{e})| + \mathrm{TB}_{\vec{b}_i^+}(a \vec{b}_i+\vec{e})$. Furthermore, this constraint is maintained by subsequent loop iterations since for $j < i$, $\pi_{\vec{b}_i^+}(\vec{b}_j) = \vec{0}$, so that adding multiples of $\vec{b}_j$ to $\vec{e}$ does not change $|\pi_{\vec{b}_i^+}(\vec{e})| + \mathrm{TB}_{\vec{b}_i^+}(\vec{e})$. It follows that when the algorithm outputs $\vec{e}$, all such constraints are satisfied, i.e., $\vec{e} \in \mathcal{F}(\basis^+)$, as claimed.
\end{proof}

\subsection{Analysis of \texorpdfstring{$\mathcal{F}(\vec{B}^+)$}{F(B+)}}
\label{subsec: F_analysis}

Here, we derive an expression (which can be computed in polynomial time when $q$ is constant) for the weight distribution of the fundamental domain and give intuition as to why a more balanced epipodal profile produces a better fundamental domain. First, for any set $S \subseteq \F_q^n$, we define
$$W_p(S) = \{ \vec{y} \in S : |\vec{y}| = p\}
\; ,$$
and we then define the \emph{weight distribution} of $S$ as
$$W(S) := (|W_0(S)|, |W_1(S)|, \dots, |W_{n}(S)|)
\; .$$

For constant $q$, we will derive a polynomial-time computable expression for $W(\mathcal{F}(\vec{B}^+))$, which one can view as a representation of the performance of size reduction with basis $\vec{B}$. E.g., if we run size reduction on a random target $\vec{y} \sim \F_q^n$, then the probability that the resulting output will have length $p$ is precisely $|W_p(S)|/q^{n-k}$.

To understand the weight distribution, first notice that we may permute the coordinates of $\vec{B}^+$ without changing the weight distribution $W(\mathcal{F}(\vec{B}^+))$. So without loss of generality, we can assume that $\mathbf{B}$ is ``upper triangular'' in the sense that $B_{i,j}  = 0$ for all $j > \ell_1 + \cdots + \ell_i$. Next, we show that we can define an isometry that will allow us to assume that all entries in $\vec{B}^+$ are either $1$ or $0$.

\begin{lemma}
    \label{lem: isometry}
    Let $\vec{c}$ be any non-zero vector in $\F_q^n$ and define $\phi(\vec{x}) := \vec{c} \odot \vec{x}$ and $\phi((\vec{b}_1; \dots; \vec{b}_k)) := (\phi(\vec{b}_1); \dots; \phi(\vec{b}_k))$. Then $\vec{y} \in \mathcal{F}(\vec{B}^+)$ if and only if $\phi(\vec{y}) \in \mathcal{F}(\phi(\vec{B})^+)$
\end{lemma}
\begin{proof}
    Notice that $\phi(\vec{B})^+ = (\vec{c} \odot \vec{b}_1^+; \dots; \vec{c} \odot \vec{b}_k^+) = (\phi(\vec{b}_1^+), \dots, \phi(\vec{b}_k^+))$. Furthermore, as we observed above, our choice of $\mathrm{TB}$ has the convenient property that $\mathrm{TB}_{\vec{b}_i^+}(\vec{x}) = \mathrm{TB}_{\phi(\vec{b}_i^+)}(\phi(\vec{x}))$. Furthermore, notice that $|\pi_{\vec{b}_i^+}(\vec{x})| = |\pi_{\phi(\vec{b}_i^+)}(\phi(\vec{x}))|$.
    Using these facts plus the linearity of $\phi$, we see that $|\pi_{\vec{b}_i^+}(\vec{y})| + \mathrm{TB}_{\vec{b}_i^+}(\vec{y}) < |\pi_{\vec{b}_i^+}(a \vec{b}_i + \vec{y})| + \mathrm{TB}_{\vec{b}_i^+}(a\vec{b}_i + \vec{y})$ if and only if $|\pi_{\phi(\vec{b}_i^+)}(\phi(\vec{y}))| + \mathrm{TB}_{\phi(\vec{b}_i^+)}(\phi(\vec{y})) < |\pi_{\phi(\vec{b}_i^+)}(a \phi(\vec{b}_i) + \phi(\vec{y}))| + \mathrm{TB}_{\phi(\vec{b}_i^+)}(a\phi(\vec{b}_i) + \phi(\vec{y}))$ for any scalar $a \in \F_q$ and $i \in [k]$. In particular, $\vec{y}$ is size reduced with respect to $\vec{B}$ if and only if $\phi(\vec{y})$ is size reduced with respect to $\phi(\vec{B})$, as claimed.
\end{proof}

We assume without loss of generality that the basis $\vec{B}$ has full support. Let $c_i$ denote the inverse of the first non-zero element in the $i$th column of $\vec{B}$, i.e., $c_i = B_{j, i}^{-1}$ where $j$ is minimal such that $B_{j,i} \neq 0$. Let $\vec{c} = (c_1, \dots, c_n)$, $\phi(\vec{x}) = \vec{c} \odot \vec{x}$ and $\phi(\vec{B}) = (\phi(\vec{b}_1), \dots, \phi(\vec{b}_k))$. \cref{lem: isometry} and the fact that $\phi$ is an isometry imply that $W(\mathcal{F}(\vec{B}^+)) = W(\mathcal{F}(\phi(\vec{B})^+))$. Furthermore, $\phi(\vec{B})^+$ consists only of zeros and ones. Therefore, by swapping $\vec{B}$ with $\phi(\vec{B})$, we may assume without loss of generality that all entries in $\vec{B}^+$ are either zero or one. 

So, we have reduced to the case when $\vec{B}^+$ consists of zeros and ones, with $B_{i,j} = 1$ if and only if $\ell_1 + \cdots + \ell_{i-1} < j \leq \ell_1 + \cdots + \ell_i$.
Now, let $Y_i$ denote the set of all vectors $\vec{y}_i \in \F_q^{\ell_i}$ such that (1) $N_{0}(\vec{y}_i) \geq N_a(\vec{y}_i)$ for all $a \in \F_q$ where $N_a$ is the number of coordinates that take value $a$; and (2) $y_{i,1} \in \{0,\ldots, q-1\}$ is minimal among the $y_{i,1} - a$ where $a$ is such that $N_a(\vec{y}_i) = N_0(\vec{y}_i)$. Then
\begin{equation}
    \label{eq: product}
    \mathcal{F}(\vec{B}^+) = \{ (\vec{y}_1, \dots, \vec{y}_k) : \vec{y}_1 \in Y_1, \dots, \vec{y}_k \in Y_k \} = Y_1 \times \cdots \times Y_k
    \; .
\end{equation}

We therefore wish to study $W(Y_i)$. To that end, let us consider the number of vectors $\vec{y}_i \in \F_q^{\ell_i}$ such that $z = N_0(\vec{y})$, $|\{a \ : \ N_a(\vec{y}_i) = z\}| = s$, and $N_a(\vec{y}_i) \leq z$ for all $a \in \F_q$. (In other words, $\vec{y}_i$ is such that zero is the most common symbol in $\vec{y}_i$, which occurs $z$ times, and there are precisely $s$ total symbols that occur $z$ times in $\vec{y}_i$.)
To count the number of such $\vec{y}_i$, we see that each $\vec{y}_i$ occurs uniquely in the following process. We first choose $s-1$ non-zero symbols $a_1, \dots, a_{s-1}$ from $\F_q^*$ (corresponding to the non-zero elements $a$ with $N_a(\vec{y}) = z$), then choose the $z$ coordinates in $\vec{y}_i$ that will contain zeros, then the $z$ coordinates that will take value $a_1$, then $a_2$, etc. Finally, we choose the values for the remaining $\ell_i-sz$ coordinates of $\vec{y}_i$ from $\F_q \setminus \{ 0, a_1, \dots, a_{s-1} \}$ such that each of those letters occurs at most $z-1$ times. Therefore, the total number of such $\vec{y}_i$ is precisely 
\[
    \binom{q-1}{s-1} \cdot \left(\prod_{t=0}^{s-1} \binom{\ell_i - zt}{z} \right) \cdot A(\ell_i-sz,q-s,z-1)
    \; ,
\]
where
\[A(n, q, m) := |\{\vec{x} \in \{1,\ldots, q\}^n \ :\ \forall c \in \{1,\ldots, q\}, N_c(\vec{x}) \leq m\}|
\; .
\]
Furthermore, notice that only a $1/s$ fraction of such $\vec{y}_i$ are in $Y_i$ (those for which $y_1 - a_i$ is minimized over all of the ``frequent symbols'' $a_i$ when $a_i = 0$).
Summing over all $s$, we see that
\begin{equation}
\label{eq: weight_distribution}
|W_{\ell_i-z}(Y_i)| = \sum_{s=1}^{\ell_i} \frac{1}{s} \binom{q-1} {s-1} \cdot \left( \prod_{t=0}^{s-1} \binom{\ell_i-zt}{z} \right) \cdot  A(\ell_i-sz, q-s, z-1)
\; ,
\end{equation}
where we adopt the common convention that $\binom{n}{k} = 0$ for $k > n$. 

For fixed $q$, $A(n, q, m)$ is computable in $\poly(n)$ time~\cite{1995252}, so that this gives us a way to compute $W(Y_i)$ in polynomial time in this case. Finally, by \cref{eq: product}, a simple convolution allows us to compute $W(\mathcal{F}(\vec{B}^+))$, as
$$W(\mathcal{F}(\vec{B}^+)) = W(Y_1) \ast \dots \ast W(Y_k)
\; .$$
Thus, $W(\mathcal{F}(\vec{B}^+))$ is computable in $\poly(n)$ time for fixed field size $q \leq O(1)$.

\section{Proper bases and primitivity}
\label{sec:proper}
We will primarily be interested in bases that are \emph{proper} in the sense that the epipodal vectors should all be non-zero.

\begin{definition}
    A basis is said to be \emph{proper} if all its epipodal vectors $\vec{b}_i^+$ are non-zero.
\end{definition}

\cite{DDvAlgorithmicReductionTheory2022} observed that, given an arbitrary basis $\basis \in \F_q^{k \times n}$ for a code, we can efficiently compute a proper basis $\basis'$ for the same code by systematizing $\vec{B}$. In particular, let $\vec{A} \in \F_q^{k \times k}$ be an invertible matrix formed from $k$ columns of $\basis$ (which must exist because $\basis$ is a full-rank matrix). Then, $\basis' := \vec{A}^{-1} \basis$ is a proper basis for the code generated by $\basis$. So, every code has a proper basis. From this, we derive the following simple but useful fact.

\begin{fact}
    \label{fact:getting_less_support}
    For any code $\mathcal{C}$ with dimension $k  \geq 1$, there is a subcode $\mathcal{C}' \subset \mathcal{C}$ with dimension $k-1$ with strictly smaller support, i.e., $\supp(\mathcal{C}') \subsetneq \supp(\mathcal{C})$.
\end{fact}
\begin{proof}
    Let $(\vec{b}_1;\ldots; \vec{b}_k)$ be any proper basis for $\mathcal{C}$. (As we observed above, such a basis exists.) Then, take $\mathcal{C}' \subset \mathcal{C}$ to be the code with basis $(\vec{b}_1;\ldots; \vec{b}_{k-1})$. It is immediate that  the dimension of $\mathcal{C}'$ is $k-1$, and the support of $\mathcal{C}'$ is $\supp(\mathcal{C}) \setminus \supp(\vec{b}_k^+)$, which is a strict subset of $\supp(\mathcal{C})$ because $\vec{b}_k^+$ is non-zero.
\end{proof}

We now introduce the notions of primitive codewords and primitive subcodes and show that they are closely related to proper bases.

\begin{definition}
    A non-zero codeword $\vec{c} \in \C$ is \emph{non-primitive in the code $\C$} if there exists a non-zero codeword $\vec{c}' \in \C$ such that $\supp(\vec{c}') \subsetneq \supp(\vec{c})$.\footnote{For example, the vector $\vec{c}:= (1,1)$ is non-primitive in the code $\F_q^2$ because the codeword $\vec{c}' = (1,0)$ is non-zero with support strictly contained in $\supp(\vec{c})$.} Otherwise, $\vec{c}$ is \emph{primitive in $\C$}. ($\vec0$ is \emph{not} primitive.)

    A non-zero subcode $\C' \subseteq \C$ is \emph{non-primitive in the code $\C$} if there exists a subcode $\C''$ with the same dimension as $\C'$ but $\supp(\C'') \subsetneq \supp(\C')$. Otherwise, $\C'$ is \emph{primitive in $\C$}. (The trivial subcode $\{\vec0\}$ is \emph{not} primitive.)
\end{definition}

When the code $\C$ is clear from context, we often simply say that $\vec{c}$ is primitive or $\C'$ is primitive, without explicitly mentioning the code $\C$.

We first observe a basic equivalent definition of a primitive subcode. (This compares nicely with the notion of a primitive \emph{sublattice}, which is the (non-zero) intersection of the lattice with a subspace.)

\begin{fact}
    \label{fact:primitive_is_intersection}
    A non-zero subcode $\C' \subseteq \C \subseteq \F_q^n$ is a primitive subcode of $\C$ if and only if there exists a set $S \subseteq [n]$ of coordinates such that $\C' = \{\vec{c} \in \C \ : \ \forall i \in S,\ c_i = 0\}$.
\end{fact}
\begin{proof}
    First, suppose $\C'$ is primitive. Then, we take $S := [n] \setminus \supp(\C')$ to be the complement of the support of $\C'$. Let $\C^\dagger := \{\vec{c} \in \C \ : \ \forall i \in S,\ c_i = 0\}$. It suffices to show that $\C^\dagger = \C'$. Indeed, clearly $\C' \subseteq \C^\dagger$, so since they are both vector spaces, it suffices to show that the dimension of $\C^\dagger$ is no larger than the dimension of $\C'$. To that end, let $m := \dim(\C^\dagger) - \dim(C')$. By applying \cref{fact:getting_less_support} $m$ times, we see that there is a subcode $\C'' \subseteq \C^\dagger$ with dimension $\dim(\C^\dagger)-m = \dim(\C)$ and $|\supp(\C'')| \leq |\supp(\C')| - m$. Since $\C'$ is primitive and $\supp(\C'') \subseteq \supp(\C')$ by definition, it follows that we must have $m = 0$. Therefore, $\C' = \C^\dagger$, as needed.

    Now, suppose that $\C' =  \{\vec{c} \in \C \ : \ \forall i \in S,\ c_i = 0\}$ for some set $S$. Then, notice that by definition any subcode $\C'' \subseteq \C$ of $\C$ with $\supp(\C'') \subseteq \supp(\C')$ must in fact be a subcode of $\C'$, i.e., $\C'' \subseteq \C'$. Since $\C'$ and $\C''$ are both vector spaces, either $\C'' = \C'$ or $\C''$ has lower dimension than $\C'$. It follows immediately that $\C'$ is primitive.
\end{proof}

We now observe some basic properties of primitivity.

\begin{fact}
    \label{fact: primitive}
    The following are all equivalent for a non-zero codeword $\vec{c} \in \C$ in a code $\C$ with dimension $k$.
    \begin{enumerate}
        \item \label{item:primitive_codeword} $\vec{c}$ is primitive in $\C$.
        \item \label{item:primitive_rank_one_subcode} The one-dimensional subcode $\C' := \{z \vec{c} \ : \ z \in \F_q\}$ is primitive in $\C$.
        \item \label{item:primitive_non-zero_projections} We have $\pi_{\vec{c}}^\perp(\vec{c}') = \vec0$ for $\vec{c}' \in \C$ if and only if $\vec{c}' = z \vec{c}$ for some $z \in \F_q$.
        \item \label{item:primitive_in_subcode}  $\vec{c}$ is primitive in some primitive subcode $\mathcal{C}'$ of $\mathcal{C}$.
        \item \label{item:primitive_shortest_in_primitive_subcode} There is some primitive subcode $\C'$ that contains $\vec{c}$ and contains no strictly shorter non-zero vectors.
        \item \label{item:primitive_full_rank_projection} $\pi_{\vec{c}}^\perp(\C)$ is a $(k-1)$-dimensional code.
        \item \label{item:primitive_contained_in_proper_basis} There exists a proper basis $(\vec{b}_1;\ldots; \vec{b}_k) \in \F_q^{k \times n}$ of $\C$ with $\vec{b}_1 = \vec{c}$.
    \end{enumerate}
\end{fact}
\begin{proof}
    The equivalence of \cref{item:primitive_codeword,item:primitive_rank_one_subcode} follows by noting that the support of a rank-one subcode $\{z \vec{c} \ : \ z \in \F_q\}$ is precisely the support of its generator $\vec{c}$.

    To see that \cref{item:primitive_non-zero_projections} is equivalent to \cref{item:primitive_codeword}, first notice that \cref{item:primitive_non-zero_projections} immediately implies \cref{item:primitive_codeword} (since any vector with support strictly contained in $\supp(\vec{c})$ must project to zero). To show that \cref{item:primitive_codeword} implies \cref{item:primitive_non-zero_projections}, we assume that there exists some $\vec{c}' \in \C$ with $\vec{c}' \notin \{z \vec{c} \ : \ z \in \F_q\}$ such that $\pi^\perp_{\vec{c}}(\vec{c}') = \vec0$ and show that this implies that $\vec{c}$ is non-primitive. Notice that $\supp(\vec{c}') \subseteq \supp(\vec{c})$  because it projects to zero. Since $\vec{c}$ and $\vec{c}'$ are non-zero by assumption, it follows that there must be some $i \in \supp(\vec{c}) \cap \supp(\vec{c}')$. Then, $\vec{c}'' := c_i' \vec{c} - c_i \vec{c}$ is a codeword with $\supp(\vec{c}'') \subseteq \supp(\vec{c}) \setminus \{i\} \subsetneq \supp(\vec{c})$. Furthermore, because $c_i$ and $c_i'$ are non-zero and $\vec{c}$ and $\vec{c}'$ are linearly independent, $\vec{c}''$ is not zero. In other words, $\vec{c}$ is not primitive, as needed.

    It is immediate that \cref{item:primitive_rank_one_subcode} implies \cref{item:primitive_in_subcode,item:primitive_shortest_in_primitive_subcode}. Furthermore, it is clear that \cref{item:primitive_shortest_in_primitive_subcode} implies \cref{item:primitive_in_subcode} (since a vector with minimal hamming weight must be primitive). 
    
    So, to show equivalence of \cref{item:primitive_in_subcode,item:primitive_shortest_in_primitive_subcode} with \cref{item:primitive_codeword,item:primitive_rank_one_subcode,item:primitive_non-zero_projections}, it suffices to show that \cref{item:primitive_in_subcode} implies \cref{item:primitive_codeword}. So, we suppose that $\vec{c}$ is primitive in some subcode $\mathcal{C}' \subseteq \C$ but $\vec{c}$ is \emph{not} primitive in $\mathcal{C}$ and we use this to show that $\mathcal{C}'$ is not primitive. By assumption, there exists non-zero $\vec{c}' \in \mathcal{C}$ whose support is strictly contained in the support of $\vec{c}$. Since $\vec{c}$ \emph{is} primitive in $\mathcal{C}'$, it must be the case that $\vec{c}' \notin \mathcal{C}'$. But, by \cref{fact:primitive_is_intersection}, this means that $\mathcal{C}'$ cannot be primitive, since a primitive code must contain all elements in its support.
    
    To see that \cref{item:primitive_full_rank_projection} is equivalent to \cref{item:primitive_non-zero_projections} (and therefore to \cref{item:primitive_codeword,item:primitive_rank_one_subcode,item:primitive_in_subcode,item:primitive_shortest_in_primitive_subcode}), let $\pi^* : \C \to \F_q^n$ be the restriction of the map $\pi_{\vec{c}}^\perp$ to $\C$. Let $K$ be the kernel of $\pi^*$ and $S$ be the image of $\pi^*$, and recall from basic linear algebra that $\dim(K) + \dim(S) = \dim(C) = k$. In particular, $\dim(S) = k-1$ if and only if $\dim(K) = 1$. In other words, \cref{item:primitive_full_rank_projection} holds if and only if \cref{item:primitive_non-zero_projections} holds, as needed.

    Finally, to see that \cref{item:primitive_contained_in_proper_basis} is equivalent to \cref{item:primitive_full_rank_projection} (and therefore to all other items), first notice that if \cref{item:primitive_full_rank_projection} does not hold, then \cref{item:primitive_contained_in_proper_basis} \emph{must} be false, since any proper basis $(\vec{b}_1;\ldots; \vec{b}_k)$ must in particular have the property that $\pi_{\vec{b}_1}^\perp(\vec{b}_2),\ldots, \pi_{\vec{b}_1}^\perp(\vec{b}_k)$ are linearly independent. 
    It therefore remains to show that if $\pi_{\vec{c}}^\perp(\C)$ has dimension $k-1$, then there is a proper basis $(\vec{b}_1;\ldots; \vec{b}_k)$ with $\vec{b}_1 = \vec{c}$. To that end, let $\basis' := (\vec{b}_2';\ldots; \vec{b}_k') \in \F_q^{(k-1) \times n}$ be a proper basis of the projected code $\pi_{\vec{c}}^\perp(\C)$. (Such vectors must exist because $\pi_{\vec{c}}^\perp(\C)$ is a $(k-1)$-dimensional code, and all codes have a proper basis as we observed above.) Let $\vec{b}_1 := \vec{c}$ and for $i \geq 2$, let $\vec{b}_i \in \C$ be any codeword with $\pi_{\vec{c}}^\perp(\vec{b}_i) = \vec{b}_i'$. Clearly, $\vec{b}_1^+ = \vec{c}$ is non-zero, and for $i \geq 2$, $\vec{b}_i^+ = (\vec{b}_i')^+$, which is non-zero by the assumption that $\basis'$ is proper. It follows that $(\vec{b}_1;\ldots; \vec{b}_k)$ is a proper basis, as needed.
\end{proof}

Next, we show how to characterize proper bases in terms of their blocks and the primitivity of the epipodal vectors.

\begin{fact}
    \label{fact:primitive_and_proper}
    For a basis $\basis := (\vec{b}_1;\ldots; \vec{b}_k) \in \F_q^{k \times n}$ for a code $\mathcal{C}$, the following are equivalent.
    \begin{enumerate}
        \item \label{item:basis_proper} $\basis$ is proper.
        \item \label{item:some_projections_proper}  There exist indices $i_j$ with $1 = i_1 < \cdots < i_\ell = k$ such that the bases $\vec{B}_{[i_1,i_2]}, \vec{B}_{[i_2+1,i_3]},\ldots, \vec{B}_{[i_{\ell-1}+1,i_\ell]}$ are proper.
            \item \label{item:projections_proper} The bases $\vec{B}_{[i,j]}$ are proper for all $i,j$.
        \item \label{item:epipodal_primitive} $\vec{b}_i^+$ is a primitive vector in the projected code $\pi_i(\mathcal{C})$ for all $i$.
        \item \label{item:primitive_first_guy_proper_projection} $\vec{b}_1$ is a primitive vector in $\C$ and $\basis_{[2,k]}$ is a proper basis.
    \end{enumerate}
\end{fact}
\begin{proof}
    Notice that the epipodal vectors of $\vec{B}_{[i,j]}$ are exactly $\vec{b}_i^+, \ldots, \vec{b}_j^+$. It follows immediately that \cref{item:basis_proper,item:some_projections_proper,item:projections_proper} are equivalent.

    To see that \cref{item:projections_proper} implies \cref{item:epipodal_primitive}, notice that the equivalence of \cref{item:primitive_codeword,item:primitive_contained_in_proper_basis} of \cref{fact: primitive} means that, since $\basis_{[i,k]}$ is a proper basis, $\vec{b}_i^+$ must be primitive in the code $\pi_i(\mathcal{C})$ generated by $\basis_{[i,k]}$.
    The fact that \cref{item:epipodal_primitive} implies \cref{item:basis_proper} follows immediately from the fact that primitive vectors are non-zero. So, \cref{item:basis_proper,item:some_projections_proper,item:projections_proper,item:epipodal_primitive} are all equivalent.

    Finally, to see that \cref{item:primitive_first_guy_proper_projection} is equivalent to, say, \cref{item:basis_proper}, first notice that if $\basis$ is proper, we know that this implies that \cref{item:epipodal_primitive} holds, and in particular that $\vec{b}_1$ is primitive. And, properness of $\vec{B}$ this implies that \cref{item:projections_proper} holds, which in particular means that $\vec{B}_{[2,k]}$ is proper. So, \cref{item:basis_proper} implies \cref{item:primitive_first_guy_proper_projection}. On the other hand, if $\vec{b}_1$ is primitive and $\vec{B}_{[2,k]}$ is proper, then the properness of $\vec{B}_{[2,k]}$ implies that $\vec{b}_i^+$ is primitive in $\pi_i(\C)$ for all $i \geq 2$. Since \cref{item:epipodal_primitive} implies properness, it follows that $\vec{B}$ itself is proper.
\end{proof}

\subsection{Algorithms for manipulating proper bases}

We now give two algorithms for manipulating proper bases. The first algorithm converts an arbitrary non-zero codeword into a primitive codeword that is at least as short (in fact, a non-zero codeword whose support is a subset of the original codeword's support). The second algorithm creates a \emph{proper} basis for a code whose first vector is a given primitive codeword. Notice that together these two algorithms allow us to create a proper basis for a code whose first vector is at least as short as some given codeword.

We will need a subprocedure guaranteed by the following claim. (The claim would be trivial if we settled for time $O(n^2 k \log^2 q)$. We work a bit harder to obtain slightly better running time.)

\begin{claim}
    \label{clm:special_subcode}
    There is an algorithm that takes as input a basis $\basis \in \F_q^{k \times n}$ for a code $\C$ and a set of coordinates $S \subseteq [n]$ and outputs a basis for the subcode $\C' := \{\vec{c} \in \C \ : \ \forall i \in S, c_i = 0 \}$ in time $O(nk^2 \log^2 q)$.
\end{claim}
\begin{proof}
    We first find a maximal linearly independent set of columns of $\vec{B}$ in $[n] \setminus S$ (i.e., a maximal linearly independent set of columns that are \emph{not} in $S$). This can be done by building a list of linearly independent columns greedily and maintaining a basis in systematic form for the span of the columns, which allows us to check whether a new column should be added and add a new column in time $O(k^2 \log^2 q)$. We assume without loss of generality that these are the first $m$ columns of $\vec{B}$. (Since we restrict our attention to columns outside of $S$, we might not have $m = k$.) We then effectively systematize with respect to these columns, i.e., we perform elementary row operations to $\vec{B}$ to obtain 
    \[
        \vec{B}' := \begin{pmatrix} \vec{I}_m & \vec{A}_1\\ \vec{0}_{(k-m) \times m} &\vec{A}_2
        \end{pmatrix}
        \; 
    \]
    for some $\vec{A}_1, \vec{A}_2$.
    Finally, we output the last $k-m$ rows of $\vec{B}'$.

    Clearly this runs in time $O(nk^2 \log^2 q)$, as claimed. To see that the resulting output is in fact a basis for $\C'$, first notice that since $\vec{B}'$ was obtained by applying elementary row operations to $\vec{B}$, $\vec{B}'$ must also be a basis for $\C$. Then, notice that if $\vec{z} \vec{B}' \in \C'$, we clearly must have that $z_1 = \cdots = z_m = 0$. So, all elements in $\C'$ can be written as a linear combination of the last $k-m$ rows of $\vec{B}'$. On the other hand, the last $k-m$ rows of $\vec{B}'$ must be linearly independent (since $\vec{B}'$ is a basis), and because we chose a maximal linearly independent set of columns in $[n] \setminus S$, all of these rows must be zero outside of the support of $\vec{c}$ (since all of their entries outside of the support of $\vec{c}$ are linear combinations of their first $m$ entries, which are zero). In other words, these rows form a basis for $\C'$, as needed.
\end{proof}

We now present in \cref{alg: make_primitive} a simple algorithm for finding a primitive codeword ``in the support of'' a given non-zero codeword.

\RestyleAlgo{ruled}
\begin{algorithm2e}
\KwIn{A basis $\vec{B} \in \F_q^{k \times n}$ for $\C$ and a non-zero vector $\vec{c} \in \C_{\neq \vec0}$}
\KwOut{A primitive codeword $\vec{p} \in \C$ with $\supp(\vec{p}) \subseteq \supp(\vec{c})$.}
\caption{Make Primitive}
\label{alg: make_primitive}
Let $\C' := \{\vec{c}' \in \C \ : \ \supp(\vec{c}') \subseteq \supp(\vec{c}) \}$.\ 

Find a basis $\vec{B}'$ of the code $\C'$.\ 

Systematize $\vec{B}'$.\ 

\Return the first row of $\vec{B}'$.\ 
\end{algorithm2e}

\begin{claim}
    \label{clm:make_primitive}
    On input a basis $\vec{B} \in \F_q^{k \times n}$ for a code $\C$ and a non-zero codeword $\vec{c}$, \cref{alg: make_primitive} always outputs a primitive codeword $\vec{p} \in \C$ with $\supp(\vec{p}) \subseteq \supp(\vec{c})$, and it can be implemented to run in $O(nk^2 \log^2 q)$ time.
\end{claim}
\begin{proof}
    Since the output vector $\vec{p} := \vec{b}_1'$ is the first element in a basis $\vec{B}'$ for $\C'$, we must of course have that $\vec{p} \in \C'$, i.e., that $\vec{p} \in \C$ and $\supp(\vec{p}) \subseteq \supp(\vec{c})$. Furthermore, since $\vec{B}'$ is a proper basis, it follows from \cref{fact: primitive} that $\vec{p}$ is primitive in $\C'$. By \cref{fact: primitive} again, to prove that $\vec{p}$ is primitive \emph{in $\C$}, it suffices to show that $\C'$ is a primitive subcode of $\C$. But, this follows immediately from \cref{fact:primitive_is_intersection} and the definition of $\C'$.

    So, the algorithm is correct. To see that the running time is as claimed, recall from \cref{clm:special_subcode} that a basis for $\C'$ can be found in time $O(n k^2 \log^2 q)$. 
    Furthermore, the final systematization step also takes $O(n k^2 \log^2 q)$ time, so that the full algorithm runs in time $O(n k^2 \log^2 q)$, as needed.
\end{proof}

We now show in \cref{alg: insert_primitive} a general algorithm that, when given a basis $\vec{B}$ for some code $\C$ and a primitive codeword $\vec{p}$ of $\C$, outputs an invertible linear transformation $\vec{A}$ such that $\vec{A}\vec{B}$ is a proper basis with $\vec{p}$ as its first vector. This will give us a way to put whichever primitive codeword we like as the first vector of a basis, which will be a useful subroutine in many of our algorithms.

\RestyleAlgo{ruled}
\begin{algorithm2e}
\KwIn{A basis $\vec{B} = (\vec{b}_1; \dots; \vec{b}_k) \in \F_q^{k \times n}$ in $\C$ and a primitive vector $\vec{p}$ for $\C$}
\KwOut{An invertible matrix $\vec{A} \in \F_q^{k \times k}$ such that $\vec{A} \vec{B} = (\vec{b}'_1; \dots; \vec{b}'_k)$ is proper and $\vec{b}'_1 = \vec{p}$.}
\caption{Insert Primitive}
\label{alg: insert_primitive}

Find $a_1, \dots a_k \in \F_q$ such that $a_1 \vec{b}_1 + \dots + a_k \vec{b}_k = \vec{p}$\

$m \leftarrow \min(\{ r \in [i, k] : a_r \neq 0 \})$\

$\vec{b}_m \leftarrow {a_1 \vec{b}_{1} + \dots + a_k \vec{b}_{k}}$\
    
Swap $\vec{b}_1$ and $\vec{b}_{m}$.\

Let $\vec{T} \in \F_q^{(k-1) \times (k-1)}$ be any matrix such that $\vec{T} \vec{B}_{[2, k]}$ is in systematic form.\ 

$\vec{B} \leftarrow (\vec{I}_{1} \oplus \vec{T}) \vec{B}$\

\Return the matrix $\vec{A}$ corresponding to linear transformation done to $\vec{B}$
\end{algorithm2e}

\begin{claim}
    \label{claim: insert_primitive}
    On input a basis $\basis \in \F_q^{k \times n}$ for a code $\C$ and a primitive vector $\vec{p}$ in $\C$, \cref{alg: insert_primitive} correctly outputs a matrix $\vec{A}$ so that $\vec{A} \vec{B}$ is a proper basis of $\C$ whose first vector is $\vec{p}$. Furthermore, \cref{alg: insert_primitive} runs in time $O(n k^2 \log^2 q)$.
\end{claim}
\begin{proof}
    We first observe that the running time is as claimed. Indeed, the first step asks us to find a system of linear equations over $\F_q$, which can be done in time $O(nk^2 \log^2 q)$, finding $\vec{T}$ takes time $O(n k^2 \log^2 q)$, and multiplying by $\vec{T}$ similarly takes $O(nk^2 \log^2 q)$ time. The remaining steps take significantly less time.

    We now prove correctness. Let $\vec{B}$ correspond to the value of the input basis, and let $\vec{B}' := (\vec{b}'_1, \dots, \vec{b}'_k)$ correspond to the value of of the basis at the end of the algorithm. Since all operations on $\vec{B}$ are either elementary row operations or multiplication by an invertible matrix, \cref{alg: insert_primitive} can in fact return an invertible matrix $\vec{A} \in \F_q^{k \times k}$ such that $\vec{B}' = \vec{A}\vec{B}$. It follows that $\vec{B}'$ is a basis for $\C$ as claimed.
    
    Furthermore, clearly, $\vec{b}'_1 = \vec{p}$ as claimed. Finally, notice that by \cref{fact: primitive}, there exists a $\vec{T}$ with the desired properties. In particular, since $\vec{p}$ is primitive, the projection $\pi_{\{\vec{p}\}}^\perp(\C)$ orthogonal to $\vec{p}$ is a $k-1$ dimensional code. Finally, we note that $\vec{B}'$ is in fact proper by \cref{fact:primitive_and_proper}, since $\vec{b}_1' = \vec{p}$ is primitive and $(\vec{B}')_{[2,k]}$ is proper by assumption.
\end{proof}

\section{Redundant sets of coordinates, the last epipodal vector, and backward reduction}
\label{sec:redundant}
We are now ready to develop the theory behind backward-reduced bases. A backward-reduced basis is one in which the last epipodal vector $\vec{b}_k^+$ is as \emph{long} as possible. In the context of lattices, such bases are called dual-reduced bases and the maximal length is simply the reciprocal $1/\lambda_1(\lat^*)$ for the length of the shortest non-zero vector in the dual lattice $\lat^*$. 

For codes, the maximal length of the last epipodal vector behaves rather differently, as we will explain below. In particular, we will see how to find a backward-reduced basis quite efficiently. In contrast, finding a dual-reduced basis is equivalent to finding a shortest non-zero vector in a lattice and is therefore NP-hard.

On our way to defining backward reduction, we first define the notion of \emph{redundant} coordinates. Notice that we only consider coordinates in the support of $\C$.

\begin{definition}
    For a code $\C \subseteq \F_q^n$, we say that a set $S \subseteq [n]$ of coordinates is \emph{redundant for $\C$} if $S \subseteq \supp(\C)$ and for every $\vec{c} \in \C$ and all $i,j \in S$, $c_i = 0$ if and only if $c_j = 0$.
\end{definition}

The following simple claim explains the name ``redundant.'' In particular, for any codeword $\vec{c} \in \C$, if we know $c_i$ for some $i \in S$, then we also know $c_j$ for any $j \in S$.

\begin{claim}
    \label{claim:redundant}
    For a code $\C \subseteq \F_q^n$, a set $S \subseteq \supp(\C)$ is redundant for $\C$ if and only if for every $i,j \in S$, there exists a non-zero scalar $a \in \F_q^*$ such that for all $\vec{c} \in \C$, $c_j = a c_i$. 
    
    Furthermore, to determine whether $S$ is a set of redundant coordinates, it suffices to check whether the latter property holds for all $\vec{c} := \vec{b}_i$ in a basis $(\vec{b}_1;\ldots; \vec{b}_k)$ of $\C$.
\end{claim}
\begin{proof}
        It is clear that if for every $i,j \in S$ there exists such a non-zero scalar, then $S$ must be redundant. So, we need only prove the other direction.
    
        To that end, suppose that $S$ is redundant. Let $i,j \in S$ and let $\vec{c} \in \C$ be a codeword with $c_i \neq 0$. (Such a codeword exists because $S$ is a subset of the support.) Let $a := c_i^{-1} c_j$. We wish to show that for any $\vec{c}' \in \C$, $c_j' = a c_i'$. Indeed, notice that $c_i' c_i^{-1} \vec{c} - \vec{c}'$ is a codeword whose $i$th coordinate is zero. It follows from the redundancy of $S$ that the $j$th coordinate must also be zero, i.e., that $c_i' c_i^{-1} c_j - c_j' = 0$. Rearranging and recalling the definition of $a$, we see that $c_j' = a c_i'$, as needed.
    
        The ``furthermore'' follows from the fact that the condition $c_j = a c_i$ is clearly preserved by linear combinations.
\end{proof}

Next, we show that redundancy is closely connected with the last epipodal vector in a basis.

\begin{lemma}
    \label{lem:last_epipodal_redundant}
    For a code $\C \subseteq \F_q^n$ with dimension $k$ and $S \subseteq [n]$, there exists a basis $\basis := (\vec{b}_1;\ldots;\vec{b}_k)$ of $\C$ with $S \subseteq \supp(\vec{b}_k^+)$ if and only if $S$ is redundant.
\end{lemma}
\begin{proof}
        First, suppose that there exists a basis $\basis := (\vec{b}_1;\ldots;\vec{b}_k)$ of $\C$ with $S \subseteq \supp(\vec{b}_k^+)$. Then, notice that for any codeword $\vec{c} := a_1 \vec{b}_1 + \cdots + a_k \vec{b}_k$ and any $i,j \in S$, we must have that $c_i = 0$ if and only if $a_k = 0$ (since $\vec{b}_k$ is the only basis vector whose support contains $i$), which happens if and only if $c_j = 0$ (for the same reason). So, $S$ is redundant, as claimed.
    
        Now, suppose that $S$ is redundant. We also suppose $S$ is non-empty, since otherwise the statement is trivial. Then, we will construct a proper basis with the desired property. To that end, consider the code $\C' := \{\vec{c} \in \C \ : \ \forall i \in S,\ c_i = 0\}$. Notice that $\C'$ has dimension exactly $k-1$. Indeed, it has dimension at most $k-1$ because $S \subseteq \supp(\C)$, so $\C'$ must be a strict subcode of $\C$ and must therefore have dimension strictly less than $\C$. But, it has codimension at most one because by redundancy we can write $\C' = \{\vec{c} \in \C \ : \ c_i = 0\}$ for any \emph{fixed} index $i \in S$. So, $\C'$ is the intersection of a $k$-dimensional subspace with a subspace with codimension one and must therefore have dimension at least $k-1$, as needed.
        
        So, let $\vec{b}_1,\ldots, \vec{b}_{k-1}$ be a proper basis for $\C'$, and let $\vec{b}_k \in \C$ be any codeword in $\C$ that is not contained in $\C'$. It follows that $\basis := (\vec{b}_1;\ldots; \vec{b}_k)$ is a basis for $\C$. And, $\vec{b}_1^+,\ldots, \vec{b}_{k-1}^+$ are non-zero by assumption. Finally, notice that the $i$th coordinate of $\vec{b}_k$ must be non-zero for all $i \in S$, since otherwise $\vec{b}_k$ would be in $\C'$. Furthermore, by definition the $i$th coordinates of $\vec{b}_1,\ldots, \vec{b}_{k-1}$ must all be zero for $i \in S$. Therefore, $i \in \supp(\vec{b}_k^+)$, as needed.
\end{proof}

The above motivates the following definition.

\begin{definition}
    For a code $\C \subseteq \F_q^n$, the repetition number of $\C$, written $\eta(\C)$, is the maximal size of a redundant set $S \subseteq \supp(\C)$.
\end{definition}

In particular, notice that by \cref{lem:last_epipodal_redundant}, $\eta(C)$ is also the maximum of $|\vec{b}_k^+|$ over all bases $(\vec{b}_1;\ldots;\vec{b}_k)$ \emph{and} this maximum is achieved by a proper basis. The next lemma gives a lower bound on $\eta(\C)$, therefore showing that codes with sufficiently large support and sufficiently low rank must have bases whose last epipodal vector is long.

\begin{lemma}
    \label{lem:lower_bound_on_max_last_epipodal}
    For any code $\C \subseteq \F_q^n$ with dimension $k$,
    \[
        \eta(\C) \geq \left\lceil \frac{q-1}{q^k-1} \cdot |\supp(\C)|\right\rceil
        \; .
    \]
\end{lemma}
\begin{proof}
        We assume without loss of generality that $|\supp(\C)| = n$. 
        
        Let
        \[
            t := \left\lceil \frac{q-1}{q^k-1} \cdot n \right\rceil 
            \; .
        \]
        Let $\vec{B} \in \F_q^{k \times n}$ be any basis for $\C$, and for every $i \in [n]$, let $a_i \in \F_q^*$ be the first non-zero entry in the $i$th column of $\vec{B}$. (By our assumption that $|\supp(\C)| = n$, such an element must exist.)
    
        We claim that there is some subset $S \subseteq [n]$ with $|S| = t$ such that for all $i, j \in S$, the $i$th column of $\vec{B}$ times $a_i^{-1}$ is equal to the $j$th column of $\vec{B}$ times $a_j^{-1}$. Indeed, notice that the total number of possible $k$-dimensional vectors $\vec{v} \in \F_q^k$ whose first non-zero coordinate is one is $1 + q + q^2 + \cdots + q^{k-1} = (q^k-1)/(q-1)$. Since $\vec{B}$ has $n$ columns, it follows from the pigeonhole principle that such a set $S$ must exist.
    
        Finally, we claim that $S$ is in fact a set of redundant coordinates for $\C$.  By \cref{claim:redundant}, it suffices to observe that for every $i,j \in S$, there is some $a \in \F_q^*$ such that the $j$th coordinate $B_{m,j}$ of the $m$th basis vector $\vec{b}_m$ satisfies $B_{m,j} = a B_{m, i}$ for all $m$. Indeed, we can take $a := a_j a_i^{-1}$.
\end{proof}

We present in \cref{alg: redundant} a simple algorithm that finds the largest redundant set of a code $\C$. (The algorithm itself can be viewed as a constructive version of the proof of \cref{lem:lower_bound_on_max_last_epipodal}.)

\RestyleAlgo{ruled}
\begin{algorithm2e}
\caption{Max Redundant Set}
\label{alg: redundant}
\KwIn{A basis $\vec{B} = (\vec{b}_1; \dots; \vec{b}_k) \in \F_q^{k \times n}$ for $\C$}
\KwOut{A redundant set $S$ for $\C$ with $|S| = \eta(\C)$}

\For{$j \in [n]$}{
    $a_j \gets B_{i,j}^{-1}$, where $i \in [k]$ is minimal such that $B_{i,j} \neq 0$.
}

Find $S \subseteq \supp(\C)$ with maximal size such that for all $j_1,j_2 \in S$ and all $i$, $a_{j_1} B_{i,j_{1}} = a_{j_2}B_{i,j_2}$.

\Return $S$
\end{algorithm2e}

\begin{claim}
    \label{clm:redundant_alg}
    \cref{alg: redundant} outputs a redundant set $S$ for $\C$ with $|S| = \eta(\C)$. Furthermore, \cref{alg: redundant} runs in time $O(k n \log(q) \log(qn))$ (when implemented appropriately).
\end{claim}
\begin{proof}
    The correctness of the algorithm follows immediately from \cref{claim:redundant}.

    We now show that the running time is as claimed. Notice that the for loop can be run in time (better than) $O(n\log q(k + \log q))$, which is significantly less than the claimed running time. 
    
    The complexity is therefore dominated by the time required to find a maximal $S$ such that for all $j_1,j_2 \in S$ and all $i$, $a_{j_1} B_{i,j_{1}} = a_{j_2}B_{i,j_2}$. To find such an $S$, we can first generate a list $a_1 \vec{B}_1,\ldots, a_n \vec{B}_n \in \F_q^k$ of scaled columns of $\vec{B}$, which can be done in $O(kn \log^2 q)$ time. We then sort this list (using any total order on $\F_q^k$), which can be done in time $O(kn \log (n) \log (q))$. We then find an element that occurs maximally often in the sorted list, which can also be done in time (better than) $O(kn \log (n) \log (q))$. Finally, we find all $j$ such that the scaled $j$th entry in our original list $a_j\vec{B}_j$ is equal to this most common element, which can be done in time $O(n k \log q)$. And, finally, we we output the resulting set. So, the total running time is $O(kn \log^2 (q) + kn \log (n) \log (q)) = O(kn \log(q) \log(qn))$, as claimed.
\end{proof}

\subsection{Backward reduction}
\label{sec:backwards_reduction}

We are now ready to present our definition of backward-reduced bases.

\begin{definition}
    Let $\vec{B} = (\vec{b}_1; \dots; \vec{b}_k) \in \F_q^{k \times n}$ be a basis of a code $\C$. We say that $\vec{B}$ is \emph{backward reduced} if it is proper and $|\vec{b}_k^+| = \eta(\C(\vec{B}))$.
\end{definition}

Finally, we give an algorithm that finds a backward-reduced basis. See \cref{alg: backward_reduce}.

\RestyleAlgo{ruled}
\begin{algorithm2e}
\caption{Backward Reduction}
\label{alg: backward_reduce}
\KwIn{A proper basis $\vec{B} = (\vec{b}_1; \dots; \vec{b}_k) \in \F_q^{k \times n}$ for $\C$}
\KwOut{An invertible matrix $\vec{A} \in \F_q^{k \times k}$ such that $\vec{A} \vec{B}$ is backward reduced.}

$\{j_1,\ldots, j_t\} \leftarrow \mathsf{MaxRedundantSet}(\basis)$\

Let $m$ be minimal such that $B_{m,j_1} \neq 0$.\

\For{$i \in [m+1,k]$}{
    $\vec{b}_i \leftarrow \vec{b}_i- B_{m,j_1}^{-1} B_{i,j_1} \vec{b}_m$
}

$(\vec{b}_1;\ldots; \vec{b}_k) \leftarrow (\vec{b}_1;\ldots;\vec{b}_{m-1}; \vec{b}_{m+1};\ldots; \vec{b}_{k}; \vec{b}_m)$\

\Return the matrix corresponding to the linear transformation done to $\vec{B}$.
\end{algorithm2e}

\begin{claim}
    \label{clm:backward_run_time}
    On input a proper basis $\vec{B}$, \cref{alg: backward_reduce} correctly outputs an invertible matrix $\vec{A}$ such that $\vec{A} \vec{B}$ is backward reduced. Furthermore, \cref{alg: backward_reduce} runs in time $O(nk \log (q) \log(qn))$.
\end{claim}

\begin{proof}
    We first analyze the running time. By \cref{clm:redundant_alg}, the first step of computing $\{j_1,\ldots, j_t\}$ runs in time $O(nk \log (q) \log(qn))$, and clearly finding $m$ can be done in much less than the claimed running time. The for loop can be run in time $O(nk \log^2 q)$ time. The last steps can be run in time $O(n k \log q)$. So, the total running time is as claimed. 
    
    Furthermore, it is clear that $\vec{A}$ is invertible. We now prove that $\vec{A} \vec{B}$ is backward reduced. To that end, we first observe the simple fact that the first $k-1$ epipodal vectors of $\vec{A} \vec{B}$ are non-zero. Indeed, by \cref{lem: late_epipodal_invariant}, the for loop does not affect the epipodal vectors at all and therefore the basis remains proper after the for loop. Let $\vec{B} := (\vec{b}_1;\ldots;\vec{b}_k)$ be the basis before the last step of the algorithm, which we have just argued is proper, and let $\vec{b}_1',\ldots, \vec{b}_k'$ be the basis after the last step. Notice that $\vec{b}_i' = \vec{b}_i$ for $i < m$, and it follows immediately that $(\vec{b}_i')^+ = \vec{b}_i^+ \neq \vec0$ for such $i$. For $m \leq i < k$, $\vec{b}_i' := \vec{b}_{i+1}$, and it follows that $|(\vec{b}_i')^+| \geq |\vec{b}_{i+1}^+| > 0$, as needed.

    We next show that $\ell_k' := |(\vec{b}_k')^+| = t := \eta(\C)$. In particular, notice that this implies that $\ell_k'$ is non-zero, since we have already shown that the first $k-1$ epipodal vectors are non-zero, and we trivially have $t \geq 1$. So, proving that $\ell_k' = t$ implies the result. To that end, we first notice that after the for loop, $B_{i,j_\ell} = 0$ for all $\ell \in [t]$ and all $i \neq m$. To see this, observe that by the definition of $m$, $B_{i,j_1} = 0$ for all $i < m$. And, it is clear from the for loop that for $i > m$, we have $B_{i,j_1} = 0$ for all $i > m$. But, since $\{j_1,\ldots,j_t\}$ is a redundant set of coordinates, it immediately follows that $B_{i,j_\ell} = 0$ for all $i \neq m$ and all $\ell$, as claimed.
    
    To finish the proof, first notice that $(\vec{b}_k')^+ = \pi_{T}^\perp(\vec{b}_m)$, where $T := \{\vec{b}_i \ : \ i \neq m\}$. Since we have seen that all elements in $T$ have supports disjoint from $\{j_1,\ldots, j_t\}$, and since $\vec{b}_m$ has support containing $\{j_1,\ldots, j_t\}$, it follows that $\{j_1,\ldots, j_t\} \subseteq \supp((\vec{b}_k')^+)$. So, $\ell_k' \geq t$. The fact that $\ell_k' \leq t$ follows from \cref{lem:last_epipodal_redundant}, which in particular implies that $\ell_k' \leq t$. Therefore, $\ell_k' = t$, as needed.
\end{proof}

\subsection{Full backward reduction}
\label{sec:dkz}
Since backward reduction can be done efficiently, it is natural to ask what happens when we backward reduce many prefixes $\basis_{[1,i]}$ of a basis. We could simply do this for all $i \in [k]$, but it is natural to be slightly more fine-grained and instead only do this for $i \leq \tau$ for some threshold $\tau$. In particular, since the last $k-\poly(\log n)$ epipodal vectors tend to have length one even in very good bases (\fullornot{\cref{sec: k1_heuristic}}{see the full version \citegs to understand why}), it is natural to take $\tau \leq \poly(\log n)$ to be quite small, which leads to very efficient algorithms. This suggests the following definition.

\begin{definition}
    For some threshold $\tau \leq k$, a basis $\basis \in \F_q^{k \times n}$ is \emph{fully backward reduced up to $\tau$} if it is proper and $\basis_{[1,i]}$ is backward reduced for all $1 \leq i \leq \tau$.
\end{definition}

We now show how to easily and efficiently compute a fully backward-reduced basis, using the backward-reduction algorithm (\cref{alg: backward_reduce}) that we developed above. We present the algorithm in \cref{alg: total_backwards} and then prove its correctness and efficiency. Notice in particular that the algorithm only changes each prefix $\basis_{[1,i]}$ (at most) once.

\RestyleAlgo{ruled}
\begin{algorithm2e}
\caption{Full Backward Reduction}
\label{alg: total_backwards}
\KwIn{A proper basis $\vec{B} := (\vec{b}_1; \dots; \vec{b}_k) \in \F_q^{k \times n}$ for a code $\C$ and a threshold $\tau \in [1, k]$}
\KwOut{A basis for $\C$ that is totally backward reduced up to $\tau$.}
\For{$i=\tau,\ldots,1$}{
    $\vec{A} \leftarrow \text{BackwardReduction}(\vec{B}_{[1,i]})$\

    $\vec{B} \leftarrow (\vec{A} \oplus \vec{I}_{k-i}) \vec{B}$
}
\Return $\vec{B}$
\end{algorithm2e}

\begin{theorem}
    On input a proper basis $\vec{B} := (\vec{b}_1; \dots; \vec{b}_k) \in \F_q^{k \times n}$ for a code $\C$ and a threshold $\tau \in [1, k]$, \cref{alg: total_backwards} correctly outputs a basis $\basis' \in \F_q^{k \times n}$ for $\C$ that is fully backward reduced up to $\tau$. Furthermore, the algorithm runs in time $O(\tau^2 n \log (q)\log(qn))$.
\end{theorem}
\begin{proof}
    We first prove that the algorithm runs in the claimed time. Indeed, the running time is dominated by $\tau$ calls of the Backward Reduction algorithm (\cref{alg: backward_reduce}). By \cref{clm:backward_run_time}, each of these calls runs in time $O(\tau n \log(q)\log(qn)))$.\footnote{Our notation might be slightly misleading here. Notice that the matrix $\vec{A}$ corresponds to the linear transformation applied by Backward Reduction to the block $\vec{B}_{[1,i]}$, and in particular the operation $\vec{B} \leftarrow (\vec{A} \oplus \vec{I}_{k-i})\vec{B}$ certainly takes no longer to apply than the time  that it takes to run Backward Reduction. Indeed, multiplication by $\vec{A}$ applies at most $i$ elementary row operations.} So, the running time is clearly as claimed.

    We now show that the algorithm correctly outputs a fully backward-reduced basis. First, notice that in iteration $i$, the code $\C_j := \C(\basis_{[1,j]})$ remains unchanged for any $j \geq i$ (since $\vec{A}$ is invertible). In particular, $\eta(\C_j)$ remains unchanged. And, \cref{lem: epipodal_locality} tells us that $|\vec{b}_j^+|$ remains unchanged as well for $j > i$. By the correctness of \cref{alg: backward_reduce}, it follows that after iteration $i$, we must have that $|\vec{b}_j^+| = \eta(\C_j)$ for all $i \leq j \leq \tau$. Furthermore, all operations preserve properness. The result follows.
\end{proof}

We next bound $|\vec{b}_1|$ of a fully backward-reduced basis. In fact, when $\tau \geq \ceil{\log_q n}$, this bound matches the Griesmer bound. It is actually not hard to see that with $\tau = k$, a fully backward-reduced basis is in fact LLL-reduced as well. But, the below theorem shows that we do not need to go all the way to $\tau = k$ to achieve the Griesmer bound. This is because in the worst case, $|\vec{b}_i^+| = 1$ for all $i > \ceil{\log_q n}$ anyway.

\begin{theorem}
    \label{thm:griesmer_backwards}
    For any positive integers $k$, $n \geq k$, and $\tau \leq k$, a basis $\basis \in \F_q^{k \times n}$ of a code $\C$ that is fully backward reduced up to $\tau$ satisfies 
    \[
        \sum_{i=1}^\tau \left\lceil \frac{|\vec{b}_1|}{q^{i-1}} \right \rceil \leq n-k+\tau
        \; .
    \]
\end{theorem}
\begin{proof}
    Let $\C_i := \C(\basis_{[1,i]})$, and let 
    $
        s_i  := \supp(\C_i) = |\vec{b}_1^+| + \cdots + |\vec{b}_i^+|$.
    By the definition of a fully backward-reduced basis, for all $1 \leq i \leq \tau$ we must have
        $|\vec{b}_i^+| = \eta(C_i)
    $.
    Applying \cref{lem:lower_bound_on_max_last_epipodal}, we see that
    \begin{equation}
    \label{eq:griesmer_again_and_again}
        |\vec{b}_i^+| \geq \left\lceil \frac{q-1}{q^i-1} \cdot s_i \right\rceil
        \; ,
    \end{equation}
    for all $1 \leq i \leq \tau$. Noting that $s_{i} = s_{i-1} + |\vec{b}_i^+|$, after a bit of algebra we see that
    \[
        s_i \geq \frac{q^i-1}{q(q^{i-1}-1)} \cdot s_{i-1} \; ,
    \]
    for all $1 \leq i \leq \tau$. Applying this repeatedly, we see that for all such $i$,
    \[
        s_i \geq s_1 \cdot  \prod_{j=2}^i \frac{q^j-1}{q(q^{j-1}-1)} =  q^{-(i-1)} \cdot \frac{q^i-1}{q-1} \cdot  s_1 
    \]
    We then recall that $s_1 = |\vec{b}_1|$ and plug this back into \cref{eq:griesmer_again_and_again} to see that $|\vec{b}_i^+| \geq \ceil{|\vec{b}_1|/q^{i-1}}$.
    It then follows from the definition of $s_\tau$ that
    \[
        s_\tau \geq \sum_{i=1}^\tau \left\lceil \frac{|\vec{b}_1^+|}{q^{i-1}} \right \rceil
        \; .
    \]

    Finally, since the basis is proper, we have that 
    \[
        s_\tau = s_k - |\vec{b}_{\tau+1}^+| - \cdots - |\vec{b}_k^+| \leq s_k - (k-\tau) \leq n-k+\tau
        \; ,
    \]
    and the result follows.
\end{proof}

\subsection{Heuristic analysis suggesting better performance in practice}

Recall that our analysis of backward-reduced bases in \cref{sec:redundant} relied crucially on the repitition number $\eta(\C)$, which is the maximum over all proper bases of $\C$ of the length of the last epipodal vector. We showed that $\eta(\C)$ can be equivalently thought of as the maximal set of redundant coordinates. E.g., when $q = 2$, $\eta(\C)$ is precisely the number of repeated columns in the basis $\vec{B}$ for $\C$.

Our analysis of fully backward-reduced bases then relies on the lower bound on $\eta(\C)$ in \cref{lem:lower_bound_on_max_last_epipodal}. The proof of \cref{lem:lower_bound_on_max_last_epipodal} simply applies the pigeonhole principle to the (normalized, non-zero) columns of a basis $\vec{B}$ for $\C$ to argue that, if there are enough columns, then one of them must be repeated many times. Of course, the pigeonhole principle is tight in general and it is therefore easy to see that this argument is tight in the worst case.

However, in the average case, this argument is not tight. For example, if the number $n$ of (non-zero) columns in our basis $\vec{B} \in \F_2^{k \times n}$ is smaller than the number of possible (non-zero) columns $2^k-1$, then it is certainly possible that no two columns will be identical. But, the birthday paradox tells us that even with just $n \approx 2^{k/2}$, a random matrix $\vec{B} \in \F_2^{k \times n}$ will typically have a repeated column. 
More generally, if a code $\C$ is generated by a random basis $\vec{B} \in \F_q^{k \times n}$, then we expect to have
$
    \eta(\C) > 1
$
with probability at least $1-1/\poly(n)$, provided that, say, $n \geq 10   q^{k/2}\log n$, or equivalently, provided that 
\[
    k \leq 2(\log_q n - \log_q (10\log n ))
    \; .
\]

We could now make a heuristic assumption that amounts to saying that the prefixes $\basis_{[1,i]}$ behave like random matrices with suitable parameters (in terms of the presence of repeated non-zero columns). We could then use such a heuristic to show that we expect the output of \cref{alg: total_backwards} to achieve
\[
    k_1 > (2-o(1)) \log_q n
    \; .
\]

We choose instead to present in \fullornot{\cref{sec:selective}}{the full version \citegs} a variant of \cref{alg: total_backwards} that provably achieves the above. This variant is identical to \cref{alg: total_backwards} except that instead of looking at all of $\basis_{[1,i]}$ and choosing the largest set of redundant coordinates in order to properly backward reduce $\basis_{[1,i]}$, the modified algorithm chooses the largest set of redundant coordinates \emph{from some small subset} of all of the coordinates. In other words, the modified algorithm ignores information. Because the algorithm ignores this information, we are able to rigorously prove that the algorithm achieves $k_1 \gtrsim 2\log_q n$ when its input is a random matrix (by arguing that at each step the algorithm has sufficiently many fresh independent random coordinates to work with). 

We think it is quite likely that \cref{alg: total_backwards} performs better (and certainly not much worse) than this information-ignoring variant. We therefore view this as strong heuristic evidence that \cref{alg: total_backwards} itself achieves $k_1 \gtrsim 2 \log_q n$. (This heuristic is also confirmed by experiment. See \fullornot{\cref{sec: experiments}}{the full version \citegs}.)

\subsubsection{Backward reducing without all of the columns}
\label{sec:selective}

\newcommand{\submatrix}[3]{\vec{#1}^{(#2,#3)}}

To describe this algorithm, it helps to define some notation. We write $\submatrix{B}{i}{j}$ for the submatrix of $\vec{B}$ consisting of the first $i$ rows and the first $j$ columns. E.g., if $\vec{B} \in \F_q^{k \times n}$ has $n$ columns, then $\submatrix{B}{i}{n} $ is simply $\basis_{[1,i]}$. (Of course, this notation is useful because we will not always take $j = n$.)

We present the algorithm as \cref{alg: selective}. The algorithm works by backward reducing $\submatrix{B}{i}{j}$ for decreasing $i$ and increasing $j$. In other words, it behaves like \cref{alg: total_backwards} except that it only considers the first $j$ columns of $\basis$, for some $j$ that increases as the algorithm progresses.

\RestyleAlgo{ruled}
\begin{algorithm2e}
\caption{Selective Backward Reduction}
\label{alg: selective}
\KwIn{A matrix $\vec{B} = (\vec{b}_1; \dots; \vec{b}_k) \in \F_q^{k \times n}$ and $\beta \geq 2$ dividing $n-k$}
\KwOut{A proper basis of $\C(\basis)$}

Let $S \subseteq [n]$ be the indices of the first $k$ linearly independent columns of $\vec{B}$.\ 

Fail if $S \not\subseteq [k+\beta]$ (or if no such set exists).\

Systematize $\vec{B}$ with respect to $S$.\

$j \leftarrow k+2\beta$ \ 

\For{$i = (n-k)/\beta-1,\ldots, 1$}{
    $\vec{A} \leftarrow \text{BackwardReduction}(\submatrix{B}{i}{j})$\

    $\vec{B} \leftarrow (\vec{A} \oplus \vec{I}_{k-i}) \vec{B}$\ 
    
    $j \leftarrow j+\beta$\
}

\Return $\vec{B}$
\end{algorithm2e}

We now prove that the algorithm achieves the desired result. In particular, notice that this next theorem implies that the output basis will have 
\[
    k_1 \geq 2\log_q(n-k) - O(\log_q \log_q(n-k))
\]
with high probability, if we take $\beta = (n-k)/\poly(\log_q(n-k))$. 

\begin{theorem}
    \label{thm:provable_k1}
    If \cref{alg: selective} is run on input a uniformly random matrix and 
    \[
        7 \leq \beta \leq \frac{n-k}{2 \log_q(n-k)+1}
    \; ,
    \]then for any $\eps \in (0,1)$, the resulting output matrix $\vec{B} := (\vec{b}_1;\ldots; \vec{b}_k) \in \F_q^{k \times n}$ will satisfy $|\vec{b}_i^+| \geq 2$ for all 
    \[
    i \leq 2\log_q \beta - \log_q (20\log(2/\eps))
    \]
    with probability at least $1 - \log_q(\beta) 2^{1-\beta/20} - \eps$.
\end{theorem}
\begin{proof}
    Since the $i$th iteration does not affect $\vec{b}_a^+$ for $a > i$ (\cref{lem: epipodal_locality}), it suffices to show that with the appropriate probability we have that $|\vec{b}_i^+| \geq 2$ at the end of the $i$th iteration, for all $i \leq 2\log_q \beta - \log_q (20\log(2/\eps))$.
    
    First, by \cref{fact:random_full_rank}, we see that the first $k+\beta$ columns do not contain a linearly independent set except with probability at most
    \[
        q^{-k-\beta} \cdot \frac{q^{k} - 1}{q-1} \leq q^{-\beta}
        \; .
    \]
    Therefore, the failure probability is bounded by $q^{-\beta}$.

    And, notice that the systemization step multiplies $\basis$ by an invertible matrix that was chosen independently of the last $n-k-\beta$ columns of $\basis$. So, after the systemization step, the last $n-k-\beta$ columns of $\basis$ are still uniformly random and independent of all other entries in the matrix.
    
    Now, similarly, notice that in every iteration of the for loop, we multiply the matrix $\vec{B}$ by an invertible matrix that was chosen independently of the last $n-j$ columns of the matrix. So, by a similar argument, we see that going into the $i$th iteration, the matrix $\vec{B}^{(i,j)} \in \F_q^{i \times j}$ has $\beta$ columns that are uniformly random and independent of all other entries in the matrix. Therefore by \cref{lem:repeated_column}, there will be at least two repeated non-zero columns in $\vec{B}^{(i,j)} \in \F_q^{i \times j}$ except with probability at most
    \[
        2^{-\beta/20} + 2^{-\beta^2/(20 q^i)}
        \; .
    \]
    Notice from our characterization of backward-reduced bases in \cref{sec:redundant}, this immediately implies that $|\vec{b}_i^+| \geq 2$ except with probability at most $2^{-\beta/20} + 2^{-\beta^2/(20 q^i)}$ for any $i \leq n/\beta - 1$. (Here, we are using the trivial fact that $|\vec{b}_i^+| \geq \eta(\C(\vec{B}^{(i,j)}))$.)
    
    Putting everything together, we see that by the union bound, the probability that this fails for at least one index $i \leq \tau := 2\log_q \beta - \log_q (20\log(2/\eps))$ (which in particular implies $i \leq n/\beta - 1$) is at most
    \[
        q^{-\beta} + \tau \cdot  2^{-\beta/20} + \sum_{i=1}^{\tau} 2^{-\beta^2/(20 q^i)} \leq 2 \log_q(\beta) \cdot 2^{-\beta/20} + \sum_{i=1}^{\tau} 2^{-\beta^2/(20 q^i)}
        \; .
        \]
    Finally, we bound the sum by $\eps$ by noting that, e.g.,
    \[
        \sum_{i=1}^{\tau} 2^{-\beta^2/(20 q^i)} = \sum_{i=0}^{\tau-1} 2^{- q^i \cdot \beta^2/(20 q^\tau)} = \sum_{i=0}^{\tau-1} (\eps/2)^{q^i} \leq \sum_{i=1}^{\infty} (\eps/2)^{i} \leq \eps
        \; ,
    \]
    as needed.
\end{proof}

We conclude this section with two remarks. First, by using suitable generalizations of \cref{lem:repeated_column} to multicollisions, we could prove a generalization of \cref{thm:provable_k1}, which would allow us to argue that we get $|\vec{b}_i^+| \geq t$ for all $i \lesssim t/(t-1) \cdot \log_q (n-k)$. Second, we note that the LLL algorithm does not reduce the parameter $k_1$, so we note in passing that one can combine \cref{alg: total_backwards} or \cref{alg: selective} with LLL without lowering $k_1$ (and without increasing $|\vec{b}_1|$). The same thing is not necessarily true for other basis-reduction algorithms (such as slide reduction and BKZ, as defined in \cref{sec: bkz,sec:slide}), but one can still imagine that composing \cref{alg: total_backwards} or \cref{alg: selective} with other reasonable basis-reduction algorithms would be unlikely to make the basis worse in practice, and perhaps would be fruitful.

    \section{BKZ and slide reduction for codes}

    \subsection{The BKZ algorithm for codes}
    \label{sec: bkz}
In this section, we introduce the notion of a BKZ-reduced basis, give an algorithm for finding such a basis, and analyze the properties of a BKZ-reduced basis. (Recall that for convenience if $\basis \in \F_q^{k \times n}$, we define $\basis_{[i,j]}$ to be $\basis_{[i,k]}$ for $j > k$.)

\begin{definition}
    \label{def: forward}
    For a basis $\vec{B} = (\vec{b}_1; \dots; \vec{b}_k) \in \F_q^{k \times n}$, we say that $\vec{B}$ is forward reduced if $\vec{b}_1$ is a shortest non-zero codeword in $\C(\vec{B})$.
\end{definition}
\begin{definition}
    A basis $\vec{B} = (\vec{b}_1; \dots; \vec{b}_k)$ is said to be $\beta$-BKZ reduced if $\vec{B}_{[i, i+\beta-1]}$ is forward reduced for all $i \in [1, k]$.
\end{definition}
Notice that BKZ-reduced bases are proper by definition, since by definition $\vec{b}_i^+$ must be a \emph{non-zero} codeword in the code generated by $\vec{B}_{[i, i+\beta-1]}$. Furthermore, a BKZ-reduced basis always exists. In particular, if we simply take $\vec{b}_1$ to be a shortest non-zero codeword in $\C$ (which guarantees that it is also a shortest codeword in $\C(\vec{B}_{[1,\beta]})$ for any basis $\basis$ with first row equal to $\vec{b}_1$), then $\pi_1(\C)$ has rank $k-1$ (by \cref{fact: primitive}). And, by induction on $k$, we may assume that $\pi_1(\C)$ has a BKZ-reduced basis.

With that, we can present the BKZ algorithm. See \cref{alg:bkz}. Notice that the algorithm relies on a subprocedure for finding shortest codewords in codes with dimension (at most) $\beta$. (So, one might call this a reduction, but we adopt the convention used in the lattice literature and call it an algorithm.)

\begin{algorithm2e}
\caption{BKZ}
\label{alg:bkz}

\KwIn{A proper basis $\vec{B} = (\vec{b}_1; \dots; \vec{b}_k) \in \F_q^{k \times n}$ for $\C$ and $\beta \in [2, k]$}
\KwOut{A $\beta$-BKZ reduced basis $\vec{B} = (\vec{b}_1; \dots; \vec{b}_k) \in \F_q^{k \times n}$ for $\C$}

$i \gets 1$\

\While{$i < k$}{
    $j \leftarrow \min(i+\beta-1,k)$ \ 
    
    \If{$\vec{B}_{[i, j]}$ is not forward reduced}{
    
        $\vec{p} \leftarrow $ shortest non-zero codeword in $\C(\vec{B}_{[i, j]})$\
        
        $\vec{A} \leftarrow \text{InsertPrimitive}(\vec{B}_{[i, j]}, \vec{p})$\
        
        $\vec{B} \leftarrow (\vec{I}_{i-1} \oplus \vec{A} \oplus \vec{I}_{k-j)}) \vec{B}$\
        
        $i \gets \max(1, i-\beta+1)$\
    }
    \Else{
        $i \gets i+1$\
    }
}
\textbf{output} $\vec{B}$
\end{algorithm2e}

We now show that the algorithm does in fact output a $\beta$-BKZ-reduced basis (when it terminates).

\begin{claim}
    If on inputs a basis $\vec{B}$ and block size $\beta$ the BKZ algorithm terminates, then it outputs a $\beta$-BKZ reduced basis for $\C(\vec{B})$.
\end{claim}

\begin{proof}
    Notice that $\vec{B}$ is only ever multiplied by invertible matrices. Thus, $\C(\vec{B})$ remains unchanged throughout the algorithm, so the resulting output is in fact a basis for the correct code. 
    
    Furthermore, notice that a shortest non-zero codeword $\vec{p} \in \C(\vec{B}_{[i, j]})$ is primitive in the code $\C(\vec{B}_{[i, j]})$ by \cref{fact: primitive}. Therefore, $(\vec{B}_{[i, j]}, \vec{p})$ is valid input to \cref{alg: insert_primitive}.

    We now show that, if \cref{alg:bkz} terminates, then $\vec{B}$ is $\beta$-BKZ reduced. In particular, we show that at every step of the algorithm, the blocks  $\vec{B}_{[1, \beta]}, \dots, \vec{B}_{[i-1, i+\beta-2]}$ are forward reduced. 
    
    The statement is vacuously true at the start of the algorithm because in this case $i = 1$. We prove that the algorithm maintains this invariant by induction, i.e., we assume that 
    $\vec{B}_{[1, \beta]}, \dots, \vec{B}_{[i-1, i+\beta-2]}$ are all forward reduced at the start of a loop and prove that this is maintained at the end of the loop. 
    If $\vec{B}_{[i,i+\beta-1]}$ also happens to be forward reduced at the start of the loop, then the algorithm will increment $i$ by one, and clearly the invariant is preserved. If not, then the algorithm performs changes to $\basis$ that have no effect on $\basis_{[1,i-1]}$. Therefore, blocks $\basis_{[1,\beta]},\ldots, \basis_{[i'-1,i'+\beta-2]}$ remain forward reduced where $i' := \max(1,i-\beta+1)$. The algorithm then updates $i$ to $i'$, and clearly the invariant is preserved.
\end{proof}

Unfortunately, just like for lattices, we are not able to show that the BKZ algorithm terminates efficiently. But, we do show in \cref{sec: experiments} experiments in which it runs rather efficiently in practice. For completeness, we include in \cref{app:BKZ_running_time} a proof that it terminates in time roughly $n^k$. (Of course, this is not an interesting running time. In particular, it is slower than brute-force searching over the entire code for $q \leq n$.)

\subsubsection{LLL revisited}
\label{sec:LLL}

We now take a brief detour to revisit the LLL algorithm, which is simply \cref{alg:bkz} in the special case when $\beta = 2$.

\begin{definition}
    An LLL-reduced basis is a $2$-BKZ reduced basis.
\end{definition}

\cite{DDvAlgorithmicReductionTheory2022} showed an upper bound of $O(k n^3)$ on the performance of the LLL algorithm for $q = 2$. Here, we show that the algorithm works for $\F_q$ for arbitrary $q$ and, with slightly more careful analysis, we bound the running time by $O(k n^2 \log^2 q)$.

\begin{claim}
    For $\beta = 2$, \cref{alg:bkz} runs in time $O(k n^2 \log^2(q))$.
\end{claim}

\begin{proof}
    Let $\ell_i := |\vec{b}_i^+|$.
    Consider the following potential function
    $$\Phi(\ell_i, \dots, \ell_k) = \sum_{j=1}^k (k-j+1) \ell_j$$
    Notice that since $\ell_i \geq 1$ and $\sum \ell_i = n$, it follows that $\Phi$ is upper bounded by $kn$. Furthermore, $\Phi$ is lower bounded by $0$ since $\ell_i \geq 0$.
    
    Next we show any time $\vec{B}$ is updated, $\Phi$ decreases by at least one. Say \cref{alg:bkz} updates block $i$ (ensures $\vec{B}_{[i, i+1]}$ is forward reduced). Let $\ell_j$ be the length of the $j$th epipodal vector before $\vec{B}$ is updated, and $\ell_j'$ be the length after. Consider the change to $\Phi$ when block $i$ is updated. We note three key facts. The first is that if $j \neq i$ and $j \neq i+1$, then $\ell_j = \ell_j'$ (\cref{lem: epipodal_locality}). The second is that $\ell_i' < \ell_i$, which follows from the fact that block $i$ was updated. The third is that $\ell_{i}+\ell_{i+1} = \ell'_{i}+\ell'_{i+1}$ (\cref{lem: constant_epipodal_sum}). These three facts are enough to show that $\Phi$ decreases by at least one each time that $\vec{B}$ is updated.  In particular,
    \begin{align*}
        \Phi(\ell_1, \dots, \ell_k) - \Phi(\ell_1', \dots, \ell_k')
        &= \sum_{j = 1}^k (n-j+1) \ell_j - \sum_{j = 1}^k (n-j+1) \ell_j'\\
        &= (n-i+1) \ell_i + (n-i) \ell_{i+1} - (n-i+1) \ell_i' - (n-i) \ell_{i+1}'\\
        &= (n-i+1) (\ell_i - \ell_i') + (n-i)(\ell_{i+1} - \ell_{i+1}')\\
        &=(n-i+1) (\ell_i - \ell_i') - (n-i)(\ell_i - \ell_i')\\
        &= \ell_i - \ell_i' \geq 1
        \; .
    \end{align*}
    
    Finally, we show that the loop in \cref{alg:bkz} terminates in $2 nk + k$ iterations. Let $u$ (for up) denote the number of times the algorithm increments $i$ and $d$ (for down) denote the number of times the algorithm does not increment $i$ (typically decrementing $i$, but possibly leaving it unchanged, as it does when $i = 1$). Notice that the algorithm runs for $u+d$ iterations since each iteration either results in $i$ being incremented or not being incremented.
     We must have that $u-d \leq k$, since the algorithm terminates as soon as $i = k$ but the final index $i$ is clearly at least $u-d$ (since $i$ is incremented $u$ times and decreases by at most $1$ a total of $d$ times). And, we must have $d \leq nk$ since $i$ is only not incremented when $\vec{B}$ is updated, which we showed happens at most $nk$ times. It follows that
     \[
        u+d = 2 d + u - d \leq  2 d + k \leq 2 nk + k
        \; .
     \]
     So, \cref{alg:bkz} terminates after at most $2 nk + k$ iterations, as claimed.
    
    Since each iteration of the while loop takes $O(n \log^2 q)$ time (see \cref{subsec: log_q} for the running time needed to find a shortest non-zero codeword in $\C(\basis_{[i,i+1]})$ and \cref{claim: insert_primitive} for the running time of the Insert Primitive algorithm), the algorithm runs in time $O(k n^2 \log^2q)$, as claimed.
\end{proof}

We now generalize the fact that an LLL-reduced basis meets the Griesmer bound from $\F_2$ to $\F_q$ for any prime power $q$.
\begin{lemma}
    \label{lem: epipodal_decay}
    If $\vec{B} = (\vec{b}_1; \dots; \vec{b}_k) \in \F_q^{k \times n}$ be an LLL-reduced basis, then for all $1 \leq i \leq k-1$
    $$|\vec{b}_{i+1}^+| \geq \left\lceil \frac{|\vec{b}_i^+|}{q} \right\rceil 
    \; .$$
\end{lemma}

\begin{proof}
    For $a \in \F_q$, let $s_a := | \pi_{\vec{b}_i^+}(\vec{b}_{i+1} + a\vec{b}_i)|$. Notice that for every coordinate $j \in \supp(\vec{b}_i)$, there is precisely one 
    value of $a \in \F_q$ such the $j$th coordinate of $\vec{b}_{i+1} + a\vec{b}_i$ is zero. It follows that 
    \[
        \sum_{a \in \F_q} s_a = (q-1) \cdot |\vec{b}_i^+|
        \; .
    \]
    Therefore, there must exist some $a \in \F_q$ such that $s_a \leq (1-1/q)|\vec{b}_i^+|$. For this choice of $a$, we have
    \[
        |\pi_i(\vec{b}_{i+1} + a \vec{b}_i)| = |\vec{b}_{i+1}^+| + s_a \leq |\vec{b}_{i+1}^+| + (1-1/q)|\vec{b}_i^+|
        \; .
    \]
    But, since $\vec{b}_i^+$ is a shortest non-zero codeword in the code generated by $\basis_{[i,i+1]}$, it follows that 
    \[
        |\vec{b}_i^+| \leq |\vec{b}_{i+1}^+| + (1-1/q)|\vec{b}_i^+|
        \; .
    \]
    Rearranging, we see that $|\vec{b}_{i+1}^+| \geq |\vec{b}_i^+|/q$, and since $|\vec{b}_{i+1}^+|$ is an integer, we may take the ceiling, as needed.
\end{proof}

\begin{corollary}
    Let $\basis := (\vec{b}_1; \dots; \vec{b}_k) \in \F_q^{k \times n}$ be a LLL-reduced basis. Then,
    $$n \geq \sum_{i=0}^{k-1} \left\lceil \frac{|\vec{b}_1|}{q^i} \right\rceil$$
\end{corollary}

\begin{proof}
    This follows immediately from \cref{fact: lengt_invariance} and \cref{lem: epipodal_decay}.
\end{proof}     \subsubsection{Bounding \texorpdfstring{$|\vec{b}_1|$}{|b1|}}
    \label{subsec: bkz_bound_l1}
We now study the length of $\vec{b}_1$ achieved by BKZ-reduced bases. To that end, we will need to consider the following function, which is extremely well studied in coding theory.

\begin{definition}
    \label{def: len_bounding_function}
    Let $s_q(d, k): \mathbb{N} \times \mathbb{N} \to \mathbb{N}$ be the minimal value such that there exists a $[s_q(d, k), k, d]_q$ code. 
\end{definition}

Unfortunately, $s_q(d,k)$ does not seem to have a simple closed form in general. Instead, there are many different bounds on $s_q(d,k)$ that are tight in different regimes.

We first observe that the lengths $\ell_1,\ldots, \ell_n$ of the epipodal vectors of a BKZ-reduced basis must satisfy many constraints in terms of $s_q(d,k)$. However, because of the fact that $s_q(d,k)$ does not have a simple form, we do not know whether the problem of determining the maximal value of $|\vec{b}_1|$ that satisfies these constraints is efficiently computable. (Notice that we do not only include constraints in which $j = \max(k,i+\beta-1)$ but also include smaller values of $j$. It is not hard to find examples where these additional constraints with smaller $j$ significantly change the feasible set of this program.)

\begin{lemma}
    If $\basis := (\vec{b}_1;\ldots; \vec{b}_k) \in \F_q^{k \times n}$ is a $\beta$-BKZ-reduced basis with $\ell_i := |\vec{b}_i^+|$, then the $\ell_i$ are integers satisfying the following constraints. For all $i$, $\ell_i \geq 1$;
    \[
    \ell_1 + \cdots + \ell_k \leq n
    \; ;
    \]
    and for any integers $1 \leq i \leq j \leq \max(k,i+\beta-1)$,
    \[
        s_q(\ell_i,j-i+1) \leq \ell_i + \cdots + \ell_j
        \; .
    \]
\end{lemma}
\begin{proof}
    The fact that the $\ell_i$ are positive integers follows from the fact that $\vec{B}$ is a proper basis. The fact that they sum to at most $n$ follows from the fact that their sum is precisely $|\supp(\C)| \leq n$ (\cref{fact: lengt_invariance}). 
    
    Finally, we fix $1 \leq i \leq j \leq \max(k,i+\beta-1)$ and show that $s_q(\ell_i,j-i+1) \leq \ell_i + \cdots + \ell_j$. 
    Since by definition $\ell_i$ is a shortest non-zero codeword in $\C(\basis_{[i,i+\beta-1]})$, it is also a shortest non-zero codeword in $\C^\dagger := \C(\basis_{[i,j]})$ (since $\C^\dagger$ is a subcode of $\C(\basis_{[i,i+\beta-1]})$). In other words, $C^\dagger$ is an $[n^\dagger,j-i+1,\ell_i]$ code. Therefore, $n^\dagger \geq s_q(\ell_i,j-i+1)$. But, (again by \cref{fact: lengt_invariance}) its support size is $n^\dagger = \ell_i + \cdots + \ell_j$, as needed.
\end{proof}

We now derive an efficiently computable bound on $|\vec{b}_1|$ for a BKZ-reduced basis. We start with the following lemmas, the first of which is simply a puncturing argument.

\begin{lemma}
    \label{lem: shortening}
    For any prime power $q$, positive integers $d$ and $\beta$, and positive integer $x < d$, $s_q(d, \beta) \geq s_q(d-x, \beta)+x$.
\end{lemma}
\begin{proof}
    Take an $[s_q(d, \beta), \beta, d]_q$ code $\C$ generated by $\vec{B} \in \F_q^{\beta \times s_q(d, \beta)}$ and puncture $x$ coordinates of $\vec{B}$ to obtain $\vec{B}' \in \F_q^{\beta \times (s_q(d, \beta)-x)}$. Let $\C' = \C(\vec{B}')$. Since there does not exist nonzero $\vec{x} \in \F_q^\beta$ such that $|\vec{x} \vec{B}| < d$, there does not exist nonzero $\vec{x} \in \F_q^\beta$ such that $|\vec{x} \vec{B}'| < d-x$. So $\vec{B}'$ is full rank and $\C'$ has distance at least $d-x$. In other words, $\C'$ is a $[s_q(d, \beta)-x, \beta, d-x]_q$ code. Since a code with such parameters exists, $s_q(d-x, \beta) \leq s_q(d, \beta)-x$.
\end{proof}

\begin{lemma}
    \label{lem: last_vec_bound}
    Let $\vec{B} = (\vec{b}_1; \dots; \vec{b}_k) \in \F_q^{k \times n}$ be an LLL-reduced basis with $\ell_i := |\vec{b}_i^+|$. If $\textbf{b}_i^+$ is the shortest nonzero vector in $\pi_i(\mathcal{C}(\textbf{b}_i, \textbf{b}_{i+1}, \dots, \textbf{b}_{ i+\beta-1}))$, then 
    $$\ell_{i+\beta-1} \geq \Bigl\lceil \frac{(q-1) s_q(\ell_i, \beta)}{(q^{\beta}-1)} \Bigr\rceil$$.
\end{lemma}
\begin{proof}
    Applying the definition of $s_q$ and the LLL property that $\ell_i \leq q \ell_{i+1}$ (\cref{lem: epipodal_decay}) respectively, we derive
    $$s_q(\ell_i, \beta) \leq \sum_{j=i}^{i+\beta-1} \ell_j 
    \leq \sum_{j=i}^{i+\beta-1} q^{i+\beta-1-j} \ell_{i+\beta-1}
    = \frac{q^{\beta}-1}{q-1} \ell_{i+\beta-1}$$
    Rearranging and noting that $\ell_{i+\beta-1}$ is an integer yields the desired result.
\end{proof}

\begin{lemma}
    \label{lem: bound_support}
    Let $\vec{B} = (\vec{b}_1; \dots; \vec{b}_k) \in \F_q^{k \times n}$ be a $\beta$-BKZ reduced basis for a code $\C$. Then for any $i \geq 0$ such that $(\beta - 1) (i+1)+1 \leq k$, if $\ell_{(\beta-1)i + 1} \geq c$, then
    $$\ell_{(\beta-1)i + 2} + \dots + \ell_{(\beta - 1) (i+1)+1} \geq -c + s_q(c, \beta)
    \; ,$$
    where $|\vec{b}_j^+| := \ell_j$
\end{lemma}
\begin{proof}
    Let $x = \ell_{(\beta-1)i+1}-c$. Then, we have
    \begin{align*}
        &\ell_{(\beta-1)i+2} + \dots + \ell_{(\beta-1)(i+1)+1}\\
        =& -\ell_{(\beta-1)i+1} + \ell_{(\beta-1)i+1} + \ell_{(\beta-1)i+2} + \dots + \ell_{(\beta-1)(i+1)+1}\\
        \geq& -\ell_{(\beta-1)i+1} + s(\ell_{(\beta-1)i+1}, \beta)\\
        \geq& -\ell_{(\beta-1)i+1} + s(\ell_{(\beta-1)i+1}-x, \beta)+x &\text{\cref{lem: shortening}}\\
        \geq& -c+s(c, \beta)
        \; ,
    \end{align*}
    as needed.
\end{proof}

The following theorem allows us to derive a bound on $|\vec{b}_1|$ via successive applications of \cref{lem: bound_support} and \cref{lem: last_vec_bound}. We restrict attention to the case when $\beta-1$ divides $k-1$ for simplicity.
\begin{theorem}
    \label{thm: l1}
    Let $\vec{B} = (\vec{b}_1; \dots; \vec{b}_k) \in \F_q^{k \times n}$ be a $\beta$-BKZ reduced basis such that $\beta-1 $ divides $k-1$, and define
    $w_{1} := s_q(\ell_1, \beta)$,
    $$c_i := \Bigl\lceil \frac{(q-1) w_i}{(q^{\beta}-1)} \Bigr\rceil
    \; ,$$
    and 
    $w_{i+1} := s_q(c_i, \beta)$ for all $i \geq 1$.
    Then,
    $$n \geq w_1-c_1 + w_2-c_2 + \dots + w_{(k-1)/(\beta-1)}
    \; .$$
\end{theorem}
\begin{proof}
    We will prove the stronger statement that for all $i \in [1, (k-1)/(\beta-1)]$, $\ell_1 + \dots + \ell_{(\beta-1) i + 1} \geq w_1 - c_1 + w_2 - c_2 + \dots + w_i$ and $\ell_{(\beta-1) i + 1} \geq c_i$. The lemma will then follow from \cref{fact: lengt_invariance}. 
    
    We first show the statement holds for $i = 1$. The code $\C(\vec{b}_1, \dots, \vec{b}_{\beta})$ has support size $\ell_1 + \dots + \ell_\beta$, which allows us to apply the definition of $s_q$ to show $\ell_1 + \dots + \ell_\beta \geq s(\ell_1, \beta) = w_1$. Furthermore, $\ell_{\beta} \geq c_1$ by \cref{lem: last_vec_bound}. 
    
    For the inductive step, suppose $\ell_1 + \dots + \ell_{(\beta-1) i + 1} \geq w_1 - c_1 + w_2 - c_2 + \dots + w_i$ and $\ell_{(\beta-1) i + 1} \geq c_i$. By the inductive hypothesis, it suffices to show $\ell_{(\beta-1)(i+1)+1} \geq c_{i+1}$ and $\ell_{(\beta-1)i + 2} + \dots + \ell_{(\beta-1) (i+1) + 1} \geq -c_{i} + b_{i+1}$. Since $\ell_{(\beta-1) i + 1} \geq c_i$, it follows from  \cref{lem: last_vec_bound} that $\ell_{(\beta-1) (i+1) + 1} \geq c_{i+1}$. Furthermore, by 
 \cref{lem: bound_support} $\ell_{(\beta-1)i + 2} + \dots + \ell_{(\beta-1) (i+1) + 1} \geq -c_{i} + b_{i+1}$.
\end{proof}
We note that the bound in \cref{thm: l1} is monotonic in $s_q$. Formally, if we were to take $t_q$ such that for all $d, k \in \mathbb{N}$, $t_q(d, k) \leq s_q(d, k)$, then \cref{thm: l1} would still be true if we were to substitute $s_q$ with $t_q$.

This technique for bounding $\ell_1$ gives us a way to lift bounds on the distance of codes with dimension $\beta$ to bounds on the distance of codes with higher dimension. This can be seen as a generalization of the Griesmer bound which lifts a bound on the distance of dimension 2 codes to a bound on the distance higher dimension codes. 
     \subsubsection{An approximation factor for BKZ}
    \label{subsec: bkz_approx_factor}
In this section, we show that a BKZ-reduced basis satisfies $|\vec{b}_1| \leq q^{k-\beta} \cdot \mindist(\C(\basis))$. This generalizes the special case when $\beta = 2$ and $q = 2$, which was proven in \cite{DDvAlgorithmicReductionTheory2022}.

\begin{theorem}
    If $\basis = (\vec{b}_1;\ldots;\vec{b}_k) \in \F_q^{k \times n}$ is a $\beta$-BKZ-reduced basis, then $|\vec{b}_1| \leq q^{k-\beta} \mindist(\mathcal{C}(\vec{B}))$.
\end{theorem}

\begin{proof}
    We will show this by induction on $k$. In the base case $\beta = k$, immediately implying $|\vec{b}_1| = d_{\min}(\mathcal{C}(\vec{B}))$.
    We now let $k > \beta$ and turn our attention to the inductive step. Let $\vec{c}$ be a codeword in $\mathcal{C}(\vec{B})$ with length $d_{\min}(\mathcal{C}(\vec{B}))$. We consider two cases.

    First, suppose that $\supp(\vec{c}) \subseteq \supp(\vec{b}_1)$. If $\supp(\vec{c}) \subsetneq \supp(\vec{b}_1)$, then $\vec{b}_1$ is not primitive. But, this contradicts the fact that a BKZ-reduced basis must be proper, by \cref{fact:primitive_and_proper}. Therefore, in this case, we must have that $\supp(\vec{c}) = \supp(\vec{b}_1)$, and in particular that $|\vec{b}_1| = |\vec{c}| = \mindist(\C(\basis))$, which is significantly stronger than what we wanted to prove.

    Now, suppose, that $\supp(\vec{c}) \not \subseteq \supp(\vec{b}_1)$. Then, by definition $\pi_{\vec{b}_1}^\perp(\vec{c})$ is a \emph{non-zero} codeword in $\C(\basis_{[2,k]})$, so in particular $\mindist(\C(\basis_{[2,k]})) \leq |\pi_{\vec{b}_1}^\perp(\vec{c})| \leq |\vec{c}| = \mindist(\C(\basis))$. 
    And, by the induction hypothesis, 
    \[
        |\vec{b}_2^+| \leq q^{k-1-\beta} \mindist(\C(\basis_{[2,k]})) \leq q^{k-1-\beta} \mindist(\C(\basis))
        \; .
    \]
    Finally, recalling that any BKZ-reduced basis is also LLL-reduced and using \cref{lem: epipodal_decay}, we see that 
    \[
        |\vec{b}_1| \leq q |\vec{b}_2^+| \leq q^{k-\beta} \mindist(\C(\basis))
        \; ,
    \]
    as claimed.
\end{proof}     
    \subsection{The slide-reduction algorithm for codes}
    \label{sec:slide}

We can now define a slide-reduced basis and present our algorithm for computing a slide-reduced basis. The definition is nearly identical to Gama and Nguyen's definition of slide-reduced bases for \emph{lattices}~\cite{gamaFindingShortLattice2008} with dual-reduced bases replaced with backward-reduced bases (and, of course, with a basis for a lattice replaced by a basis for a code).

\begin{definition}
    \label{def:slide}
    A basis $\vec{B} \in \F_q^{k \times n}$ with $k = p\beta$ for some integer $p \geq 1$ is $\beta$-slide reduced if it satisfies the following three properties. 
    \begin{enumerate}
    \item $\vec{B}$ is proper.
    \item For all $i \in [0, k/\beta-1]$, $\vec{B}_{[i\beta+1, (i+1)\beta]}$ is forward reduced. \label{item:slide_forward}
    \item For all $i \in [0, k/\beta-2]$, $\vec{B}_{[i\beta+2, (i+1)\beta+1]}$ is backward reduced. \label{item:slide_backward}
    \end{enumerate}
\end{definition}

We present in \cref{alg:slide} our algorithm that produces a slide-reduced basis.

\RestyleAlgo{ruled}
\begin{algorithm2e}
\caption{Slide Reduction}
\label{alg:slide}

\KwIn{A proper basis $\vec{B} = (\vec{b}_1; \dots; \vec{b}_k) \in \F_q^{k \times n}$ for a code $\C$ and block size $\beta \in [2, k]$ that divides $k$.}
\KwOut{A $\beta$-Slide reduced basis $\vec{B} = (\vec{b}_1; \dots; \vec{b}_k) \in \F_q^{k \times n}$ for $\C$.}

$i \gets 0$\\
\While{$i < k/\beta$}{

    \uIf{$\vec{B}_{[i\beta+1, (i+1)\beta]}$ is not forward reduced}{
    
        $\vec{c} \leftarrow $ shortest codeword in $\C(\vec{B}_{[i\beta+1, (i+1)\beta]})$\ 

        $\vec{A} \leftarrow \text{InsertPrimitive}(\vec{B}_{[i\beta+1, (i+1)\beta]}, \vec{c})$\
        
        $\vec{B} \leftarrow (\vec{I}_{i \beta} \oplus \vec{A} \oplus \vec{I}_{k-(i+1)\beta}) \vec{B}$\ 

        $i \gets \max(0, i-1)$\ 
    }
    
    \uElseIf{$i \leq k/\beta-2$ and $\vec{B}_{[i\beta+2, (i+1)\beta+1]}$ is not backward reduced}{
    
        $\vec{A} \leftarrow \text{BackwardReduction}(\vec{B}_{[i\beta+2, (i+1)\beta+1]})$\
        
        $\vec{B} \leftarrow (\vec{I}_{i \beta + 1} \oplus \vec{A} \oplus \vec{I}_{k-(i+1)\beta-1}) \vec{B}$\
        
        $i \gets \max(0, i-1)$\
    }
    \Else{
        $i \gets i+1$\
    }

}

\Return $\vec{B}$
\end{algorithm2e}

We first prove that, if the algorithm ever terminates, then it must output a slide-reduced basis. (We will then prove that it terminates efficiently.)

\begin{claim}
    If \cref{alg:slide} terminates on input a proper basis $\basis \in \F_q^{k \times n}$ for a code $\C$ and a block size $\beta \in [2,k]$ that divides $k$, then it outputs a $\beta$-slide-reduced basis for $\C$.
\end{claim}

\begin{proof}
First, notice that by \cref{fact: primitive}, the vector $\vec{c}$ found by the forward reduction subprocedure is always primitive.
    So, by the correctness of the Insert Primitive (\cref{claim: insert_primitive}) and Backward Reduction (\cref{clm:backward_run_time}) subroutines, the basis remains a proper basis for $\C$ throughout the algorithm. So, we need only show that the output basis satisfies \cref{item:slide_forward,item:slide_backward} of \cref{def:slide}.

    We say block $j$ is forward reduced if $\vec{B}_{[j\beta+1, (j+1)\beta]}$ is forward reduced and we say block $j$ is backward reduced if $\vec{B}_{[j\beta+2, (j+1)\beta+1]}$ is backward reduced. (By convention, we say that block $k/\beta-1$ is always backward reduced.) 
    
    We will show that the loop invariant that block $j$ is forward and backward reduced for all $j < i$ always holds at the beginning and end of each loop iteration of \cref{alg:slide}. Notice that this holds vacuously before entering the loop, since $i$ starts at zero. Now we show that if this holds at the beginning of the loop, it also holds at the end of the loop. By assumption, we know that blocks $1, \dots, i-1$ are both forward and backward reduced at the start of the loop. There are two cases to consider.

    First, suppose no change at all happened to $\vec{B}$ during the loop. Notice that when this occurs, it must be the case that block $i$ was already both forward and backward reduced, and for all $j < i$, block $j$ was forward reduced by the induction hypothesis. At the end of the loop, $i$ is then replaced by $i' := i+1$, and we will in fact have that for all $j < i'$, block $j$ is both forward and backward reduced, as needed.

    On the other hand, if some change does occur during the loop, notice that any such change only affects blocks $i-1$ and $i$. And when this happens, $i$ is updated to $i' := \max(0,i-1)$. Of course, if $i' = 0$, then it is vacuously true that for all $j < i'$ block $j$ is both forward and backward reduced. If $i' = i-1$, then by assumption we know that for all $j < i'$, block $j$ was both forward and backward reduced at the beginning of the loop. Since block $j$ remains unchanged for $j < i'$, we see that this remains true at the end of the loop, as needed.
    
    Finally, notice that if \cref{alg:slide} terminates, then $i = k/\beta$, which means that block $j$ is both forward and backward reduced for all $j < k/\beta$, implying $\vec{B}$ is $\beta$-slide reduced.
\end{proof}

We now prove that our algorithm terminates efficiently, thus showing that there is an efficient algorithm for computing a slide-reduced basis. The proof itself is an adaptation of the beautiful potential-based proof of Gama and Nguyen in the case of lattices~\cite{gamaFindingShortLattice2008} (which itself adapts the celebrated proof of Lenstra, Lenstra, and Lov{\'a}sz for the LLL algorithm~\cite{lll82}). Notice that the final running time is 
\[
O(kn T/\beta) + kn^2 \beta \poly(\log n, \log q))
\; ,
\]where $T$ is the running time of the procedure used to find a shortest non-zero codeword in a $\beta$-dimensional block. In particular, if $\beta \gg \log_q n$, then we expect the first term to dominate and the running time to therefore be $O(kn T/\beta)$.

\begin{theorem}
    The while loop in \cref{alg:slide} runs at most $4kn/\beta $ times, and the total running time of the algorithm is $O(kn(T/\beta + n \beta \log^2 q + n \log(q) \log(qn)))$, where $T$ is the running time of the procedure used to find a shortest non-zero codeword in a $\beta$-dimensional block.
\end{theorem}
\begin{proof}
    First, we bound the running time of a single iteration of the while loop by $T+O(n \beta^2 \log^2 q + n \beta \log q \log n)$. Indeed, it takes time $T$ to find a shortest codeword in $\C(\basis_{[i\beta + 1, (i+1)\beta]})$ by assumption. (And, when this vector is found, it can be used to check if the block is forward reduced.) By \cref{claim: insert_primitive}, the Insert Primitive procedure on a block of size $\beta$ runs in time $O(n \beta^2 \log^2 q)$. And, the final matrix multiplication in the forward reduction part of the loop can be done in time $O(n\beta^2  \log^2 q)$. In the backward reduction part of the loop, the backward reduction algorithm itself runs in time $O(n \beta \log (q)\log(qn))$ on a $\beta$-dimensional block, by \cref{clm:backward_run_time}. The matrix multiplication step then takes time $O(n \beta^2 \log^2 q)$. So, the total running time of a single iteration of the while loop is $T+O(n \beta^2 \log^2 q + n \beta \log(q) \log(qn))$, as claimed.

    Next, we show that the total number of loop iterations is bounded by $4kn/\beta$. To that end, we first observe that it suffices to show that the backward reduction subprocedure runs at most $kn/\beta$ times. Indeed, we claim that the number $F$ of times the forward reduction subprocedure runs is at most $kn/\beta + 2B$, where $B$ is the number of times the backward reduction subprocedure runs. To see this, consider the number $N$ of blocks $\basis_{[j\beta + 1, (j+1)\beta]}$ that are forward reduced at a given step of the algorithm. Notice that $N$ increases by one each time the forward reduction subprocedure runs, since by definition this causes the block $\basis_{[i\beta + 1, (i+1)\beta]}$ to go from not being forward reduced to being forward reduced (and blocks $\basis_{[j\beta+1,(j+1)\beta]}$ for $i \neq j$ are unchanged, by \cref{lem: epipodal_locality}). When the backward reduction subprocedure runs, $N$ decreases by at most $2$, since the ``backward block'' $\basis_{[i\beta + 2, (i+1)\beta + 1]}$ overlaps with two blocks of the form $\basis_{[j\beta+1,(j+1)\beta]}$ (specifically, blocks $j = i$ and $j = i+1$) and all other blocks are unchanged (\cref{lem: epipodal_locality}). Since $N$ is non-negative and ends at $kn/\beta$, it follows that $F - 2B \leq kn/\beta$, as needed.
    
    It remains to bound the number of times that the backward reduction subprocedure runs. To that end, let $\ell_i := |\vec{b}_i^+|$ and consider the potential function\footnote{To get some intuition for $\Phi$, notice that $\Phi$ can be thought of as a weighted sum of the lengths of the epipodal vectors, where earlier vectors get more weight. In particular, $\Phi$ tends to be smaller when the earlier epipodal vectors are shorter relatively to the later epipodal vectors. Since the algorithm works by making earlier epipodal vectors shorter than later epipodal vectors later, it makes some intuitive sense that this potential function should decrease as the algorithm runs.}
    \[
        \Phi(\ell_1,\ldots, \ell_k) :=  \sum_{j=0}^{k/\beta - 1} \sum_{a=1}^{(j+1) \beta} \ell_a 
        \; .
    \]
    First, notice that for any basis, we must have $\Phi \geq 0$, since $\ell_a \geq 0$. And, similarly, we must have $\Phi \leq kn/\beta$, since for each $j$ the inner sum is at most $|\supp(\C)| \leq n$.\footnote{In fact, for a proper basis, we must have 
    \[\Phi \geq \beta + 2\beta + \cdots + (k/\beta-1) \beta + |\supp(\C)| = k(k/\beta -1)/2 + |\supp(\C)| 
    \; .
    \]
    And, proper bases must also satisfy 
    \[\Phi \leq (|\supp(\C)| - \beta(k/\beta-1)) + (|\supp(\C)| - \beta(k/\beta-2)) + \cdots + (|\supp(\C)| - \beta) + |\supp(\C)| = k|\supp(\C)|/\beta  -  k(k/\beta -1)/2
    \; .
    \]
    But, such optimizations would only affect the lower-order terms in our running time, so we do not bother with them.}
    Furthermore, notice that $\Phi$ is left unchanged by forward reduction steps, since forward reduction steps do not change $\sum_{a=j\beta + 1}^{(j+1)\beta} \ell_a$ for any $j$ (\cref{lem: epipodal_locality,lem: constant_epipodal_sum}).
    
    So, in order to show that the number of backward reduction steps is at most $kn/\beta$, it suffices to show that $\Phi$ decreases by at least one every time that the backward reduction subprocedure is called. Let $\ell_j$ be $|\vec{b}_j^+|$ immediately before a backward reduction step, and let $\ell_j'$ be $|\vec{b}_j^+|$ immediately after a backward reduction step performed on the block $\basis_{[i\beta + 2,(i+1)\beta + 1]}$. Notice that $\ell_j' = \ell_j$ for all $j \leq i\beta +1$ and all $j \geq (i+1)\beta + 2$ (\cref{lem: epipodal_locality}). Furthermore, by \cref{lem: constant_epipodal_sum}, we must have
    \[
        \ell_{i\beta + 2}' + \cdots + \ell_{(i+1)\beta + 1}' = \ell_{i\beta + 2} + \cdots + \ell_{(i+1)\beta + 1}
        \; .
    \]
    It follows that 
    \begin{align*}
        \Phi(\ell_1,\ldots,\ell_k) - \Phi(\ell_1',\ldots,\ell_k') 
            &=\sum_{a=i\beta + 2}^{(i+1)\beta} \ell_a  - \sum_{a=i\beta + 2}^{(i+1)\beta} \ell_a' \\
            &= \ell_{(i+1)\beta + 1}' - \ell_{(i+1)\beta + 1}  + \sum_{a=i\beta + 2}^{(i+1)\beta+1} \ell_a -\sum_{a=i\beta + 2}^{(i+1)\beta+1} \ell_a' \\
            &= \ell_{(i+1)\beta + 1}' - \ell_{(i+1)\beta + 1}
        \; .
    \end{align*}
    Finally, recalling that the backward reduction subprocedure always increases the length $\ell_{(i+1)\beta + 1}$ of the last epipodal vector in the block (since it starts out with a basis that is \emph{not} backward reduces and backward reduces it), we see that $\Phi(\ell_1,\ldots,\ell_k) - \Phi(\ell_1',\ldots,\ell_k') \geq 1$ as needed.
\end{proof}

\subsubsection{Bounding \texorpdfstring{$|\vec{b}_1|$}{|b1|}}

Finally, we prove a bound on the length of the first basis vector in a slide-reduced basis.
To that end, it is convenient to first define the notion of \emph{twin} reduction. (The analogous notion for lattices is implicit in \cite{gamaFindingShortLattice2008} and formally defined in \cite{aggarwalSlideReductionRevisited2020}.)

\begin{definition}
    A basis $\basis \in \F_q^{(\beta+1) \times n}$ is \emph{twin reduced} if $\basis_{[1,\beta]}$ is forward reduced and $\basis_{[2,\beta+1]}$ is backward reduced.
\end{definition}

As the next lemma shows, twin-reduced bases are convenient because they give us a way to relate the length of the first epipodal vector to the length of the last epipodal vector. (And, of course, twin-reduced bases are closely related to slide-reduced bases.)

\begin{lemma}
    \label{lem:twin_reduced}
    If $\basis := (\vec{b}_1;\ldots;\vec{b}_{\beta + 1}) \in \F_q^{(\beta+1) \times n}$ is a twin-reduced basis for a code $\C$, then 
    \[
        \ell_{\beta + 1} \geq \Big\lceil \frac{q-1}{q^\beta-q} \cdot (s_q(\ell_1,\beta)-\ell_1) \Big\rceil
        \; ,
    \]
    where $\ell_i := |\vec{b}_i^+|$.
\end{lemma}
\begin{proof}
    Let $\C_1 := \C(\basis_{[1,\beta]})$ and $\C_2 := \C(\basis_{[2,\beta+1]})$, and let 
    \[
    s := \supp(\C_1) = \ell_1 + \cdots + \ell_\beta
    \; .
    \]
    Because $\basis_{[2,\beta]}$ is backward reduced, we must have
    \[
        \ell_{\beta + 1} = \eta(\C_2)
        \; .
    \]
    But, by \cref{lem:lower_bound_on_max_last_epipodal}, we have that 
    \[
        \ell_{\beta + 1} = \eta(\C_2) \geq \frac{q-1}{q^\beta - 1} \cdot |\supp(\C_2)| = \frac{q-1}{q^\beta - 1} \cdot (s-\ell_1 + \ell_{\beta + 1})
        \;.
    \]
    Rearranging, we see that
    \[
        \ell_{\beta + 1} \geq \frac{q-1}{q^\beta - q} \cdot (s-\ell_1)
        \; .
    \]
    But, since $\basis_{[1,\beta]}$ is forward reduced, we have that $\ell_1 = \mindist(\C_1)$, so that by the definition of $s_q$, 
    \[
    s = |\supp(\C_1)| \geq s_q(\mindist(\C_1),\beta) = s_q(\ell_1,\beta)
    \; .
    \]
    Therefore,
    \[
        \ell_{\beta + 1} \geq \frac{q-1}{q^\beta - q} \cdot (s_q(\ell_1,\beta)-\ell_1)
    \]
    The result then follows by recalling that $\ell_{\beta + 1}$ is an integer, allowing us to take the ceiling of the right-hand side above.
\end{proof}

\begin{theorem}
    \label{thm: slide_bound_l1_prime}
    Let $\vec{B} = (\vec{b}_1;\ldots;\vec{b}_k) \in \F_q^{p\beta \times n}$ be a $\beta$-slide reduced basis, and let $c_1 := |\vec{b}_1|$, $w_{i} := s_q(c_i,\beta)$, and 
    \[
        c_{i+1} := \Big\lceil \frac{(q-1) (w_i-c_{i})}{q^{\beta}-q} \Big\rceil\; .
    \]
    Then,
    \[
        n \geq w_1 + w_2 + \dots + w_{p}
        \; .
    \]
\end{theorem}
\begin{proof}
   For $1 \leq i \leq p$, let
    \[
        s_i := |\vec{b}_{(i-1)\beta +1}^+| + \cdots + |\vec{b}_{i\beta}^+|
    \]
    be the support of the code generated by $\basis_{[(i-1)\beta + 1,i\beta]}$, and 
    let
    $
        a_i := |\vec{b}_{(i-1)\beta +1}^+|
    $ be the length of the first vector in this block.
    
    Notice that for all $1 \leq i \leq p-1$, the basis $\basis_{[(i-1)\beta + 1,i\beta + 1]}$ is twin reduced. 
    By \cref{lem:twin_reduced}, we therefore have that
    \begin{equation}
        \label{eq:ell_to_s}
        a_{i+1} \geq \Big\lceil \frac{q-1}{q^\beta-q} \cdot (s_q(a_i,\beta)-a_{i})\Big\rceil
        \; .
    \end{equation}
    (Notice that this inequality is of course quite closely related to the recurrence defining $c_i$.)

    We first prove by induction that for any $1 \leq i \leq p$, we have $a_i \geq c_{i}$. Indeed, the base case $i=1$ is trivial by the definition of $c_1 := a_1$. Then, for some $1 \leq i < p$, we assume that $a_i \geq c_{i}$, and we use this to prove that $a_{i+1} \geq c_{i+1}$. Indeed, by \cref{eq:ell_to_s}, we have that
    \[
        a_{i+1} \geq \Big\lceil \frac{q-1}{q^\beta-q} \cdot (s_q(a_i,\beta)-a_i)\Big\rceil
        \; .
    \]
    By \cref{lem: shortening}, $s_q(a,\beta)-a$ is a non-decreasing of $a$. Together with the induction hypothesis, this gives
    \[
        a_{i+1} \geq \Big\lceil \frac{q-1}{q^\beta-q} \cdot (s_q(a_i,\beta)-a_i)\Big\rceil \geq \Big\lceil \frac{q-1}{q^\beta-q} \cdot (s_q(c_i,\beta)-c_i)\Big\rceil = c_{i+1}
        \; ,
    \]
    as needed.

    To finish the proof, we notice that for all $1 \leq i \leq p$, since $\basis_{[(i-1)\beta + 1,i\beta]}$ is forward reduced, by definition $a_i$ is the minimum distance of the code generated by this basis. We therefore must have that the support $s_i$ of the code generated by this basis satisfies
    \[
        s_i \geq s_q(a_i,\beta)
        \; .
    \]
    Now, using the (trivial) fact that $s_q(a,\beta)$ is a non-decreasing function of $a$ together with the (now proven) inequality $a_i \geq c_i$, we see that
    \[
        s_i \geq s_q(c_i,\beta) = w_i
        \; .
    \]
    Finally, we have
    \[
        n \geq \supp(\C) = s_1 + \cdots + s_p \geq w_1 + \cdots + w_{p}
        \; ,
    \]
    as needed.
\end{proof}

    \section{How good can a basis be?}
        \label{sec: k1_heuristic}

To understand the promise and limitations of block reduction for codes, we now try to answer the question of how good a basis we can possibly hope for.

Of course, there are many different ways that one can imagine defining the quality of a basis $\basis$. \cite{DDvAlgorithmicReductionTheory2022} propose the very nice quantity
$$k_1  := |\{ i : |\vec{b}_i^+| > 1\}|
\; .$$
In \cite{DDvAlgorithmicReductionTheory2022}, the authors argue that bases with larger $k_1$ are in some sense ``better.'' In particular, they give an algorithm (Lee-Brickell-Babai) for finding short codewords and argue heuristically that the running time of this algorithm decreases as $k_1$ gets larger. 

We do not delve into the details of this here. Instead, we note that, intuitively, $k_1$ is a measure of how balanced a profile is, and thus can be thought of as a measure of the quality of a basis for the purpose of decoding, with larger $k_1$ suggesting a more balanced profile. We also note that when basis reduction is run in practice, the vast majority of epipodal vectors tend to have length one. Typically the $k_1$ non-trivial epipodal vectors are simply the first epipodal vectors $\vec{b}_1^+,\ldots, \vec{b}_{k_1}^+$, in which case basis-reduction algorithms with block size larger than $k_1$ perform no better than one-block reduction (\cref{sec: one_block}) with the same block size. So, one can think of upper bounds on $k_1$ as suggesting limitations of block reduction.

In this section, we will prove two results. First, we will argue (using unproven but mild heuristics) that any basis for a random code over $\F_2$ has $k_1 \leq O(\log (k) \log(n-k))$ (for a wide range of parameters $n$ and $k$). And, we will argue that all ``non-degenerate'' codes over $\F_2$ have a basis with $k_1 \geq \Omega(\log^2 k)$ (though we certainly do not provide an efficient algorithm for computing this basis). More specifically, we prove that any code over $\F_2$ has a basis $\basis$ such that $\sum_{i \geq \Omega(\log^2 k)} (|\vec{b}_i^+|-1)$ is large. This implies that $k_1$ is large unless the code has a remarkably low-support subcode.

\subsection{A heuristic upper bound on \texorpdfstring{$k_1$}{k1} for random codes}
\label{sec:k1_upper}

We first present our (unproven) heuristic assumption about the bases of random codes. The heuristic is rather technical, but it should be interpreted simply as a specific and rather weak version of the more general idea that ``projections of random codes look random.'' In particular, we note that our result is not very sensitive to the specific parameters chosen in \cref{heuristic:projections_random}, and that these parameters are chosen to be quite loose (thus making the heuristic a weaker assumption). Note that we state the heuristic as a property of all bases $\vec{B}$ of random codes for simplicity, but we of course only need the heuristic to hold for a basis maximizing $k_1$.

\begin{heuristic}
    \label{heuristic:projections_random}
    For any positive integers $k$, $k \leq n \leq k^{5}-1$, and $i < k$, if $\basis \in \F_2^{k \times n}$ is a proper basis of a random dimension-$k$ code $\C \subseteq \F_2^n$, then 
    \[
        \mindist(\C(\basis_{[i,k]})) \geq \frac{n'-k'}{10 \log k'}
        \; ,
    \]
    where $n' := |\supp(\basis_{[i,k]})|$ and $k' := k-i+1$.
\end{heuristic}

In fact, a random $[n',k']_2$ code typically will have such a minimum distance with very high probability. The following lemma justifies the heuristic by showing this formally. (The lemma is not actually necessary for the rest of the section, and is simply meant to provide evidence that the heuristic is justified. The reader may skip it if she is willing to take \cref{heuristic:projections_random} on faith.) In particular, notice that the heuristic trivially holds if $n'-k' \leq 10 \log k'$. The lemma below shows that \cref{heuristic:projections_random} holds with probability at least $1-1/(k')^5$ for $n'-k' \geq 10 \log k'$ under the assumption that the projection $\basis_{[i,k]}$ behaves like a random code.

\begin{lemma}
    \label{lem: loose_gv}
    For any positive integers  $n$ and $k \geq 2$ satisfying $k \leq n \leq k^{5}-1$, if $\vec{B} \sim \F_2^{k \times n}$ is sampled uniformly at random, then
    \[d_{\min}(\C(\vec{B})) > \frac{n-k}{10 \log k}
    \]
    with probability at least $1-2^{-(n-k)/2}$.
\end{lemma}
\begin{proof}
    Consider a fixed non-zero vector $\vec{z} \in \F_2^{k}$. Then, the random variable $\vec{B}\vec{z}$ is a uniformly random element in $\F_2^n$. Therefore,
    \begin{align*}
        \Pr[|\vec{B} \vec{z}| \leq d] 
        \leq \frac{\sum_{i=0}^d {n \choose i}}{2^n}
        \leq \frac{(n+1)^d}{2^n}
        \; .
    \end{align*}
    By the union bound, the probability that there exists a non-zero $\vec{z} \in \F_2^k$ such that $|\vec{B}\vec{z}| \leq d$ for $d := (n-k)/(10 \log k)$ is at most
    \[
        2^k \cdot \frac{n^d}{2^n} = 2^{k-n} n^{(n-k)/(10 \log k)} \leq 2^{k-n} k^{(n-k)/(2\log k)} = 2^{-(n-k)/2}
        \; ,
    \]
    as claimed.
\end{proof}

We now prove the main result of this section, which shows that under \cref{heuristic:projections_random}, the support of projected codes $\C(\basis_{[i,k]})$ decays rather quickly.

\begin{theorem}
    \label{thm: k1_upper_bound}
    For positive integers $k \geq 2$ and $n$ satisfying $k \leq n \leq k^5-1$, and a proper basis $\vec{B} \in \F_2^{k \times n}$ satisfying \cref{heuristic:projections_random}, we have 
    \[
        |\supp(\C(\basis_{[i,k]}))| \leq k + e^{-(i-1)/(10 \log k)}\cdot (n-k) 
    \]
    for all $i \in [k]$.
\end{theorem}
\begin{proof}
    Let $\C_i := \C(\basis_{[i,k]})$ be the $i$th projected code, $n_i := |\supp(\C_i)|$ its support, $k_i := k-i+1$ its dimension, and $s_i := n_i - k$. (Note that the fact that we subtract $k$ here and not $k_i$ is intentional.) Set 
    \[
        \delta := 1 - \frac{1}{10 \log k}
    \; .
    \]
    We prove by induction that $s_i \leq \delta^{i-1} (n-k)$. The result then follows by noting that $\delta < e^{-1/(10 \log k)}$.

    The base case $i = 1$ is trivial, since clearly $s_1 := n_1 - k \leq n - k$. For the induction step, we may assume that $s_i \leq \delta^{i-1} (n-k)$ for some $1 \leq i < n$ and use this to prove that $s_{i+1} \leq \delta^i (n-k)$.
    
    Indeed, by the heuristic, we must have 
    \[
        \mindist(\C_i) \geq \frac{n_i-k_i}{10 \log k_i} \geq  \frac{s_i}{10 \log k}
        \; .
    \]
    Since $\vec{b}_i^+ \in \C_i$ is non-zero (because the basis is proper), we must have $|\vec{b}_i^+| \geq s_i/(10 \log k_i)$. The result then follows by noting that 
    \[
        s_{i+1} = s_i - |\vec{b}_i^+| \leq \delta s_i \leq \delta^{i} (n-k)
        \; ,
    \]
    as needed.
\end{proof}

Finally, the following corollary shows how the above relates to the number $k_1$ of epipodal vectors with length greater than one. Notice that this holds even if $n = k^5$ is quite large relative to $k$.

\begin{corollary}
    For positive integers $k \geq 2$ and $n$ satisfying $k \leq n \leq k^5-1$, a proper basis $\vec{B} \in \F_2^{k \times n}$ satisfying \cref{heuristic:projections_random} has at most $20 \log (k) \log(n-k)+1 \leq O(\log^2 n)$ epipodal vectors $\vec{b}_i^+$ with length at least two.
\end{corollary}
\begin{proof}
    Let 
    \[
    i := 10 \log (k) \log(n-k)+1
    \; ,
    \]
    and let $\C' := \C(\basis_{[i,k]})$. By \cref{thm: k1_upper_bound},
        \[
        |\supp(\C')| \leq k + 1
        \; .
    \]
     Notice that $|\vec{b}_i^+| + \cdots + |\vec{b}_k^+| = |\supp(\C')|$, and since the basis is proper, we must have $|\vec{b}_j^+| \geq 1$ for all $j$. It follows that there are at most $i$ values of $j \geq i$ such that $|\vec{b}_j^+| \geq 2$. Therefore, the total number of epipodal vectors with length at least two is at most $i-1 + i = 2i - 1$, as claimed.
\end{proof}

\subsection{Griesmer-reduced bases have good profiles}
\label{sec:k1_lower}

We now turn to proving something like the fact that $k_1 \geq \Omega(\log^2 k)$ for Griesmer-reduced bases with suitable parameters. We will not be able to literally argue this in the worst case because of simple counterexamples. For example, some codes do not have any proper bases with $k_1 > 1$. So, we instead prove that the maximal index $k_1^*$ of an epipodal vector with length greater than one is large, which is the best that we can hope for.

 We begin by deriving a loose version of the Hamming bound, which will be convenient for our purposes.

\begin{lemma}
    \label{lem: loose_hamming}
    For positive integers $k$ and $n > k$ satisfying $n \leq k + \sqrt{k}$, any code $\C \subseteq \F_2^n$ with dimension $k$ must have 
    \[
        \mindist(\C) \leq 4 \cdot \frac{n-k}{\log n} + 2
        \; .
    \]

    Furthermore, if $n > k + \sqrt{k}$, then
    \[
        \mindist(\C) \leq 4 \cdot \frac{\sqrt{k}}{\log k} + 2 + (n-k-\sqrt{k})
    \]
\end{lemma}
\begin{proof}
    We prove the first statement first. Let $d := \mindist(\C)$.
    The Hamming bound (which one can prove via a packing argument) tells us that 
    \[
        2^{n-k} \geq \sum_{i=0}^{t} \binom{n}{i}
        \; ,
    \]
    where $t := \floor{(d-1)/2}$.
    In particular,
    \[
        2^{n-k} \geq \binom{n}{t} \geq (n/t)^t
        \; .
    \] 
    Rearranging, we see that
    \[
        t \leq (n-k)/\log (n/t) 
        \; .
    \]
    Furthermore, recall from the Singleton bound that $d \leq n-k+1$, so $t \leq d-1 \leq (n-k+1)-1 = n-k \leq \sqrt{n}$.
    Plugging this in to the above, we see that
    \[
        t \leq (n-k)/\log(\sqrt{n}) = 2(n-k)/\log n
        \; .
    \]
    It follows that $d \leq 2t + 2 \leq 4(n-k)/\log n + 2$.

    The case when $n > k + \sqrt{k}$ then follows, e.g., by \cref{lem: shortening}.
\end{proof}

We now prove our main theorem of this section, which in particular tells us that for, say, $i = \log^2(k)/20$, a Griesmer-reduced basis will have $|\supp(\basis_{[i,k]})| > k$, or equivalently that $k_1^* > i$, as needed. Notice that this holds even if for very small supports $n \approx k + \sqrt{k}$.

\begin{theorem}
    \label{thm:k1_lower_bound_new}
    For sufficiently large integers $k$ and $n \geq k + \sqrt{k}/2$, a Griesmer-reduced basis $\vec{B} \in \F_2^{k \times n}$ of a dimension-$k$ code $\C$ with $|\supp(C)| = n$ satisfies 
    \[
        |\supp(\basis_{[i,k]})| \geq k- i + 1 + e^{-10i/\log k} \cdot \sqrt{k}/2
    \]
    for all $i \leq \log^2(k)/20$.
\end{theorem}
\begin{proof}
    Let $\C_i := \C(\basis_{[i,k]})$ be the $i$th projected code, $n_i := |\supp(\C_i)|$ its support, $k_i := k-i+1$ its dimension, and $s_i := n_i - k_i$ be the ``excess support'' of $\C_i$. Set 
    \[
        \delta := 1 - \frac{5}{\log k}
    \; .
    \]
    (Since $k$ is sufficiently large, $\delta$ is positive.)
    We prove by induction that $s_i \geq \delta^i \cdot \sqrt{k}/2$. The result then follows by noting that $\delta \geq e^{-10/\log k}$.

    The base case $i = 1$ is trivial, since by assumption $s_1 = |\supp(\C_1)|-k = |\supp(\C)|-k = n-k \geq \sqrt{k}/2 > \delta \sqrt{k}/2$.
    For the induction step, we assume that $s_i \geq \delta^i \cdot \sqrt{k}$ for some $1 \leq i < \log^2 (k)/20$ and use this to prove that $s_{i+1} \geq \delta^{i+1} \sqrt{k}/2$. Notice that by our bound on $i$, we have that $ k_i \geq k - \log^2(k)/20 \geq k/2$, where we have used the fact that $k$ is sufficiently large.

    There are two cases to consider. If $s_i \leq \sqrt{k_i}$, then we may apply \cref{lem: loose_hamming} to conclude that 
    \[
        \mindist(\C_i) \leq 4 \frac{s_i}{\log n_i} + 2
        \; .
    \]
    Since $\basis$ is Griesmer reduced, we must have $|\vec{b}_i^+| = \mindist(\C_i)$, and therefore that
    \[
        s_{i+1} = s_i - |\vec{b}_i^+| + 1 \geq s_i \cdot \Big(1- \frac{4}{\log n_i} \Big) - 1 \geq \delta^{i} (\sqrt{k}/2) \cdot \Big(1- \frac{4}{\log (k/2)} \Big) - 1 \geq \delta^{i+1} \sqrt{k}/2
        \; ,
    \]
    where in the second-to-last step we used the inequality $n_i \geq k_i \geq k/2$, and in the last step we used the assumption that $k$ is sufficiently large.

    If $s_i > \sqrt{k_i}$, then by the ``furthermore'' in \cref{lem: loose_hamming}, we see that
    \[
        \mindist(\C_i) \leq 4 \cdot \frac{\sqrt{k_i}}{\log k_i} + 2 + (s_i - \sqrt{k_i})
        \; .
    \]
    Again using the fact that the basis is Griesmer reduced, we have
    \[
        s_{i+1} = s_i - |\vec{b}_i^+| + 1 \geq \sqrt{k_i} - 4 \cdot \frac{\sqrt{k_i}}{\log k_i} -1 > \delta^{i+1} \sqrt{k}/2 
        \; ,
    \]
    where the last inequality uses the fact that $k_i \geq k/2$ and that $k$ is sufficiently large.

    The result follows.
\end{proof}

    \section{Two illustrative algorithms}
    \label{sec:pathological}
    
    \subsection{Approximate Griesmer reduction}
    \label{sec: approx}
We first present an algorithm that we call approximate Griesmer reduction. See \cref{alg:approx}. In fact, this is a family of algorithms parameterized by some subprocedure that is used to find a short (but not necessarily shortest) non-zero codeword in a code. This algorithm does not seem to be amenable to analysis, but we view it as an important algorithm to keep in mind when studying basis reduction for codes. Indeed, in our experiments in \cref{sec: experiments}, this algorithm when instantiated with the Lee-Brickell-Babai algorithm of \cite{DDvAlgorithmicReductionTheory2022} produces bases with much better profiles than our other algorithms---at the expense of a much worse running time.

\RestyleAlgo{ruled}
\begin{algorithm2e}
\caption{Approximate Griesmer Reduction}
\label{alg:approx}

\KwIn{A proper basis $\vec{B} = (\vec{b}_1; \dots; \vec{b}_k) \in \F_q^{k \times n}$ for $\C$}
\KwOut{A proper basis $\vec{B} = (\vec{b}_1; \dots; \vec{b}_k) \in \F_q^{k \times n}$ for $\C$}

\For{$i = 1,\ldots, k-1$}{

    Find a short non-zero codeword $\vec{c}$ in the code generated by $\basis_{[i,k]}$\

     $\vec{p} \leftarrow \text{MakePrimitive}(\basis_{[i,k]}, \vec{c})$\
     
     $\vec{A} \leftarrow \text{InsertPrimitive}(\basis_{[i,k]}, \vec{p})$\

    $\vec{B} \leftarrow (\vec{I}_{i-1} \oplus \vec{A}) \vec{B}$
}

\textbf{output} $\vec{B}$

\end{algorithm2e}

\begin{claim}
    \cref{alg:approx} returns a proper basis for $\C$ and runs in time $\sum_{i=1}^{k-1} f(n, k-i+1, q) + O(n k^3 \log^2 q)$, where $f(n, k, q)$ is a bound on the running time of the subprocedure for finding short non-zero codewords on $[n',k]_q$ codes for all $n' \leq n$.
\end{claim}
\begin{proof}
    The fact that the output basis is proper follows from the fact that the $\text{InsertPrimitive}$ subroutine always outputs $\vec{A}$ such that $\vec{A} \vec{B}_{[i,k]}$ is proper, together with \cref{lem: epipodal_locality}, which in particular implies that $(I_{i-1} \oplus \vec{A}) \vec{B}$ is proper if $\vec{B}_{[1,i-1]}$ is proper and $\vec{A} \vec{B}_{[i,k]}$ is proper (both of which are true).

    The running time of the all of the calls to a short-codeword-finding algorithm is clearly $\sum_{i=1}^{k-1} f(n,k-i+1, q)$, as claimed. The MakePrimitive subprocedure runs in time $O(nk^2 \log^2 q)$ by \cref{clm:make_primitive}, as does the InsertPrimitive subprocedure, by \cref{claim: insert_primitive}. Therefore, the total running time is 
    \[
        \sum_{i=1}^{k-1} (f(n,k-i+1, q) + O(n k^2 \log^2 q)) = \sum_{i=1}^{k-1} f(n,k-i+1, q) + O(nk^3 \log^2 q)
        \; ,
    \]
    as claimed.
\end{proof}     
    \subsection{One-block reduction}
    \label{sec: one_block}
We now introduce a very simple reduction algorithm that we argue is, in many cases, as good as more advanced reduction techniques for finding short non-zero codewords. See \cref{alg: one_block}. The algorithm itself can be viewed as a sort of algorithmic instantiation of the shortening technique that is often used to bound the minimum distance of codes (i.e., an algorithmic instantiation of \cref{lem: shortening}). (This is closely related to information set decoding.)

\RestyleAlgo{ruled}
\begin{algorithm2e}
\caption{One-block Reduction}
\label{alg: one_block}

\KwIn{A basis $\vec{B} = (\vec{b}_1; \dots; \vec{b}_k) \in \F_q^{k \times n}$ and block size $\beta \in [2, k]$}
\KwOut{A short non-zero codeword in $\C(\basis)$}

Systematize $\vec{B}$\

\Return A shortest non-zero codeword in $\C(\vec{b}_1, \dots, \vec{b}_{\beta})$

\end{algorithm2e}

\begin{claim}
    \label{clm:one_block}
    On input a basis $\vec{B} \in \F_q^{k \times n}$ and block size $\beta \in [2,k]$, \cref{alg: one_block} finds a non-zero codeword with length at most $d_q(n-k+\beta,\beta)$, where $d_q(n',k')$ is the maximal minimum distance of an $[n',k']_q$ code. Furthermore, the algorithm runs in time $O(n k^2 \log^2 q) + T$ where $T$ is a bound on the running time needed to find the shortest codeword in a $[n', \beta]_q$ code for any $n' \leq n-k+\beta$.
\end{claim}
\begin{proof}
    We first bound the running time. The systematization step runs in time $O(n k^2 \log^2 q)$. The systematization of $\vec{B}$ ensures that there exist at least $k-\beta$ coordinates such that $\vec{b}_1, \dots, \vec{b}_\beta$ are all zero on those coordinates. In particular, the support of $\C' := \C(\vec{b}_1, \ldots, \vec{b}_\beta)$ has size at most $n-k+\beta$, and therefore the time to find the shortest non-zero codeword is bounded by $T$.

    To prove correctness, we again notice that $\C'$ has support $n' \leq n-k+\beta$. Since $d_q$ is non-decreasing as a function of $n$, it follows that $d_q(n',\beta) \leq d_q(n-k+\beta,\beta)$. Since the output vector is a shortest codeword in an $[n',\beta]_q$ code, it must be shorter than $d_q(n',\beta)$, and the result follows.
\end{proof}

Finally, we note that, just like the LLL algorithm, \cref{alg: one_block} with $\beta = \lceil \log_q n \rceil$ finds a codeword that matches the Griesmer bound. Notice also that with $\beta$ so small, the step of finding the short non-zero codeword can be done in time $T \leq O(q^\beta nk) \leq q n^2 k$ (by simply enumerating all $q^\beta$ codewords), which is efficient for $q \leq n^{O(1)}$. So, one can view this as an alternative algorithmic instantiation of the Griesmer bound.

\begin{claim}
    \label{clm:one_block_griesmer}
    Let $\vec{c}$ be the codeword found by \cref{alg: one_block} on input $\vec{B} \in \F_q^{k \times n}$ and $\beta=\ceil{\log_q(n)}$.
    $$n \geq \sum_{i = 1}^{k} \left\lceil {\frac{|\vec{c}|}{q^i}} \right\rceil$$
\end{claim}

\begin{proof}
    By \cref{clm:one_block}, 
    \[
    |\vec{c}| \leq d_q(n-k+\beta,\beta)
    \; .
    \]
    Applying the Griesmer bound for an $[n-k+\beta,\beta]$ code, we see that
    \[
        n-k+\beta \geq \sum_{i=1}^\beta \left\lceil \frac{|\vec{c}|}{q^i} \right\rceil 
        \; .
    \]
    Finally, since $|\vec{c}| \leq n$, notice that for $i > \beta$, we must have $\ceil{|\vec{c}|/q^i} = 1$. Therefore,
    \[
        n \geq \sum_{i=1}^\beta \left\lceil \frac{|\vec{c}|}{q^i}\right\rceil + (k-\beta) = \sum_{i=1}^k \left\lceil \frac{|\vec{c}|}{q^i} \right\rceil
        \; ,
    \]
    as claimed.
\end{proof}     
    
    \section{Experiments}
    \label{sec: experiments}
\usetikzlibrary{calc}

Finally, we present the results of our experiments. At a high level, the results seem to present a relatively clear message, which is that full backward reduction is a good idea. The more detailed picture is a bit nuanced, though.

For our experiments we considered bases $\vec{B} \in \F_2^{k \times n}$ where $k = n/2$. Each experiment was run by sampling a basis uniformly at random from the set $\F_2^{k \times n}$, systematizing the basis (so that it is proper), and then running a particular reduction algorithm with the systematized basis as input. Algorithm running times do not include the sampling or systematization steps. 
We stress that the running times are meant to be a rough guide only. Our code is written in Python and not heavily optimized, and the experiments were run using an Apple M2 processor. 

\paragraph{The tested algorithms.} We now discuss the different algorithms that we test.

We tested BKZ with block sizes $\beta=2$, $\beta=4$, and $\beta=8$. In particular, recall that BKZ with block size $\beta = 2$ is the LLL algorithm. 

We also tested slide-reduction with block sizes $\beta = 4$ and $\beta = 8$, followed by LLL. We chose to apply LLL to the basis after running slide reduction, since unlike our other reduction algorithms, a basis being slide-reduced only provides guarantees on $|\vec{b}_{i}^+|$ where $i = 1 \bmod \beta$, so that an additional post-processing step is necessary if we wish to study, e.g., the entire profile of the basis. Similar issues arise in the literature on lattices, where one typically either \emph{defines} a slide-reduced basis to also be LLL reduced, or runs LLL as postprocessing. It is not hard to prove that running LLL after slide reduction does not affect $|\vec{b}_1|$. (To be clear the running time of LLL is counted as part of the running time.) 

For approximate Griesmer reduction, we use the Lee-Brickell-Babai implementation from \cite{DDvAlgorithmicReductionTheory2022} with $w_2 = 2$ as our subprocedure for finding short non-zero codewords. We include the running time of approximate Griesmer reduction for completeness. However, since the Lee-Brickell-Babai subroutine is implemented in C++ and heavily optimized, and our other algorithms are implemented in Python (and not heavily optimized), we note that the results should not be taken too seriously. Since approximate Griesmer reduction is particularly expensive to run (even with this optimized version of Lee-Brickell-Babai), we do not reduce a block $i$ if $\ell_i$ is already less than six since the expected change to basis quality in such a case is rather small.

Full backward reduction was run with threshold $\tau = 3 \log_2 n$. This threshold was chosen to be well above $k_1$, so that choosing a larger $\tau$ would not affect the basis quality.

\paragraph{The figures. } In \cref{fig: graph_decay}, we show a plot of the average length of the $i$th longest epipodal vector. (Here, we follow \cite{DDvAlgorithmicReductionTheory2022} in averaging the lengths of the \emph{$i$th longest} epipodal vectors, rather than the length $\ell_i$ of the $i$th epipodal vector itself. There is little difference for small $i$, but as $i$ becomes large and $\ell_i$ becomes small, the mean of $\ell_i$ itself becomes dominated by outliers.) This shows how well each reduction algorithm does at balancing the epipodal profile of a basis. We also zoom in to show just the first three values, and just $\ell_1$. 

Unsurprisingly, approximate Griesmer reduction produces the most balanced profile. And, we see that all of the new algorithms perform significantly better than LLL. Perhaps surprisingly, full backward reduction creates quite a balanced profile, though unsurprisingly it does not produce very short vectors (since it does not explicitly find any short codewords in any blocks). We also see that increasing the block size yields a notable improvement in basis quality for both slide reduction and BKZ, and that BKZ does seem to produce better profiles than slide reduction (as in the case of lattices).

\cref{fig: runtimes_1} gives a comparison of the running times of our various algorithms. Again, these running times are meant to only provide a rough sense of the relative performance of these algorithms. However, we do still see here that full backward reduction is quite fast, as is slide reduction. BKZ with $\beta = 8$ seems notably slower than these algorithms. The running time of approximate Griesmer reduction is difficult to compare with the others because, as we mentioned above, we are using an optimized algorithm from \cite{DDvAlgorithmicReductionTheory2022}, but we still see that even with these caveats approximate Griesmer becomes notably slower than the other algorithms as the dimension increases.

Finally, \cref{fig: average_k1} shows the average value of $k_1$ for all our reduction algorithms.
Again, approximate Greismer reduction clearly produces the best (most balanced basis) but also has the longest running time. Here, we see that slide reduction only offers a modest improvement over LLL, while BKZ with block size $\beta = 8$ is notably better (at the cost of a higher running time). Finally, full backward reduction again performs nearly as well as BKZ with block size $\beta = 8$ with far superior running time.

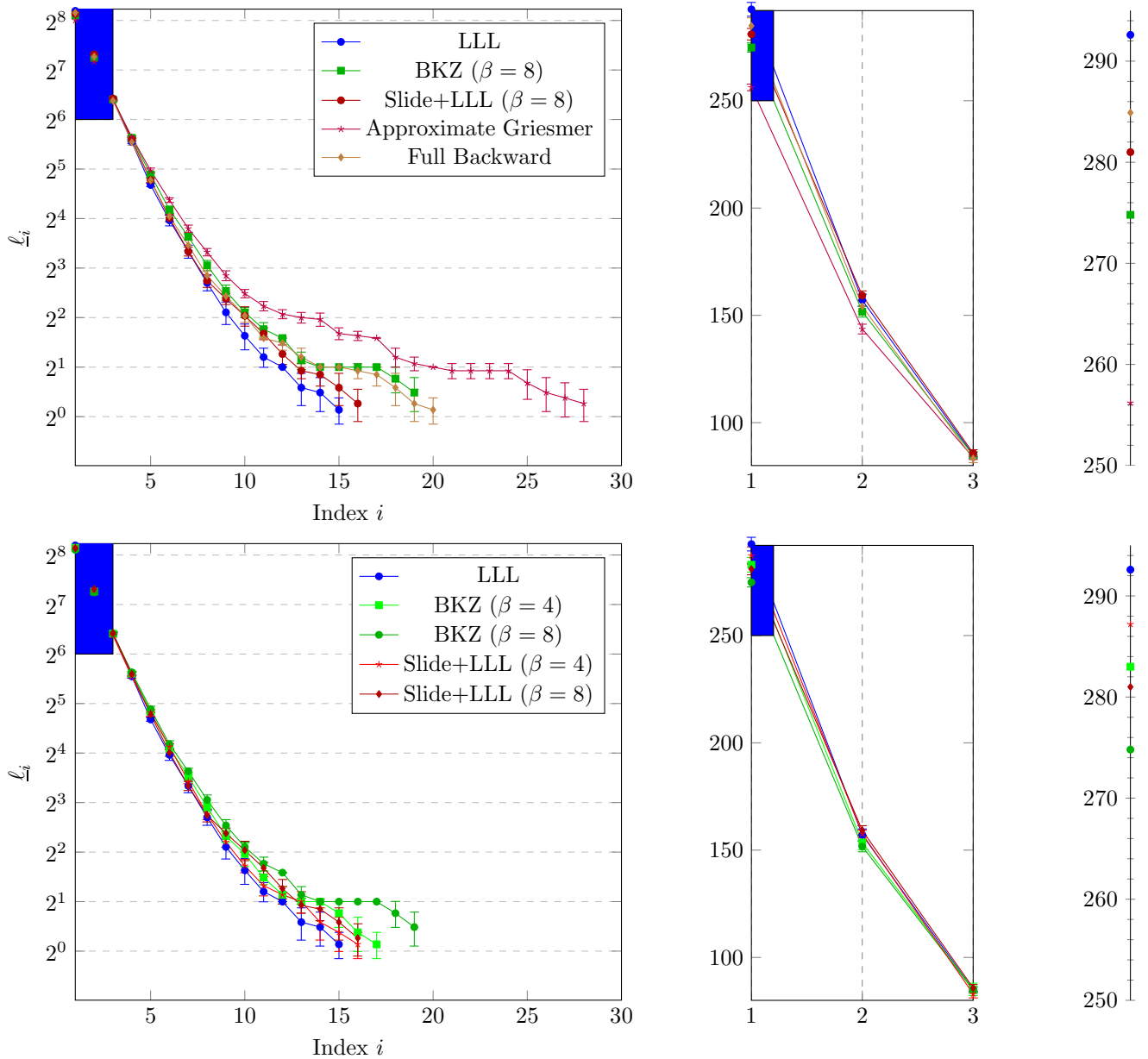
\begin{figure}
\begin{tikzpicture}
\begin{axis}[
    xlabel={Index $i$},
    ylabel={$\underline{\ell}_i$},
    xmin=1, xmax=30,
    ymin=0, ymax=300,
    legend pos=north east,
    ymajorgrids=true,
    grid style=dashed,
    ymode=log,
    log basis y={2},
    name=ax1,
]

\addplot+[
  blue, mark options={blue, scale=0.75},
  error bars/.cd, 
    y fixed,
    y dir=both, 
    y explicit
] table [x=x, y=y,y error=error, col sep=comma] {
    x,  y,       error
    1,  292.6, 3.214
    2, 157.1, 2.772
    3, 85, 1.154
    4, 46.8, 1.067
    5, 25.6, 0.611
    6, 15.6, 1.162
    7, 10.1, 0.917
    8, 6.5, 0.683
    9, 4.3, 0.670
    10, 3.1, 0.554
    11,  2.3, 0.306
    12, 2, 0
    13, 1.5, 0.333
    14, 1.4, 0.327
    15, 1.1, .2
};
\addlegendentry{LLL}

\addplot+[
  black!30!green, mark options={black!30!green, scale=0.75},
  error bars/.cd, 
    y fixed,
    y dir=both, 
    y explicit
] table [x=x, y=y,y error=error, col sep=comma] {
    x,  y,       error
    1,  274.8, 2.146
    2, 151.7, 2.513
    3, 84.8, 1.572
    4, 49.5, 0.537
    5, 29.6, 1.2
    6, 18.1, 0.917
    7, 12.4, 0.533
    8, 8.3, .6
    9, 5.8, 0.499
    10, 4.3, 0.306
    11, 3.4, 0.327
    12, 3, 0
    13, 2.2,  0.267
    14, 2, 0
    15, 2, 0
    16, 2, 0
    17, 2, 0
    18, 1.7, 0.306
    19, 1.4, 0.327
};
\addlegendentry{BKZ ($\beta=8$)}

\addplot+[
  black!30!red, mark options={black!30!red, scale=0.75},
  error bars/.cd, 
    y fixed,
    y dir=both, 
    y explicit
] table [x=x, y=y,y error=error, col sep=comma] {
    x,  y,       error
    1,  281, 2.683
    2, 159.2, 2.187
    3, 85.8, 1.572
    4, 48.4, 1.466
    5, 27.4, 0.904
    6, 16.1, 0.814
    7, 10.1, 0.629
    8, 6.7, 0.6
    9, 5.2, 0.4
    10, 4.1, 0.554
    11, 3.2, 0.266
    12, 2.4, 0.327
    13, 1.9, 0.2
    14, 1.8, 0.267
    15, 1.5, 0.333
    16, 1.2, 0.26666666326
};
\addlegendentry{Slide+LLL ($\beta=8$)}

\addplot+[
  purple, mark options={purple, scale=0.75},
  error bars/.cd, 
    y fixed,
    y dir=both, 
    y explicit
] table [x=x, y=y,y error=error, col sep=comma] {
    x,  y,       error
    1, 256.2, 1.54387755102
    2, 143.6, 2.27448979592
    3, 83.8, 0.97959183673
    4, 50.1, 0.75714285714
    5, 31.3, 1.11734693878
    6, 20.7,0.6693877551
    7, 13.9, 0.69591836734
    8, 10, 0.51632653061
    9, 7.2, 0.49897959183
    10, 5.6, 0.32653061224
    11, 4.7, 0.30510204081
    12, 4.2, 0.26632653061
    13, 4, 0.29795918367
    14, 3.9, 0.35918367346
    15, 3.2, 0.26632653061
    16, 3.1, .2
    17, 3, 0
    18, 2.3, 0.30510204081
    19, 2.1, .2
    20, 2, 0
    21, 1.9, .2
    22, 1.9, .2
    23, 1.9, .2
    24, 1.9, .2
    25, 1.6, 0.32653061224
    26, 1.4, 0.32653061224
    27, 1.3, 0.30510204081
    28, 1.2, 0.26632653061
};
\addlegendentry{Approximate Griesmer}

\addplot+[
  brown, mark options={brown, scale=0.75},
  error bars/.cd, 
    y fixed,
    y dir=both, 
    y explicit
] table [x=x, y=y,y error=error, col sep=comma] {
    x,  y,       error
    1, 284.9, 3.75306122449
    2, 154.2, 2.45510204082
    3, 83.3, 1.93367346939
    4, 46.5, 1.79489795918
    5, 27.3, .734
    6, 16.5, 1.044
    7, 11, 0.789
    8, 7.2, 0.499
    9, 5.4, 0.327
    10, 4.1, 0.359
    11, 3, 0
    12, 2.8, 0.266
    13, 2.3, 0.305
    14, 2, 0
    15, 2, 0
    16, 1.9, 0.2
    17, 1.8, 0.266
    18, 1.5, 0.334
    19, 1.2, 0.266
    20, 1.1, 0.2
};
\addlegendentry{Full Backward}

  \draw[draw=black, fill=blue, opacity=0.2] (0,6) rectangle ++(20,5);

\end{axis}

 \begin{axis}[
   name=ax2,
   height= 8.6cm,
   width=5cm,
   scaled ticks=false,
   xmin=1,xmax=3,
   ymin=80,ymax=292,
   xtick distance=1,
   xmajorgrids=true,
   grid style={help lines,dashed},
   at={($(ax1.south east)+(2cm,0)$)},
   clip=false
 ]
\addplot+[
  blue, mark options={blue, scale=0.75},
  error bars/.cd, 
    y fixed,
    y dir=both, 
    y explicit
] table [x=x, y=y,y error=error, col sep=comma] {
    x,  y,       error
    1,  292.6, 3.214
    2, 157.1, 2.772
    3, 85, 1.154
    };
    
\addplot+[
  black!30!green, mark options={black!30!green, scale=0.75},
  error bars/.cd, 
    y fixed,
    y dir=both, 
    y explicit
] table [x=x, y=y,y error=error, col sep=comma] {
    x,  y,       error
    1,  274.8, 2.146
    2, 151.7, 2.513
    3, 84.8, 1.572
    };

\addplot+[
  black!30!red, mark options={black!30!red, scale=0.75},
  error bars/.cd, 
    y fixed,
    y dir=both, 
    y explicit
] table [x=x, y=y,y error=error, col sep=comma] {
    x,  y,       error
    1,  281, 2.683
    2, 159.2, 2.187
    3, 85.8, 1.572
    };

\addplot+[
  purple, mark options={purple, scale=0.75},
  error bars/.cd, 
    y fixed,
    y dir=both, 
    y explicit
] table [x=x, y=y,y error=error, col sep=comma] {
    x,  y,       error
    1, 256.2, 1.54387755102
    2, 143.6, 2.27448979592
    3, 83.8, 0.97959183673
    };

\addplot+[
  brown, mark options={brown, scale=0.75},
  error bars/.cd, 
    y fixed,
    y dir=both, 
    y explicit
] table [x=x, y=y,y error=error, col sep=comma] {
    x,  y,       error
    1, 284.9, 3.75306122449
    2, 154.2, 2.45510204082
    3, 83.3, 1.93367346939
    };

\draw[draw=black, fill=blue, opacity=0.2] (0,170) rectangle ++(20,42);
\end{axis}

\begin{axis}[
            axis x line=none,
            axis y line=middle,
            height=8.6cm,
            width=\axisdefaultwidth,
            ymin=250,
            ymax=295,
            axis line style={-},
            minor y tick num=4,
            at={($(ax2.south east)+(-1cm,0)$)}
        ]
            \addplot+[blue, mark options={blue, scale=0.75}] coordinates {
                (0, 292.6)
            };
            \addplot+[black!30!green, mark options={black!30!green, scale=0.75}] coordinates {
                (0, 274.8)
            };
            \addplot+[black!30!red, mark options={black!30!red, scale=0.75}] coordinates {
                (0, 281)
            };
            \addplot+[purple, mark options={purple, scale=0.75}] coordinates {
                (0, 256.2)
            };
            \addplot+[brown, mark options={brown, scale=0.75}] coordinates {
                (0, 284.9)
            };
\end{axis}

\end{tikzpicture}

\begin{tikzpicture}
\begin{axis}[
    xlabel={Index $i$},
    ylabel={$\underline{\ell}_i$},
    xmin=1, xmax=30,
    ymin=0, ymax=300,
    legend pos=north east,
    ymajorgrids=true,
    grid style=dashed,
    ymode=log,
    log basis y={2},
    name=ax1,
]
\addplot+[
  blue, mark options={blue, scale=0.75},
  error bars/.cd, 
    y fixed,
    y dir=both, 
    y explicit
] table [x=x, y=y,y error=error, col sep=comma] {
    x,  y,       error
    1,  292.6, 3.181
    2, 157.1, 2.772
    3, 85, 1.154
    4, 46.8, 1.067
    5, 25.6, 0.611
    6, 15.6, 1.162
    7, 10.1, 0.917
    8, 6.5, 0.683
    9, 4.3, 0.670
    10, 3.1, .554
    11,  2.3, 0.306
    12, 2, 0
    13, 1.5, 0.333
    14, 1.4, 0.327
    15, 1.1, 0.2
};
\addlegendentry{LLL}

\addplot+[
  green, mark options={green, scale=0.75},
  error bars/.cd, 
    y fixed,
    y dir=both, 
    y explicit
] table [x=x, y=y,y error=error, col sep=comma] {
    x,  y,       error
    1,  283, 3.412
    2, 153.5, 3.296
    3, 85.1, 2.867
    4, 48.6, 1.638
    5, 28.6, 0.8
    6, 17.2, 0.980
    7, 11.5, 0.683
    8, 7.5, 0.447
    9, 5, 0.422
    10, 3.9, 0.629
    11, 2.8, 0.267
    12, 2.2, 0.267
    13, 2, 0
    14, 2, 0
    15, 1.7, 0.306
    16, 1.3, 0.306
    17, 1.1, 0.2
};
\addlegendentry{BKZ ($\beta=4$)}

\addplot+[
  black!30!green, mark options={black!30!green, scale=0.75},
  error bars/.cd, 
    y fixed,
    y dir=both, 
    y explicit
] table [x=x, y=y,y error=error, col sep=comma] {
    x,  y,       error
    1,  274.8, 2.146
    2, 151.7, 2.513
    3, 84.8, 1.572
    4, 49.5, 0.537
    5, 29.6, 1.2
    6, 18.1, 0.917
    7, 12.4, 0.533
    8, 8.3, .6
    9, 5.8, 0.499
    10, 4.3, 0.306
    11, 3.4, 0.327
    12, 3, 0
    13, 2.2,  0.267
    14, 2, 0
    15, 2, 0
    16, 2, 0
    17, 2, 0
    18, 1.7, 0.306
    19, 1.4, 0.327
};
\addlegendentry{BKZ ($\beta=8$)}

\addplot+[
  red, mark options={red, scale=0.75},
  error bars/.cd, 
    y fixed,
    y dir=both, 
    y explicit
] table [x=x, y=y,y error=error, col sep=comma] {
    x,  y,       error
    1, 287.2, 4.140
    2, 157.6, 1.692
    3, 83, 1.862
    4, 46.5, 0.856
    5, 28.1, 1.209
    6, 17.5, 0.745
    7, 10.8, 0.4
    8, 6.8, 0.499
    9, 4.6, 0.327
    10, 3.3, 0.306
    11, 2.5, 0.333
    12, 2.2, 0.267
    13, 2, 0.298
    14, 1.5, 0.333
    15, 1.3, 0.306
    16, 1.1, 0.2
};
\addlegendentry{Slide+LLL ($\beta=4$)}

\addplot+[
  black!30!red, mark options={black!30!red, scale=0.75},
  error bars/.cd, 
    y fixed,
    y dir=both, 
    y explicit
] table [x=x, y=y,y error=error, col sep=comma] {
    x,  y,       error
    1,  281, 2.683
    2, 159.2, 2.187
    3, 85.8, 1.572
    4, 48.4, 1.466
    5, 27.4, 0.904
    6, 16.1, 0.814
    7, 10.1, 0.629
    8, 6.7, 0.6
    9, 5.2, 0.4
    10, 4.1, 0.554
    11, 3.2, 0.266
    12, 2.4, 0.327
    13, 1.9, 0.2
    14, 1.8, 0.267
    15, 1.5, 0.333
    16, 1.2, 0.26666666326
};
\addlegendentry{Slide+LLL ($\beta=8$)}

  \draw[draw=black, fill=blue, opacity=0.2] (0,6) rectangle ++(20,5);

\end{axis}

 \begin{axis}[
   name=ax2,
   height= 8.6cm,
   width=5cm,
   scaled ticks=false,
   xmin=1,xmax=3,
   ymin=80,ymax=292,
   xtick distance=1,
   xmajorgrids=true,
   grid style={help lines,dashed},
   at={($(ax1.south east)+(2cm,0)$)},
   clip=false
 ]
\addplot+[
  blue, mark options={blue, scale=0.75},
  error bars/.cd, 
    y fixed,
    y dir=both, 
    y explicit
] table [x=x, y=y,y error=error, col sep=comma] {
    x,  y,       error
    1,  292.6, 3.181
    2, 157.1, 2.772
    3, 85, 1.154
    };

\addplot+[
  green, mark options={green, scale=0.75},
  error bars/.cd, 
    y fixed,
    y dir=both, 
    y explicit
] table [x=x, y=y,y error=error, col sep=comma] {
    x,  y,       error
    1,  283, 3.412
    2, 153.5, 3.296
    3, 85.1, 2.867
    };

\addplot+[
  black!30!green, mark options={black!30!green, scale=0.75},
  error bars/.cd, 
    y fixed,
    y dir=both, 
    y explicit
] table [x=x, y=y,y error=error, col sep=comma] {
    x,  y,       error
    1,  274.8, 2.146
    2, 151.7, 2.513
    3, 84.8, 1.572
    };

\addplot+[
  red, mark options={red, scale=0.75},
  error bars/.cd, 
    y fixed,
    y dir=both, 
    y explicit
] table [x=x, y=y,y error=error, col sep=comma] {
    x,  y,       error
    1, 287.2, 4.140
    2, 157.6, 1.692
    3, 83, 1.862
    };

\addplot+[
  black!30!red, mark options={black!30!red, scale=0.75},
  error bars/.cd, 
    y fixed,
    y dir=both, 
    y explicit
] table [x=x, y=y,y error=error, col sep=comma] {
    x,  y,       error
    1,  281, 2.683
    2, 159.2, 2.187
    3, 85.8, 1.572
    };

\draw[draw=black, fill=blue, opacity=0.2] (0,170) rectangle ++(20,42);
\end{axis}

\begin{axis}[
            axis x line=none,
            axis y line=middle,
            height=8.6cm,
            width=\axisdefaultwidth,
            ymin=250,
            ymax=295,
            axis line style={-},
            minor y tick num=4,
            at={($(ax2.south east)+(-1cm,0)$)}
        ]
            \addplot+[blue, mark options={blue, scale=0.75}] coordinates {
                (0, 292.6)
            };
            
            \addplot+[green, mark options={green, scale=0.75}] coordinates {
                (0, 283)
            };
            \addplot+[black!30!green, mark options={black!30!green, scale=0.75}] coordinates {
                (0, 274.8)
            };
            
            \addplot+[red, mark options={red, scale=0.75}] coordinates {
                (0, 287.2)
            };
            \addplot+[black!30!red, mark options={black!30!red, scale=0.75}] coordinates {
                (0, 281)
            };
\end{axis}

\end{tikzpicture}
\caption{Average sorted output profile for various reduction algorithms run on random bases over 10 iterations  for $n=1280$ and $k=640$ (cryptographic parameters). If we let $\ell_i'$ denote the length of the $i$th longest epipodal vector after basis reduction, $\underline{\ell}_i$ is the average of $\ell_i'$ over 10 iterations. Also shown on the leftmost graphs are error bars representing $\pm 2$ sample standard deviations.}
\label{fig: graph_decay}
\end{figure}

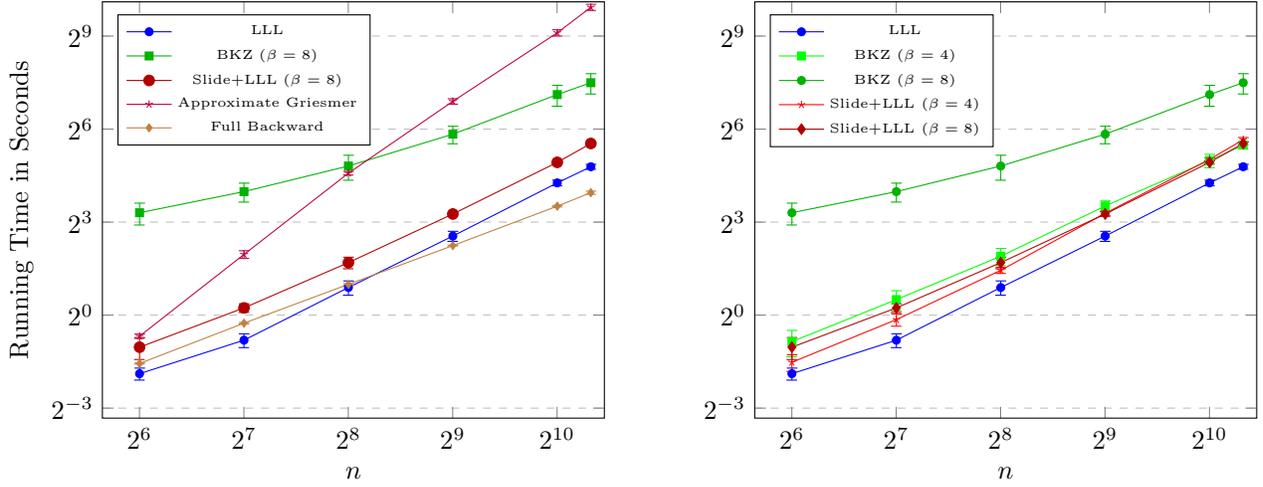
\begin{figure}
\begin{tikzpicture}
\begin{axis}[
    width=0.5 \textwidth,
    xlabel={$n$},
    ylabel={Running Time in Seconds},
    xmin=50, xmax=1400,
    ymin=0.1, ymax=1100,
    legend pos=north west,
    ymajorgrids=true,
    grid style=dashed,
    xmode=log,
    log basis x={2},
    ymode=log,
    log basis y={2},
    name=ax1,
    legend style={font=\tiny}
]

\addplot+[
  blue, mark options={blue, scale=0.75},
  error bars/.cd, 
    y fixed,
    y dir=both, 
    y explicit
] table [x=x, y=y,y error=error, col sep=comma] {
    x,  y,       error
    64,  0.26991240978240966, 0.03601360163
    128, 0.5708382606506348, 0.08748068073
    256, 1.8532333374023438, 0.29354902558
    512, 5.853604507446289, 0.65658109665
    1024, 19.22052412033081, 1.20989317019
    1280, 27.57225368022919, 1.53023538445
};
\addlegendentry{LLL}

\addplot+[
  black!30!green, mark options={black!30!green, scale=0.75},
  error bars/.cd, 
    y fixed,
    y dir=both, 
    y explicit
] table [x=x, y=y,y error=error, col sep=comma] {
    x,  y,       error
    64,  9.85727026462555, 2.36180240804
    128, 15.80422534942627, 3.30317440161
    256, 28.01509294509888, 7.62594305365
    512, 57.09714403152466, 11.0257417945
    1024, 138.37516667842866, 31.7749156332
    1280, 180.43294060230255, 40.3384154707
};
\addlegendentry{BKZ ($\beta=8$)}

\addplot+[
  black!30!red, mark options={black!30!red},
  error bars/.cd, 
    y fixed,
    y dir=both, 
    y explicit
] table [x=x, y=y,y error=error, col sep=comma] {
    x,  y,       error
    64,  0.48792319297790526, 0.11730224311
    128, 1.1748393297195434, 0.11317179875
    256, 3.2336469411849977, 0.4156064305
    512, 9.583624076843261, 0.47412220344
    1024, 30.404942274093628, 1.1299291273
    1280, 46.35443739891052, 1.39524915397
};
\addlegendentry{Slide+LLL ($\beta=8$)}

\addplot+[
  purple, mark options={purple, scale=0.75},
  error bars/.cd, 
    y fixed,
    y dir=both, 
    y explicit
] table [x=x, y=y,y error=error, col sep=comma] {
    x,  y,       error
    64,  0.624, 0.030
    128, 3.883, 0.336
    256, 23.877, 0.957
    512, 118.890, 7.240
    1024, 550.067, 38.5887755102
    1280, 973.778884100914, 66.8704081633
};
\addlegendentry{Approximate Griesmer}

\addplot+[
    brown, mark options={brown, scale=0.75},
    error bars/.cd, 
    y fixed,
    y dir=both, 
    y explicit
] table [x=x, y=y,y error=error, col sep=comma] {
    x,  y,       error
    64,  0.339, 0.004
    128, 0.835, 0.014
    256, 1.990, 0.033
    512,  4.731, 0.045
    1024, 11.414, 0.2
    1280, 15.432, 0.538
};
\addlegendentry{Full Backward}

\end{axis}

\begin{axis}[
    width=0.5 \textwidth,
    xlabel={$n$},
    xmin=50, xmax=1400,
    ymin=0.1, ymax=1100,
    legend pos=north west,
    ymajorgrids=true,
    grid style=dashed,
    xmode=log,
    log basis x={2},
    ymode=log,
    log basis y={2},
    at={($(ax1.south east)+(2cm,0)$)},
    legend style={font=\tiny}
]
\addplot+[
  blue, mark options={blue, scale=0.75},
  error bars/.cd, 
    y fixed,
    y dir=both, 
    y explicit
] table [x=x, y=y,y error=error, col sep=comma] {
    x,  y,       error
    64,  0.26991240978240966, 0.03601360163
    128, 0.5708382606506348, 0.08748068073
    256, 1.8532333374023438, 0.29354902558
    512, 5.853604507446289, 0.65658109665
    1024, 19.22052412033081, 1.20989317019
    1280, 27.57225368022919, 1.53023538445
};
\addlegendentry{LLL}

\addplot+[
  green, mark options={green, scale=0.75},
  error bars/.cd, 
    y fixed,
    y dir=both, 
    y explicit
] table [x=x, y=y,y error=error, col sep=comma] {
    x,  y,       error
    64,  0.5521990299224854, 0.15794969637
    128, 1.4144768953323363, 0.31013923078
    256, 3.7319756507873536, 0.70243936121
    512, 11.461761975288391, 1.36641228292
    1024, 31.706561613082886, 4.67621176117
    1280, 44.69328365325928, 3.61865976889
};
\addlegendentry{BKZ ($\beta=4$)}

\addplot+[
  black!30!green, mark options={black!30!green, scale=0.75},
  error bars/.cd, 
    y fixed,
    y dir=both, 
    y explicit
] table [x=x, y=y,y error=error, col sep=comma] {
    x,  y,       error
    64,  9.85727026462555, 2.36180240804
    128, 15.80422534942627, 3.30317440161
    256, 28.01509294509888, 7.62594305365
    512, 57.09714403152466, 11.0257417945
    1024, 138.37516667842866, 31.7749156332
    1280, 180.43294060230255, 40.3384154707
};
\addlegendentry{BKZ ($\beta=8$)}

\addplot+[
  red, mark options={red, scale=0.75},
  error bars/.cd, 
    y fixed,
    y dir=both, 
    y explicit
] table [x=x, y=y,y error=error, col sep=comma] {
    x,  y,       error
    64,  0.34838132858276366, 0.06425359424
    128, 0.9035999774932861, 0.12336972264
    256, 2.7252809047698974, 0.18862387498
    512, 9.732483792304993, 0.53418582012
    1024, 32.542892527580264, 1.16171831126
    1280, 50.416817498207095, 2.56491173869
};
\addlegendentry{Slide+LLL ($\beta=4$)}

\addplot+[
  black!30!red, mark options={black!30!red},
  error bars/.cd, 
    y fixed,
    y dir=both, 
    y explicit
] table [x=x, y=y,y error=error, col sep=comma] {
    x,  y,       error
    64,  0.48792319297790526, 0.11730224311
    128, 1.1748393297195434, 0.11317179875
    256, 3.2336469411849977, 0.4156064305
    512, 9.583624076843261, 0.47412220344
    1024, 30.404942274093628, 1.1299291273
    1280, 46.35443739891052, 1.39524915397
};
\addlegendentry{Slide+LLL ($\beta=8$)}

\end{axis}
\end{tikzpicture}
\caption{Running time for various basis reduction algorithms run on random bases averaged over 10 iterations. Also shown are error bars representing $\pm 2$ sample standard deviations.}
\label{fig: runtimes_1}
\end{figure}

\begin{figure}
\begin{tikzpicture}
\begin{axis}[
    width=0.5 \textwidth,
    xlabel={$n$},
    ylabel={$k_1$},
    xmin=50, xmax=1400,
    ymin=0, ymax=25,
    legend pos=north west,
    ymajorgrids=true,
    grid style=dashed,
    xmode=log,
    log basis x={2},
    name=ax1,
    legend style={font=\tiny}
]
\addplot+[
  blue, mark options={blue, scale=0.75},
  error bars/.cd, 
    y fixed,
    y dir=both, 
    y explicit
] table [x=x, y=y,y error=error, col sep=comma] {
    x,  y,       error
    64,  7.9, 0.69602043392
    128, 9.1, 0.55377492419
    256, 10.1, 0.69602043392
    512, 11.1, 0.86666666666
    1024, 12.5, 0.85634883857
    1280, 12.7, 0.8969082698
};
\addlegendentry{LLL}

\addplot+[
  black!30!green, mark options={black!30!green, scale=0.75},
  error bars/.cd, 
    y fixed,
    y dir=both, 
    y explicit
] table [x=x, y=y,y error=error, col sep=comma] {
    x,  y,       error
    64,  11.3, 0.52068331172
    128, 13.6, 0.8
    256, 14.4, 0.85374989832
    512, 16, 0.51639777949
    1024, 17, 0.78881063774
    1280, 18.1, 0.35901098714
};
\addlegendentry{BKZ ($\beta=8$)}

\addplot+[
  black!30!red, mark options={black!30!red},
  error bars/.cd, 
    y fixed,
    y dir=both, 
    y explicit
] table [x=x, y=y,y error=error, col sep=comma] {
    x,  y,       error
    64,  8.5, 0.85634883857
    128, 10.7, 0.8969082698
    256, 12, 0.596284794
    512, 13.4, 0.61101009266
    1024, 13.8, 0.65319726474
    1280, 14.4, 0.8
};
\addlegendentry{Slide+LLL ($\beta=8$)}

\addplot+[
  purple, mark options={purple, scale=0.75},
  error bars/.cd, 
    y fixed,
    y dir=both, 
    y explicit
] table [x=x, y=y,y error=error, col sep=comma] {
    x,  y,       error
    64,  9.8, 0.77746025264
    128, 13.1, 0.69602043392
    256, 16, 0.78881063774
    512, 20.3, 0.8969082698
    1024, 23, 0.84327404271
    1280, 24.3, 0.7916228058
};
\addlegendentry{Approximate Griesmer}

\addplot+[
    brown, mark options={brown, scale=0.75},
    error bars/.cd, 
    y fixed,
    y dir=both, 
    y explicit
] table [x=x, y=y,y error=error, col sep=comma] {
    x,  y,       error
    64,  8.9, 0.46666666666
    128, 10.6, 0.53333333333
    256, 13.1, 0.62893207547
    512, 14.7, 0.6
    1024, 16.5, 0.80277297191
    1280,  16.9, 0.62893207547
};
\addlegendentry{Full Backward}

\addplot [
    domain=1:2000, 
    samples=100, 
    color=gray,
    ]
    {log2(x)};

\addplot [
    domain=1:2000, 
    samples=100, 
    color=gray,
    ]
    {2*log2(x)};

\end{axis}

\begin{axis}[
    width=0.5 \textwidth,
    xlabel={$n$},
    xmin=50, xmax=1400,
    ymin=0, ymax=25,
    legend pos=north west,
    ymajorgrids=true,
    grid style=dashed,
    xmode=log,
    log basis x={2},
    at={($(ax1.south east)+(2cm,0)$)},
    legend style={font=\tiny}
]
\addplot+[
  blue, mark options={blue, scale=0.75},
  error bars/.cd, 
    y fixed,
    y dir=both, 
    y explicit
] table [x=x, y=y,y error=error, col sep=comma] {
    x,  y,       error
    64,  7.9, 0.69602043392
    128, 9.1, 0.55377492419
    256, 10.1, 0.69602043392
    512, 11.1, 0.86666666666
    1024, 12.5, 0.85634883857
    1280, 12.7, 0.8969082698
};
\addlegendentry{LLL}

\addplot+[
  green, mark options={green, scale=0.75},
  error bars/.cd, 
    y fixed,
    y dir=both, 
    y explicit
] table [x=x, y=y,y error=error, col sep=comma] {
    x,  y,       error
    64,  9.5, 0.53748384988
    128, 10.7, 0.66999170807
    256, 12.1, 0.81377037438
    512, 12.8, 0.65319726474
    1024, 14.7, 0.84590516936
    1280, 15.1, 0.75718777944
};
\addlegendentry{BKZ ($\beta=4$)}

\addplot+[
  black!30!green, mark options={black!30!green, scale=0.75},
  error bars/.cd, 
    y fixed,
    y dir=both, 
    y explicit
] table [x=x, y=y,y error=error, col sep=comma] {
    x,  y,       error
    64,  11.3, 0.52068331172
    128, 13.6, 0.8
    256, 14.4, 0.85374989832
    512, 16, 0.51639777949
    1024, 17, 0.78881063774
    1280, 18.1, 0.35901098714
};
\addlegendentry{BKZ ($\beta=8$)}

\addplot+[
  red, mark options={red, scale=0.75},
  error bars/.cd, 
    y fixed,
    y dir=both, 
    y explicit
] table [x=x, y=y,y error=error, col sep=comma] {
    x,  y,       error
    64,  8.7, 0.52068331172
    128, 9.7, 0.8969082698
    256, 10.8, 0.65319726474
    512, 12.3, 0.6
    1024, 13.9, 0.81377037438
    1280, 13.8, 0.77746025264
};
\addlegendentry{Slide+LLL ($\beta=4$)}

\addplot+[
  black!30!red, mark options={black!30!red},
  error bars/.cd, 
    y fixed,
    y dir=both, 
    y explicit
] table [x=x, y=y,y error=error, col sep=comma] {
    x,  y,       error
    64,  8.5, 0.85634883857
    128, 10.7, 0.8969082698
    256, 12, 0.596284794
    512, 13.4, 0.61101009266
    1024, 13.8, 0.65319726474
    1280, 14.4, 0.8
};
\addlegendentry{Slide+LLL ($\beta=8$)}

\addplot [
    domain=1:2000, 
    samples=100, 
    color=gray,
    ]
    {log2(x)};

\addplot [
    domain=1:2000, 
    samples=100, 
    color=gray,
    ]
    {2*log2(x)};
    
\end{axis}
\end{tikzpicture}
\caption{Average $k_1$ over 10 iterations for basis reduction algorithms run on random bases. Also shown are error bars representing $\pm 2$ sample standard deviations. Gray lines correspond to $k_1 = \log_2(n)$ and $k_1 = 2 \log_2(n)$.}
\label{fig: average_k1}
\end{figure}
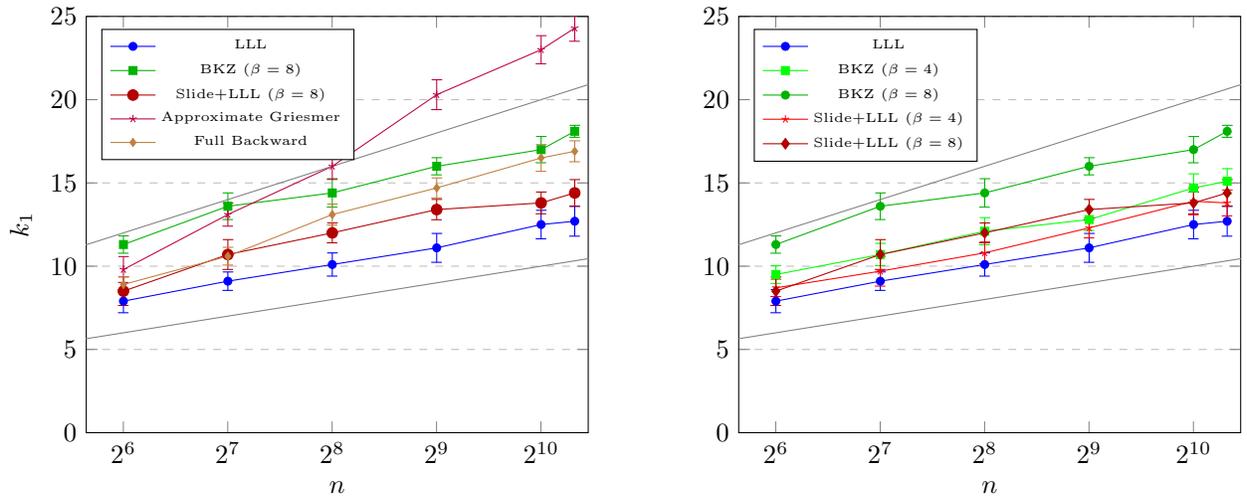

\appendix

\section{A very weak bound on the running time of BKZ}
\label{app:BKZ_running_time}
\begin{claim}
    The while loop in \cref{alg:bkz} runs at most $\beta n^{k-1} (n-k+2) - 1 + k$ times.
\end{claim}

\begin{proof}
    Consider the following potential function
    $$\Phi(\ell_1, \dots, \ell_k) = \sum_{j=1}^{k} n^{k-j} \ell_j$$

    Since $\sum_{i=1}^k \ell_i \leq n$, $\Phi$ is upper bounded by the maximum of $\Phi$ under the constraint that $\sum_{i=1}^k \ell_i \leq n$ and that $\ell_i \geq 0$ for all $i$. Under those constraints, $\Phi$ is maximized if $\ell_1 = n$ and $\ell_i = 0$ for all $i \in [2, k]$. A simple calculation confirms that $\Phi(n-k+1, 1, \dots, 1) = n^{k-1} (n-k+1) + n^{k-2} + n^{k-3} + \dots + 1 = n^{k-1} (n-k+2) - 1$.
    Thus $\Phi$ is upper bounded by $n^{k-1} (n-k+2) - 1$. Notice also that since $\ell_i \geq 1$ for $i \in [1, k]$, $\Phi$ always positive. 
    
    Next, we observe that any time $\vec{B}$ is changed, $\Phi$ decreases by at least one. 
    For all $j \in [1, k]$, let $\ell_j'$ be the value of $\ell_j$ after the while loop is executed and consider the change to $\Phi$ when block $i$ is updated (to ensure $\vec{B}_{[i, i+\beta-1]}$ is forward reduced). We note three important facts relating $\{ \ell_j \}_{j \in [1, k]}$ to $\{ \ell_j' \}_{j \in [1, k]}$ and then use them to show $\Phi(\ell_1, \dots, \ell_k) - \Phi(\ell_1', \dots, \ell_k') \geq 1$.

    \begin{enumerate}
        \item
        $\ell_i - \ell_i' \geq 1$. Follows from the fact that a change was made to $\vec{B}$.

        \item
        $\ell_j = \ell_j'$ for all $j \notin [i, i+\beta-1]$. Follows from \cref{lem: epipodal_locality}.

        \item
        $\sum_{j=i+1}^k \ell_j' \leq n-1$. Follows from the fact that $\sum_{j=1}^k \ell_j' \leq n$ and $\ell_i' \geq 1$.
    \end{enumerate}
    \begin{align*}
        \Phi(\ell_1, \dots, \ell_k) - \Phi(\ell_1', \dots, \ell_k')
        &= \sum_{j=1}^k n^{k-j} \ell_j - \sum_{j=1}^k n^{k-j} \ell_j'\\
        &= \sum_{j=1}^{i-1} n^{k-j} (\ell_j - \ell_j') + n^{k-i} (\ell_i - \ell_i') + \sum_{j=i+1}^{k} n^{k-j} (\ell_j - \ell_j')\\
        &\geq 0 + n^{k-i} + \sum_{j=i+1}^{k} n^{k-j} (\ell_j - \ell_j')\\
        &\geq n^{k-i} - \sum_{j = i+1}^k n^{k-j} \ell_j'\\
        &\geq n^{k-i} - n^{k-i-1} (n-1)\\
        &\geq 1
    \end{align*}
     The second to last inequality follows from the fact that $\sum_{j = i+1}^k n^{k-j} \ell_j'$ is upper bounded by the maximum of $\sum_{j = i+1}^k n^{k-j} \ell_j'$ under the constraints that $\sum_{j=i+1}^k \ell_j' \leq n-1$ and $\ell_j' \geq 0$ for all $j \in [i+1, k]$. A simple calculation shows that the maximum is $n^{k-i-1} (n-1)$. Since $\Phi$ is upper bounded by $n^{k-1} (n-k+1) - 1$, lower bounded by $0$, and decreases by one each time $\vec{B}$ is changed, $\vec{B}$ changes at most $n^{k-1} (n-k+1) - 1$ times.

     Finally, we show that \cref{alg:bkz} terminates in $\beta n^k + k$ iterations. Let $u$ (for up) denote the number of times the algorithm increments $i$ and $d$ denote the number of times the algorithm does not increment $i$ (possibly leaving it unchanged). 
     Notice that the algorithm runs for $u+d$ iterations since each iteration either results in $i$ being incremented or not being incremented. 
     We must have that $u-(\beta-1)d \leq k$, since the algorithm terminates as soon as $i = k$, but the final index $i$ is clearly at least $u-(\beta-1)d$ (since $i$ is incremented $u$ times and decreases by at most $(\beta-1)$ a total of $d$ times). And, we must have $d \leq n^{k-1} (n-k+2) - 1$ since $i$ is only not incremented when $\vec{B}$ is updated, which we showed happens at most $n^{k-1} (n-k+2) - 1$ times. It follows that
     \[
        u+d = \beta d + u - (\beta-1)d \leq  \beta d + k\leq \beta n^{k-1} (n-k+2) - 1 + k
        \; .
     \]
     So, \cref{alg:bkz} terminates after at most $\beta n^{k-1} (n-k+2) - 1 + k$ iterations, as claimed.
\end{proof}

\section{Subprocedures for one- and two-dimensional codes}
\label{subsec: log_q}

Here, we describe very simple $O(n \log^2 (q))$ algorithms that solve two simple problems that arise in our size-reduction algorithm and in LLL. In one, we wish to find $a \in \F_q$ minimizing $\vec{x}_1 + a \vec{x}_2$ (and worry about tie breaking), which one can think of as the nearest codeword problem in a one-dimensional code. In the other, we wish to find a non-zero shortest codeword in the two-dimensional code generated by $\vec{x}_1$ and $\vec{x}_2$.

\begin{definition}
    The problem Size Reduction Coefficient is defined as follows. Given $\vec{x}_1 = (x_{1, 1}, \dots, x_{1, n}) \in \F_q^n$, $\vec{x}_2 = (x_{2, 1}, \dots, x_{2, n}) \in \F_q^n$, find $a \in \F_q$ such that $|\vec{x}_1 + a \vec{x}_2|+\text{TB}_{\vec{x}_2}(\vec{x}_1 + a \vec{x}_2)$ is minimized.
\end{definition}

Notice that the $i$th coordinate of $\vec{y} = \vec{x}_1 + a \vec{x}_2$ is zero if and only if $a = -x_{1, i} x_{2, i}^{-1}$. It is also not hard to see that if $|\vec{x}_1 + a \vec{x}_2|+\text{TB}_{\vec{x}_2}(\vec{x}_1 + a \vec{x}_2)$ is minimized, then $\vec{x}_1 + a \vec{x}_2$ contains at least one zero, and thus $a = -x_{1, i} x_{2, i}^{-1}$ for some $i \in [1, n]$. Therefore finding $a$ is simply a matter of enumerating over all $i \in [1, n]$ to find $a = -x_{1, i} x_{2, i}^{-1}$ minimizing $|\vec{x}_1 + a \vec{x}_2|+\text{TB}_{\vec{x}_2}(\vec{x}_1 + a \vec{x}_2)$. The whole processes requires time $O(n \log^2(q))$.

\begin{definition}
    The problem Shortest Codeword Coefficients is defined as follows. Given $\vec{x}_1 \in \F_q^n$, $\vec{x}_2 \in \F_q^n$, find $a_1 \in \F_q$ and $a_2 \in \F_q$ such that $(a_1, a_2) \neq (0, 0)$ and $|a_1 \vec{x}_1 + a_2 \vec{x}_2|$ is minimized.
\end{definition}

There are two cases to consider. If $a_1 \neq 0$, then without loss of generality, we can assume $a_1 = 1$ since $|a_1 \vec{x}_1 + a_2 \vec{x}_2| = |a_1^{-1} (a_1 \vec{x}_1 + a_2 \vec{x}_2)| = |\vec{x}_1 + a_1^{-1} a_2 \vec{x}_2|$. If $a_1 = 0$, then for all $a_2 \neq 0$, $|a_1 \vec{x}_1 + a_2 \vec{x}_2|$ is the same. Thus to solve Shortest Codeword Coefficients, we solve Size Reduction Coefficients to find $a$ and return $(1, a)$ if $|\vec{x}_1 + a \vec{x}_2| < |\vec{x}_2|$ and $(0, 1)$ otherwise. The whole process requires only time $O(n \log^2(q))$


\begin{thebibliography}{DDvW22}

\bibitem[ABSS97]{aroraHardnessApproximateOptima1997}
Sanjeev Arora, L{\'a}szl{\'o} Babai, Jacques Stern, and Z.~Sweedyk.
\newblock The hardness of approximate optima in lattices, codes, and systems of linear equations.
\newblock {\em J. Comput. Syst. Sci.}, 54(2):317--331, 1997.

\bibitem[AD97]{ajtaiPublicKeyCryptosystemWorstCase1997}
Mikl{\'o}s Ajtai and Cynthia Dwork.
\newblock A public-key cryptosystem with worst-case/average-case equivalence.
\newblock In {\em {{STOC}}}, 1997.

\bibitem[ADRS15]{conf/stoc/AggarwalDRS15}
Divesh Aggarwal, Daniel Dadush, Oded Regev, and Noah Stephens{-}Davidowitz.
\newblock Solving the {Shortest Vector Problem} in $2^n$ time using discrete {Gaussian} sampling.
\newblock In {\em STOC}, 2015.

\bibitem[Ajt96]{ajtaiGeneratingHardInstances1996}
Mikl{\'o}s Ajtai.
\newblock Generating hard instances of lattice problems.
\newblock In {\em {{STOC}}}, 1996.

\bibitem[Ajt98]{ajtaiShortestVectorProblem1998}
Mikl{\'o}s Ajtai.
\newblock The {{Shortest Vector Problem}} in {{L2}} is {{NP-hard}} for randomized reductions.
\newblock In {\em {{STOC}}}, 1998.

\bibitem[AKS01]{ajtaiSieveAlgorithmShortest2001}
Mikl{\'o}s Ajtai, Ravi Kumar, and D.~Sivakumar.
\newblock A sieve algorithm for the {{Shortest Lattice Vector Problem}}.
\newblock In {\em {{STOC}}}, pages 601--610, 2001.

\bibitem[Ale03]{alekhnovichMoreAverageCase2003}
Michael Alekhnovich.
\newblock More on average case vs approximation complexity.
\newblock In {\em {{FOCS}}}, pages 298--307, 2003.

\bibitem[ALL19]{DecodingChallenge}
Nicolas Aragon, Julien Lavauzelle, and Matthieu Lequesne.
\newblock decodingchallenge.org, 2019.

\bibitem[ALNS20]{aggarwalSlideReductionRevisited2020}
Divesh Aggarwal, Jianwei Li, Phong~Q. Nguyen, and Noah {Stephens-Davidowitz}.
\newblock Slide reduction, revisited---{{Filling}} the gaps in {{SVP}} approximation.
\newblock In {\em {{CRYPTO}}}, 2020.

\bibitem[AS18]{ASGapETHHardness2018}
Divesh Aggarwal and Noah {Stephens-Davidowitz}.
\newblock ({{Gap}}/{{S}}){{ETH}} hardness of {{SVP}}.
\newblock In {\em {{STOC}}}, 2018.

\bibitem[Bab86]{babaiLovaszLatticeReduction1986}
L.~Babai.
\newblock On {{Lov{\'a}sz}}' lattice reduction and the nearest lattice point problem.
\newblock {\em Combinatorica}, 6(1):1--13, 1986.

\bibitem[BDGL16]{BDGLNewDirectionsNearest2016}
Anja Becker, L{\'e}o Ducas, Nicolas Gama, and Thijs Laarhoven.
\newblock New directions in nearest neighbor searching with applications to lattice sieving.
\newblock In {\em {{SODA}}}, 2016.

\bibitem[BGS17]{bennettQuantitativeHardnessCVP2017}
Huck Bennett, Alexander Golovnev, and Noah {Stephens-Davidowitz}.
\newblock On the quantitative hardness of {{CVP}}.
\newblock In {\em {{FOCS}}}, 2017.

\bibitem[BMvT78]{berlekampInherentIntractabilityCertain1978}
E.~Berlekamp, R.~McEliece, and H.~van Tilborg.
\newblock On the inherent intractability of certain coding problems.
\newblock {\em IEEE Transactions on Information Theory}, 24(3):384--386, 1978.

\bibitem[Boa81]{boasAnotherNPcompleteProblem1981}
Peter van~Emde Boas.
\newblock Another {{NP-complete}} problem and the complexity of computing short vectors in a lattice.
\newblock Technical report, University of Amsterdam, 1981.

\bibitem[BSW18]{BSWMeasuringSimulatingExploiting2018}
Shi Bai, Damien Stehl{\'e}, and Weiqiang Wen.
\newblock Measuring, simulating and exploiting the head concavity phenomenon in {{BKZ}}.
\newblock In {\em Asiacrypt}, 2018.

\bibitem[CN11]{CNBKZBetterLattice2011}
Yuanmi Chen and Phong~Q. Nguyen.
\newblock {{BKZ}} 2.0: {{Better}} lattice security estimates.
\newblock In {\em Asiacrypt}, 2011.

\bibitem[DDvW22]{DDvAlgorithmicReductionTheory2022}
Thomas {Debris-Alazard}, L{\'e}o Ducas, and Wessel P.~J. van Woerden.
\newblock An algorithmic reduction theory for binary codes: {{LLL}} and more.
\newblock {\em IEEE Transactions on Information Theory}, 68(5):3426--3444, 2022.
\newblock \url{https://eprint.iacr.org/2020/869}.

\bibitem[DKRS03]{DKRSApproximatingCVPAlmostpolynomial2003}
Irit Dinur, Guy Kindler, Ran Raz, and Shmuel Safra.
\newblock Approximating {CVP} to within almost-polynomial factors is {NP}-hard.
\newblock {\em Combinatorica}, 23(2):205--243, 2003.

\bibitem[DMS03]{dumerHardnessApproximatingMinimum2003}
I.~Dumer, D.~Micciancio, and M.~Sudan.
\newblock Hardness of approximating the minimum distance of a linear code.
\newblock {\em IEEE Transactions on Information Theory}, 49(1):22--37, 2003.

\bibitem[GN08]{gamaFindingShortLattice2008}
Nicolas Gama and Phong~Q. Nguyen.
\newblock Finding short lattice vectors within {Mordell's} inequality.
\newblock In {\em {{STOC}}}, 2008.

\bibitem[Gri60]{GriBoundErrorcorrectingCodes1960}
J.~H. Griesmer.
\newblock A bound for error-correcting codes.
\newblock {\em IBM Journal of Research and Development}, 4(5):532--542, 1960.

\bibitem[Hoe63]{hoeffdingProbabilityInequalitiesSums1963}
W.~Hoeffding.
\newblock Probability inequalities for sums of bounded random variables.
\newblock {\em Journal of the American Statistical Association}, 58:13--30, 1963.

\bibitem[HPS98]{hoffsteinNTRURingbasedPublic1998}
Jeffrey Hoffstein, Jill Pipher, and Joseph~H. Silverman.
\newblock {NTRU}: A ring-based public key cryptosystem.
\newblock In {\em {{ANTS}}}, pages 267--288, 1998.

\bibitem[LB88]{LBObservationSecurityMcEliece1988}
P.~J. Lee and E.~F. Brickell.
\newblock An observation on the security of {{McEliece}}'s public-key cryptosystem.
\newblock In {\em Eurocrypt}, 1988.

\bibitem[LLL82]{lll82}
Arjen~K. Lenstra, Hendrik~W. Lenstra, Jr., and L\'{a}szl\'{o} Lov\'{a}sz.
\newblock Factoring polynomials with rational coefficients.
\newblock {\em Mathematische Annalen}, 261(4):515--534, December 1982.

\bibitem[McE78]{mceliecePublickeyCryptosystemBased1978}
Robert~J. McEliece.
\newblock A public-key cryptosystem based on algebraic coding theory.
\newblock {{DSN Progress Report}}, Jet Propulsion Laboratory, 1978.

\bibitem[Mic01]{Mic01svp}
Daniele Micciancio.
\newblock The {{Shortest Vector Problem}} is {{NP-hard}} to approximate to within some constant.
\newblock {\em SIAM Journal on Computing}, 30(6):2008--2035, 2001.

\bibitem[MW16]{micciancioPracticalPredictableLattice2016}
Daniele Micciancio and Michael Walter.
\newblock Practical, predictable lattice basis reduction.
\newblock In {\em Eurocrypt}, 2016.

\bibitem[NV10]{nguyenLLLAlgorithmSurvey2010}
Phong~Q. Nguyen and Brigitte Vallée, editors.
\newblock {\em The {LLL} Algorithm: {Survey} and Applications}.
\newblock Springer-Verlag, 2010.

\bibitem[Reg09]{regevLatticesLearningErrors2009}
Oded Regev.
\newblock On lattices, learning with errors, random linear codes, and cryptography.
\newblock {\em J. ACM}, 56(6):Art. 34, 40, 2009.

\bibitem[Rie16]{1995252}
Marko Riedel.
\newblock How many 10-digit number such no digit occurs 4 or more times? (leading 0 not allowed).
\newblock Mathematics Stack Exchange, 2016.
\newblock \url{https://math.stackexchange.com/q/1995252}.

\bibitem[Sch87]{schnorrHierarchyPolynomialTime1987}
Claus-Peter Schnorr.
\newblock A hierarchy of polynomial time lattice basis reduction algorithms.
\newblock {\em Theor. Comput. Sci.}, 53(23):201--224, 1987.

\bibitem[SV19]{stephens-davidowitzSETHhardnessCodingProblems2019}
Noah {Stephens-Davidowitz} and Vinod Vaikuntanathan.
\newblock {{SETH-hardness}} of coding problems.
\newblock In {\em {{FOCS}}}, 2019.

\bibitem[Var97]{VarAlgorithmicComplexityCoding1997}
Alexander Vardy.
\newblock Algorithmic complexity in coding theory and the {{Minimum Distance Problem}}.
\newblock In {\em {{STOC}}}, 1997.

\bibitem[Wal20]{walterLatticeBlogReduction}
Michael Walter.
\newblock Lattice blog reduction: The {{Simons Institute}} blog.
\newblock \url{https://blog.simons.berkeley.edu/2020/04/lattice-blog-reduction-part-i-bkz/}, 2020.

\end{thebibliography}
\end{document}